\documentclass[letterpaper,USenglish]{lipics}

\usepackage{algorithm}
\usepackage{algpseudocode}
\usepackage{cite}
\usepackage{microtype}
\usepackage{xspace}
\usepackage{color}
\usepackage{verbatim}

\title{Aggregations over Generalized Hypertree Decompositions}
\author[1]{Manas Joglekar}
\author[1]{Rohan Puttagunta}
\author[1]{Chris R\'e}
\affil[1]{Department of Computer Science\\
  Stanford University\\
  \texttt{\{manasrj, rohanp, chrismre\} @ cs.stanford.edu}}

\newtheorem{proposition}{Proposition}
\newtheorem{theorem*}{Theorem}
\newtheorem{lemma*}{Lemma}
\newtheorem{corollary*}{Corollary}
\newtheorem{definition*}{Definition}
\newtheorem{example*}{Example}
\newcommand{\mD}{\mathcal{D}}
\newcommand{\mE}{\mathcal{E}}
\newcommand{\mH}{\mathcal{H}}
\newcommand{\mT}{\mathcal{T}}
\newcommand{\mV}{\mathcal{V}}
\newcommand{\Ot}{\widetilde{O}}
\newcommand{\IN}{\textsc{IN}}
\newcommand{\OUT}{\textsc{OUT}}
\newcommand{\pre}{\textsf{PREC}}
\newcommand{\dnc}{\textsf{DNC}}
\newcommand{\pa}{\textsf{PA}}
\newcommand{\ajar}{\textsc{Ajar}\xspace}

\newcommand{\main}[1]{}

\begin{document}

\maketitle

\abstract
We study a class of aggregate-join queries with multiple aggregation operators evaluated over annotated relations. We show that straightforward extensions of standard multiway join algorithms and generalized hypertree decompositions (GHDs) provide best-known runtime guarantees. In contrast, prior work uses bespoke algorithms and data structures and does not match these guarantees. Our extensions to the standard techniques are a pair of simple tests that (1) determine if two orderings of aggregation operators are equivalent and (2) determine if a GHD is compatible with a given ordering. These tests provide a means to find an optimal GHD that, when provided to standard join algorithms, will correctly answer a given aggregate-join query. The second class of our contributions is a pair of complete characterizations of (1) the set of orderings equivalent to a given ordering and (2) the set of GHDs compatible with some equivalent ordering. We show by example that previous approaches are incomplete. The key technical consequence of our characterizations is a decomposition of a compatible GHD into a set of (smaller) {\em unconstrained} GHDs, i.e. into a set of GHDs of sub-queries without aggregations. Since this decomposition is comprised of unconstrained GHDs, we are able to connect to the wide literature on GHDs for join query processing, thereby obtaining improved runtime bounds, MapReduce variants, and an efficient method to find approximately optimal GHDs.

\section{Introduction}
Generalized hypertree decompositions (GHDs), introduced by Gottlob et al.~\cite{GHDIntro, gottlob:ght} and further developed by Grohe and Marx~\cite{FHTW}, provide a means for performing early projection in join processing, which can result in dramatically faster runtimes. In this work, we extend GHDs to handle queries that include aggregations, which allows us to capture both SQL-aggregate processing and message passing problems. Motivated by our own database engine based on GHDs~\cite{Duncecap1, Duncecap2, EH}, we seek to more deeply understand the space of optimization for aggregate-join queries.

We build upon work by Green, Karvounarakis, and Tannen~\cite{Semiring} on annotated relations to define our notion of aggregation. These annotations provide a general definition of aggregation, allowing us to represent a wide-ranging set of problems as aggregate-join queries. Our queries, which we call $\ajar$ (Aggregations and Joins over Annotated Relations) queries, contain semiring quantifiers that ``sum over'' or ``marginalize out'' values. We formally define $\ajar$ queries in Section~\ref{sec:prelim}, but they are easy to illustrate by example:

\begin{figure}
\main{\begin{minipage}{.2\textwidth}}
\begin{tabular}[t]{|c|c||c}
\multicolumn{3}{c}{$R$} \\
A & B & $\mathbb{Z}$ \\ \hline
1 & 3 & 3 \\
1 & 2 & 1 \\
1 & 1 & 2 
\end{tabular}
\main{\\}
\begin{tabular}[t]{|c|c||c}
\multicolumn{3}{c}{$S$} \\
B & C & $\mathbb{Z}$ \\ \hline
1 & 1 & 4 \\
3 & 3 & 6
\end{tabular}
\main{\end{minipage}}
\hfill
\main{\begin{minipage}{.2\textwidth}}
\begin{tabular}[t]{|c|c|c||c}
\multicolumn{4}{c}{$R \Join S$} \\
A & B & C & $\mathbb{Z}$ \\ \hline
1 & 3 & 3 & 18\\
1 & 1 & 1 & 8
\end{tabular}
\main{\\}
\begin{tabular}[t]{|c||c}
\multicolumn{2}{c}{$\sum_C \sum_B R \Join S$} \\
A & $\mathbb{Z}$ \\ \hline
1 & 26
\end{tabular}
\main{\end{minipage}}
\caption{Illustrating the computation of Example~\ref{ex:firstex}}
\label{fig:firstex}
\end{figure}

\begin{example}
\label{ex:firstex}
Consider two relations with attributes $\{A,B\}$ and $\{B,C\}$ such that each tuple is annotated with some integer; we call these relations $\mathbb{Z}$-relations. Consider the query:
$$\sum_C \sum_B R(A,B) \Join S(B,C)$$
Our output will then be a $\mathbb{Z}$-relation with attribute set $\{A\}$. Each value $a$ of attribute $A$ in $R$ is associated with a set $X_a$ of pairs $(b,z_R)$ composed of a value $b$ of attribute $B$ and an annotation $z_R$. Furthermore for each $b$ value in $X_a$, there is a set $X_b$ from relation $S$ of pairs $(c ,z_S)$ composed of a value $c$ of attribute $C$ and an annotation $z_S$. Given $X_a$ and each $X_b$ associated with a given value $a$, the annotation associated with $a$ in our output will simply be $$\sum_{b,z_1 \in X_a} \sum_{c, z_2 \in X_b} z_1 * z_2.$$
\end{example}

$\ajar$ queries capture both classical SQL-style queries and newer data processing problems like probabilistic inference via message passing on graphical models~\cite{Kask:2005:UTD:1090725.1090730}. In fact, Aji and McEliece proposed the ``Marginalize a Product Function'' (MPF) problem~\cite{Aji}, which is a special case of an $\ajar$ query, and showed how the problem and its solution capture a number of classic problems and algorithms, including fast Hadamard transforms, Viterbi's algorithm, forward-backward algorithm, FFT, and probabilistic inference in Bayesian networks. These algorithm are fundamental to various fields; for example the forward-backward algorithm over conditional random fields forms the basis for state of the art solutions to named entity recognition, part of speech tagging, noun phrase segmentation, and other problems in NLP~\cite{CRF}. We are motivated by the wide applicability of queries over annotated relations; annotated relations may provide a framework for combining classical query processing, linear algebra, and statistical inference in a single data processing system.

We consider a generalization of MPF with multiple aggregation operators.  We represent an aggregate-join query as a join $Q$ and an aggregation ordering, which specifies both the ordering and the aggregation of each attribute. Our language directly follows from the work of Abo Khamis, Ngo, and Rudra~\cite{FAQ}, who investigated the ``Functional Aggregate Query'' (FAQ) problem. In addition to MPF, FAQ is a generalization of Chen and Dalmau's QCQ problem~\cite{QCQChenDalmau}, in which the only aggregates are logical quantifiers (AND and OR).

The key technical challenge in both problems is characterizing the permissible aggregations orders to answer the query. Chen and Dalmau give a complete characterization of which variable orders are permissible for QCQ via a procedure. We first give a simple (complete) procedure for our more general class of queries with multiple aggregations, and then we provide a complete characterization of permissible orders.

\begin{itemize}
  \item {\em A Simple Test for Equivalence: }
A query can be thought of as a body $Q$ and a string of attribute-operator pairs $\alpha$. Given a query $Q$ and two orders $\alpha$ and $\beta$, we provide a simple test to determine whether $\alpha$ and $\beta$ are equivalent (i.e., return the same output for any input database). The technical challenge is that different aggregation operators (e.g., $\sum$ and $\max$) cannot freely commute. We show that attribute-operator pairs can commute for only two reasons: $(1)$ their operators commute or $(2)$ their attributes are ``independent'' in the query, e.g., in the query $\min_{B} \max_{A} \sum_{C} R(A,B),S(B,C)$ the aggregations involving $A$ and $C$ can commute -- even though $\max$ and $\sum$ do not commute as operators, the query body renders them {\it independent} given $B$. We show that these two conditions are {\em complete}, which leads to a simple test for equivalence (Algorithm~\ref{algo:equivalance-test}).

\item {\em A Simple Test for GHD and Order Compatibility: } We say a GHD is compatible with an ordering if we can run standard join algorithms on the GHD while performing aggregations in the order given by the ordering. We show that testing for compatibility amounts to verifying that for any two attributes $A$, $B$, if the topmost GHD node containing $A$ occurs above the topmost node containing $B$, then $A$ occurs before $B$ in the ordering.  
\end{itemize}

This pair of results gives us a simple algorithm that achieves the best known runtime results. Given a query $(Q,\alpha)$, enumerate each order $\beta$ and each GHD $G$, checking if $\alpha$ is equivalent to $\beta$ and $G$ is compatible with $\beta$. If so, record the cost of solving the query using $G$, according to (say) fractional hypertreewidth. Solve the query using the lowest cost $(G,\beta)$ with a standard join algorithm~\cite{FHTW}.\footnote{Two technical notes: (1) methods like submodular width~\cite{Marx:2010:THP:1806689.1806790} or Joglekar and R\'e~\cite{2015arXiv150801239J} require that we first partition the instances and then run the above algorithm; (2) FAQ~\cite{FAQ} is not output sensitive (it does not use GHDs), and so it handles output attributes less efficiently than the above algorithm, as seen in Example~\ref{ex:faqoutput}.}

The preceding simple algorithm runs in time exponential in the query size. But finding the optimal GHD even without aggregation is an $\mathsf{NP}$-hard problem, so the brute force optimizer has essentially optimal runtime. It is easy to implement, and a variant is in our prototype database~\cite{EH,Duncecap1}.

The more interesting problem is to characterize the notions of equivalence, mirroring Chen and Dalmau. To that end, we give two new, complete characterizations:

\begin{itemize}
\item {\em A Complete Characterization of Equivalent Orders:} Given an order $\alpha$ and two attribute-operator pairs $x,y \in \alpha$, we describe a set of constraints of the form {\it ``in any order, $x$ must appear after $y$.''} Our constraints are sound and complete, i.e., a string $\beta$ satisfies these constraints if and only if it is equivalent to $\alpha$. In contrast, previous approaches have an incomplete characterization, as shown in Example~\ref{example:faq-incompleteness} in the Appendix.

\item {\em A Complete Characterization of GHDs compatible with any Equivalent Order.} Given an order $\alpha$ and a query hypergraph $Q$, we call a GHD `valid' if it is compatible with any ordering equivalent to $\alpha$. We give a succinct characterization for all valid GHDs. We then describe a decomposition of the query $(Q, \alpha)$ into a series of \emph{characteristic hypergraphs} (without attached aggregation orderings). GHDs for these hypergraphs can be combined into a valid GHD for the original query. We show that for any ``node-monotone''
\footnote{Informally, a map is {\em node monotone} if adding more {\em nodes} to a graph does not reduce the measure, but additional edges may reduce the measure, see Definition~\ref{def:node-monotone}.}
width function, there is a GHD with optimal width $w$ that can be constructed with this decomposition.
\footnote{In contrast, FAQ's decomposition strategy may miss the optimal GHD. Appendix Example~\ref{example:faq-2x} shows a case in which using the FAQ decomposition gives a width $2n$ while AJAR obtains width $n$ for $n\geq 1$. We also exhibit a family of queries and instances on which FAQ runs in time $\Omega(N^{3n/2})$ while AJAR runs in time $O(N^{n})$ for $n\geq 1$, see Appendix Example~\ref{example:faq-decomposition}.}
Treewidth, Fractional hypertreewidth, and Submodular width are all node-monotone.
\end{itemize}

Conceptually, we think the latter result is especially important for tying our work to existing GHD literature; the result reduces our problem to operating on standard GHDs. Pragmatically, we can apply existing GHD results to our characteristic hypergraphs and obtain the following results for free:

\begin{itemize}
  \item Based on Grohe and Marx~\cite{AGM}, we are able to describe our runtime in terms of classical metrics like fractional hypertreewidth. In turn, we can use standard notions to upper bound the runtime like fractional hypertree width, Marx's submodular width~\cite{Marx:2010:THP:1806689.1806790}, or Joglekar's efficiently computable variant~\cite{2015arXiv150801239J}. 
    
\item Based on Afrati et al.~\cite{GYM}, who bound the communication costs of join processing in terms of a ``width'' parameter for GHDs, we can develop efficient MapReduce algorithms for solving \ajar queries.

\item Based on Marx's approximation~\cite{Marx:2010:AFH:1721837.1721845} for GHDs, we can find approximately optimal GHDs for the popular fractional hypertreewidth measure in polynomial time.
\end{itemize}

We get the above results essentially for free from forging this connection to GHDs. We view this simple link as a strength of our approach.

Finally, we discuss an extension to handle ``product aggregations'' that allows us to aggregate away an attribute \emph{before} we join the relations containing the attribute when the aggregation operator is the multiplication operator of the semiring. FAQ was the first to observe that this special case can improve certain types of logical queries. This opens up a new space of equivalent orderings and valid GHDs; mirroring the above results, we give a simple test and a complete characterization of the valid GHDs for queries that include this aggregation. As a result, we obtain similar improvements in runtime relative to previous work.

\textbf{Outline.} We discuss related work in Section~\ref{sec:rel}. In Section~\ref{sec:prelim}, we introduce notation and algorithms that are relevant to our work before defining the $\ajar$ problem and discussing its solution, which involves running existing algorithms on a restricted class of GHDs. Section~\ref{sec:equivalent-orderings} provides a succinct characterization of all orderings that are equivalent to a given ordering. Section~\ref{sec:decomposing} discusses how to connect our work with recent research on GHDs, explaining how to construct valid optimal query plans and how to further improve and parallelize our results. In Section~\ref{sec:univ-aggregation}, we discuss how to incorporate product aggregations. 

\section{Related Work}
\label{sec:rel}
\textbf{Join Algorithms.} The Yannakakis algorithm, introduced in 1981, guarantees a runtime of $O(\IN + \OUT)$ for $\alpha$-acyclic join queries~\cite{Yannakakis81}. Modern multiway algorithms can process any join query and have worst-case optimal runtime. In particular, Atserias, Grohe, and Marx~\cite{AGM} derived a tight bound on the worst-case size of a join query given the input size and structure. Ngo et al.~\cite{NPRR} presented the first algorithm to achieve this runtime bound, i.e. the first worst-case optimal algorithm. Soon after, Veldhuizen presented Leapfrog Triejoin, a very simple worst-case optimal algorithm that had been implemented in LogicBlox's commercial database system~\cite{LFTJ}. Ngo et al.~\cite{NRR} later presented the simplified and unified algorithm GenericJoin (GJ) that captured both of the previous worst-case optimal algorithms.

\textbf{GHDs.} First introduced by Gottlob, Leone, and Scarcello~\cite{GHDIntro}, hypertree decompositions and the associated hypertree width generalize the concept of tree decompositions~\cite{treedecomp}. Conceptually, the decompositions capture a hypergraph's cyclicity, allowing them to facilitate the selective use of GJ and Yannakakis in the standard hybrid join algorithm GHDJoin. There are deep connections between variable orderings and GHDs~\cite{FAQ}, which we leverage extensively.  Grohe and Marx~\cite{FHTW} introduced the idea of fractional hypertree width over GHDs, which bounds the runtime of GHDJoin by $\Ot(IN^w + OUT)$ ($\Ot$ hides poly-logarithmic factors) for $w$ defined to be the minimum fractional hypertree width among all GHDs.

\textbf{Semirings and Aggregations.}  Green, Karvounarakis, and Tannen developed the idea of annotations over a semiring~\cite{Semiring}. Our notation for the annotations is superficially different from theirs, solely for notational convenience. We delve into more detail in Section~\ref{sec:prelim}. This also has been used as a mechanism to capture aggregation in probabilistic databases~\cite{DBLP:journals/vldb/ReS09}.

\textbf{MPF.} Aji and McEliece~\cite{Aji} defined the ``Marginalize a Product Function'' (MPF) problem, which is equivalent to the the space of $\ajar$ queries with only one aggregation operator. They showed that MPF generalizes a wide variety of important algorithms and problems, which also implies that $\ajar$ queries are remarkably general. They also provided a message passing algorithm to solve MPF, which has since been refined~\cite{Kask:2005:UTD:1090725.1090730}. We provide runtime guarantees that improve the current state of the art.

\textbf{Aggregate-Join Queries.} There is a standard modification to Yannakakis to handle aggregations~\cite{Yannakakis81}, but the classic analysis provides only a $O(\IN \cdot \OUT)$ bound. Bakibayev, Kocisky, Olteanu, and Zavodny study aggregation-join queries in factorized databases~\cite{OZVLDB13}, and later Olteanu and Zavodny connected factorized databases and GHDs/GHDJoin~\cite{OZTODS15}. They develop the intuition that if output attributes are above non-output attributes, the $+\OUT$ runtime is preserved; we use the same intuition to develop and analyze AggroGHDJoin, a variant to GHDJoin for aggregate-join queries.

Abo Khamis, Ngo, and Rudra present the ``Functional Aggregate Query'' (FAQ) problem~\cite{FAQ}, which is equivalent to $\ajar$. The FAQ/$\ajar$ problems arose out of discussions between Ngo, Rudra, and R\'e at PODS12 about how to extend the worst-case result to queries using aggregation and message passing via Green et al.'s semiring formulation. We originally worked jointly on the problem, but we developed substantially different approaches. As a result, we split our work. We argue the the $\ajar$ approach is simpler, as it yields the best known runtime results in only a few simple statements in Section~\ref{sec:prelim}. We also describe new complete characterizations as described above. Pragmatically, these completeness results allow us to connect to more easily to existing literature. We have already implemented the algorithm described here in the related database engine EmptyHeaded~\cite{EH}.
\footnote{We have been told that LogicBlox has implemented a similar algorithm recently, but their approach is not public. We shared our implementation with them several months ago.}
This engine has run motif finding, pagerank, and single-source shortest path queries dramatically faster than previous high-level approaches that take datalog-like queries as input.

A primary application of multiple aggregation operators is quantified conjunctive queries (QCQ) and the counting variant, which can be expressed as $\ajar$ queries over the semiring $(\vee, \wedge)$ with aggregations involving both operators. Here, we follow FAQ's idea to formulate this as a query with product aggregation. Chen and Dalmau~\cite{QCQChenDalmau} completely characterized the space of tractable QCQ by defining a notion of width that relies on variable orderings. Chen and Dalmau's width definition includes a complete characterization of the permissible variable orderings for a QCQ instance. Their characterization is similar in spirit to the partial ordering we define in Section~\ref{sec:equivalent-orderings} that characterizes the space of valid GHDs for an $\ajar$ query. However, their results are focused on tractability rather than the optimal runtime exponents; our characterization extends theirs and has improved runtime bounds. 

\section{AJAR and A Simple Solution}
\label{sec:prelim}
We start by describing some background material needed to define the \ajar problem. After that, we formally define the \ajar problem and our solution to it. 

\subsection{Background}\label{subsec:background-1}
We use the classic hypergraph representation for database schema and queries~\cite{ALICE}. A hypergraph $\mH$ is a pair $(\mV, \mE)$, where $\mV$ is a non-empty set of {\em vertices} and $\mE \subseteq 2^\mV$ is a set of {\em hyperedges}. Each $A \in \mV$ is called an {\em attribute}. Each attribute has a corresponding {\em domain} $\mD^A$.
\begin{itemize}
\item \textbf{Data} For each hyperedge $F \in \mE$, there is a
  corresponding relation $R_F \subseteq \prod_{A \in F} \mD^A$; we use
  the notation $\mD^F$ to denote the domain of the tuples $\prod_{A
    \in F} \mD^A$. 

\item \textbf{Join Query} Given a set $\mE$ and a relation $R_F$ for each $F \in \mE$, let $\mV = \cup_{F \in \mE} F$. The join query is written $\Join_{F \in \mE} R_F$ and is defined as
  \[ \left\lbrace t \in \mD^{\mV} \mid \forall F \in \mE : \pi_{F}(t) \in R_F \right\rbrace \]
  
We use $n$ to denote the number of attributes $|\mV|$ and
$m$ to denote the number of relations $|\mE|$. $\IN$ denotes the sum of sizes of input relations in a query, and $\OUT$ denotes the output size.
\end{itemize}

A {\em path} from $A \in \mV_\mH$ to $B \in \mV_\mH$ in a hypergraph $\mH$ is a sequence of attributes, starting with $A$ and ending with $B$, such that each consecutive pair of attributes in the sequence occur together in a hyperedge. The number of attributes in the sequence is the {\em length} of the path.

We now define a GHD of a hypergraph.

\begin{definition}
\footnote{Traditionally GHDs are defined as a triple $(\mT, \chi, \lambda)$ where the function $\lambda: \mV_\mT \to 2^{\mE_\mH}$ assigns relations to each bag. Here we omit this function and implicitly assign \emph{every} relation to each bag (so $\lambda(t) = \mE_\mH$ for all $t \in \mV_\mT$). Though this makes a difference for certain notions of width, it leaves the fractional hypertree width unchanged, as adding more relations to the linear program will never make the objective value worse.}
Given a hypergraph $\mH = (\mV_\mH, \mE_\mH)$, a \emph{generalized hypertree decomposition} is a pair $(\mT, \chi)$ of a tree $\mT = (\mV_\mT, \mE_\mT)$ and function $\chi: \mV_\mT \to 2^{\mV_\mH}$ such that
\begin{itemize}
\item For each relation $F \in \mE_\mH$, there exists a tree node $t \in \mV_\mT$ that covers the edge, i.e. $F \subseteq \chi(t)$.
\item For each attribute $A \in \mV_\mH$, the tree nodes containing $A$, i.e. $\{t \in \mV_\mT | A \in \chi(t)\}$, form a connected subtree.
\end{itemize}
\end{definition}

The latter condition is called the ``running intersection property''. The $\chi(t)$ sets are referred to as `bags' of the GHD. GHDs are assumed to be `rooted' trees, which imposes a top-down partial order on their nodes. Leveraging this order, for any GHD $(\mT, \chi)$ and attribute $A \in \mV_\mH$, we define $TOP_\mT(A)$ to be the top-most node $v \in \mV_\mT$ such that $A \in \chi(v)$.

When each bag of a GHD consists of the attributes of a single relation, the GHD is also called a {\em join tree}. Joins over a join tree can be processed using Yannakakis' algorithm~\cite{Yannakakis81} (pseudo-code in Algorithm~\ref{NoAggroY}). The runtime of Yannakakis' algorithm is $O(\IN + \OUT)$.

GHDs can be interpreted as query plans for joins. Given a GHD, we first join the attributes in each bag using worst case optimal algorithms~\cite{NPRR,LFTJ} to get  one intermediate relation per bag. The intermediate relations can then be joined using Yannakakis' algorithm. This combined algorithm is called GHDJoin; Algorithm~\ref{NoAggroGHD} in Appendix~\ref{sec:background} gives the pseudo-code for GHDJoin. 

\begin{algorithm}[t]
\caption{Yannakakis($\mT = (\mV, \mE$), $\{R_F | F \in \mV\}$)}
\label{NoAggroY}
\textbf{Input:} Join tree $\mT = (\mV, \mE)$, Relations $R_F$ for each $F \in \mV$
\begin{algorithmic}[1]
\ForAll{$F \in \mV$ in some bottom-up order}
\State $P \gets$ parent of $F$
\State $R_P \gets R_P \ltimes R_F$
\EndFor
\ForAll{$F \in \mV$ in some top-down order}
\State $P \gets$ parent of $F$
\State $R_F \gets R_F \ltimes R_P$
\EndFor
\While{$F \in \mV$ in some bottom-up order}
\State $P \gets$ parent of $F$
\State $R_P \gets R_P \Join R_F$
\EndWhile
\State \Return $R_R$ for the root $R$
\end{algorithmic}
\end{algorithm}

The runtime of GHDJoin can be expressed in terms of the {\em fractional hypertree width} of the GHD:

\begin{definition}
Given a hypergraph $\mH = (\mV_\mH, \mE_\mH)$ and a GHD $(\mT, \chi)$, the \emph{fractional hypertree width}, denoted $fhw(\mT, \mH)$, is defined to be $\max_{t \in \mT} \rho^*_t$ in which  $\rho^*_t$ is the optimal value of the following linear program defined for each $t \in \mV_\mT$:
\begin{align*}
\textrm{Minimize } \sum_{F \in \mE_\mH} x_F \log_{\IN}(|R_F|) \text{ such that } \\
\forall A \in \chi(t) : \sum_{F : A \in F} x_F \ge 1, \forall F \in \mE_\mH: x_F \ge 0
\end{align*}
\end{definition}

The fractional hypertree width is just the AGM bound~\cite{AGM} placed on the bags. Thus $IN^{fhw(\mT, \mH)}$ is an upper bound on the sizes of the intermediate relations of GHDJoin. GHDJoin runs in time $\Ot(\IN^{fhw(\mT, \mH)} + \OUT)$ for Join queries. 

\paragraph*{Annotated Relations}
To define a general notion of aggregations, we look to relations annotated with semirings~\cite{Semiring}.

\begin{definition}
A \emph{commutative semiring} is a triple $(S, \oplus, \otimes)$ of a set $S$ and operators $\oplus: S \times S \to S$, $\otimes: S \times S \to S$ where there exist $0, 1 \in S$ such that for all $a,b,c \in S$ the following properties hold:
\begin{itemize}
\item Identity and Annihilation: $a \oplus 0 = a$, $a \otimes 1  = a$, $0 \otimes a = 0$
\item Associativity: $(a \oplus b) \oplus c = a \oplus (b \oplus c)$, $(a \otimes b) \otimes c = a \otimes (b \otimes c)$
\item Commutativity: $a \oplus b = b \oplus a$, $a \otimes b = b \otimes a$
\item Distributivity: $a \otimes (b \oplus c) = (a \otimes b) \oplus (a \otimes c)$
\end{itemize}
\end{definition}

Suppose we have some domain $\mathbb{K}$ and an operator set $O = \{\oplus^1, \oplus^2,\ldots \oplus^k, \otimes\}$ such that $0$ is the identity for each $\oplus^i \in O$ and $(\mathbb{K}, \oplus^i, \otimes)$ forms a commutative semiring for each $i$. We then define a relation with an annotation from $\mathbb{K}$ for each tuple.

\begin{definition}
\label{annotated}
An \emph{annotated relation} with annotations from $\mathbb{K}$, or a $\mathbb{K}$-relation, over attribute set $F$ is a set $\{(t_1, \lambda_1)$, $(t_2, \lambda_2)$, $\dots$, $(t_N, \lambda_N)\}$ such that for all $1 \le i \le N$, $t_i \in \mD^F, \lambda_i \in \mathbb{K}$ and for all $1 \le j \le N : i \neq j \rightarrow t_i \neq t_j$.
\end{definition}

Green et al. define a $\mathbb{K}$-relation to be a function $R_F: \mD^F \to \mathbb{K}$~\cite{Semiring}. Our notion can be viewed as an explicit listing of this function's support. Note that unlike an explicit listing of the function's support, our table does allow tuples with $0$ annotations. However, under our definitions of the operators below, an annotation of $0$ is semantically equivalent to a tuple being absent (we discuss this further in Section~\ref{sec:univ-aggregation}). Note that we can have an annotated relation of the form $R_\emptyset$ of size $1$ containing the empty tuple with some annotation. We now define joins and aggregations over annotated relations. 

\paragraph*{Joins over Annotated Relations} Informally, a join over annotated relations is obtained as follows: (i) We perform a regular join on the non-annotated part of the relations. (ii) For each output tuple $t$ of the join, we set its annotation to the product of the annotations of the input tuples used to produce $t$. We define a join $\Join_{F \in \mE} R_F$ as:

\[ \Join_{F \in \mE} R_{F} = \{ (t,\lambda) : \lambda = \prod_{F \in \mE} \lambda_{F} \text{ in which } (\pi_{F}(t),\lambda_{F}) \in R_{F} \} \] 

\paragraph*{Aggregations over Annotated Relations} An aggregation over an
annotated relation $R_F$ is specified by a pair $(A,\oplus)$ where $A
\in F$, and $\oplus \in O$. The aggregation takes groups of tuples in
$R_F$ that share values of all attributes other than $A$, and produces
a single tuple corresponding to each group, whose annotation is the
$\oplus$-aggregate of the annotations of the tuples in the
group. Suppose that $R$ has schema $R(A,B)$ in which $A$ is a single
attribute and $B$ is a set of attributes. Then, the result of
aggregation $(A,\oplus)$ has only the attributes $B$ and
\[ \sum_{(A,\oplus)} R_{A,B} = \{ (t_{B},\lambda) : t_{B} \in \pi_{B} R \text{ and } \lambda = \sum^{\oplus}_{ (t,\lambda_t) \in R : \pi_B t = t_B  } \lambda_t  \} \]

One can define the meaning of aggregate queries in a straightforward
way: first compute the join and then perform aggregations. Figure~\ref{fig:operations-ex} shows some examples of
operators on relations.  For the remainder of our work, we assume that all relations are $\mathbb{K}$-relations.


\begin{figure}
\begin{tabular}[t]{|c|c||c}
\multicolumn{3}{c}{$R$}\\
A & B & $\mathbb{K}$ \\\hline
1 & 1 & 1 \\
2 & 1 & 2
\end{tabular}
\hspace{5mm}
\begin{tabular}[t]{|c|c||c}
\multicolumn{3}{c}{$S$}\\
B & C & $\mathbb{K}$ \\\hline
1 & 1 & 3 \\
1 & 2 & 4
\end{tabular}
\hspace{5mm}
\begin{tabular}[t]{|c|c|c||c}
\multicolumn{4}{c}{$R \Join S$}\\
A & B & C & $\mathbb{K}$ \\\hline
1 & 1 & 1 & 3 \\
1 & 1 & 2 & 4 \\
2 & 1 & 1 & 6 \\
2 & 1 & 2 & 8
\end{tabular}
\hspace{5mm}
\begin{tabular}[t]{|c||c}
\multicolumn{2}{c}{$\Sigma_{C} S$}\\
B & $\mathbb{K}$ \\\hline
1 & 7
\end{tabular}
\caption{Selected examples illustrating the operators over the semiring $(\textbf{R}_+, +, \cdot)$}
\vspace{-5pt}
\label{fig:operations-ex}
\end{figure}

\subsection{The AJAR problem}
\begin{definition}
Given some global attribute set $\mV$ and operator set $O$, we define an \emph{aggregation ordering} to be a sequence $\alpha = \alpha_1, \alpha_2, \dots, \alpha_s$ such that for each $1 \le i \le s$, $\alpha_i = (a_i, \oplus_i)$ for some $a_i \in \mV, \oplus_i \in O$
\footnote{Note that, by this definition, the operators in aggregation ordering can be the product aggregation $\otimes$. However, product aggregations require different definitions, see Section~\ref{sec:univ-aggregation}.}.
In addition, attributes occur at most once, i.e., $a_j \neq a_k$ for each $1\le j < k \le s$.
\end{definition}

Informally, the aggregation ordering is just a sequence of
attribute-operator pairs such that each attribute in the sequence
occurs at most once. Note that the aggregation ordering specifies the
order and manner in which attributes are aggregated. The ordering does
not need to contain every attribute; we use the term
\emph{output attributes} to denote the attributes not in the ordering.

$V(\alpha)$ represents the set of attributes that appear in $\alpha$, and $V(-\alpha)$ represents $V \backslash V(\alpha)$ (i.e. the output attributes). When $F \subseteq V(\alpha)$, we use $\alpha_F$ to represent a sequence $\beta$ that is equivalent to $\alpha$ restricted to the attributes in $F$, i.e. $V(\beta) = F$, and any $(A, \oplus), (B, \oplus') \in \alpha$ such that $A,B \in F$ must also appear in $\beta$ with their order preserved.

\begin{definition}[\ajar]
Given some hypergraph $\mH = (\mV, \mE)$ and an aggregation ordering $\alpha$, an $\ajar$ query $Q_{\mH, \alpha}$ is a function over instances of $\mH$ such that
\[Q_{\mH, \alpha} (\{R_F | F \in \mE\}) = \Sigma_{\alpha_1} \cdots \Sigma_{\alpha_{|\alpha|}} \Join_{F \in \mE} R_F .\]
\end{definition}

For an \ajar query, we define $\OUT$ to be the final output size, rather than the output size of the join. There are two technical challenges when it comes to solving an \ajar query:
\begin{itemize}
\item Multiple aggregation orders can give the same output over any
  database instance, and using some aggregation orders may give faster
  runtimes than others, e.g. some orders may allow early
  aggregation. Thus we need to identify which orders are equivalent to
  the given order and which order leads to the smallest runtime.
\item $\OUT$ for an \ajar query with $|\alpha| > 0$ is smaller than the output size of the join part of the query. Thus the standard GHDJoin runtime of $\IN^{fhw} + \OUT$ is harder to achieve for \ajar queries. Naively applying a variant of GHDJoin that performs aggregations (Appendix Algorithm~\ref{Aggro}) to \ajar leads to a higher runtime of $\IN^{fhw}\cdot \OUT$ (see Appendix~\ref{sec:background}). Thus we need to identify which GHDs can be used for efficient processing of \ajar queries.
\end{itemize}
We handle these technical challenges in turn.

\subsection{Equivalent Orderings}\label{subsec:simple-solution}
Distinct aggregation orders can be equivalent in that they produce the same output on every instance. For example, suppose $\alpha = ((A, +), (B, +))$ and $\beta = ((B,+), (A,+))$, where $A, B$ are two attributes in some $\mH$. Then two $\ajar$ queries with orderings $\alpha$ and $\beta$ clearly produce the same output for any instance $I$ over $\mH$. This is because we can obtain $\beta$ from $\alpha$ by switching the positions of two adjacent aggregations {\em with the same aggregation operator}. Similarly, if $\mH$ consists only of relations $\{A,B\}, \{B,C\}$, then the orderings $\alpha = ((A, +),(C, \max))$ and $\beta = ((C, \max),(A, +))$ are equivalent, since you can independently aggregate the two attributes away before joining the two relations on $B$. We now formally define equivalent orderings.

\begin{definition}[Equivalent Orderings]
Given a hypergraph $\mH$, define the equivalence relation between orderings $\equiv_\mH$ such that $\alpha \equiv_\mH \beta$ if and only if $Q_{\mH, \alpha} (I) = Q_{\mH, \beta}(I)$ for all database instances $I$ over the schema $\mH$. 
\end{definition}

We say that two operators $\oplus$, $\oplus'$ are distinct over a domain $\mathbb{K}$ (denoted by $\oplus \neq \oplus'$) if $\exists x, y \in \mathbb{K} : x \oplus y \neq x \oplus' y$. And $\oplus = \oplus'$ means that $\forall x,y \in \mathbb{K}$, $x \oplus y = x \oplus' y$. Of course, distinct operators do not, in general, commute.


We now state a theorem specifying two conditions under which aggregations can commute. We will later show these conditions to be complete.

\begin{theorem}
\label{commute}
Suppose we are given a relation $R_F$ such that $A,B \in F$ and two operators $\oplus', \oplus \in O$. Then 
\[ \Sigma_{(A, \oplus)} \Sigma_{(B, \oplus')} R_F = \Sigma_{(B, \oplus')} \Sigma_{(A, \oplus)} R_F\]
if one of the following conditions hold:
\begin{itemize}
\item $\oplus = \oplus'$
\item There exist relations $R_{F_1}$ and $R_{F_2}$ such that $A \notin F_1$, $B \notin F_2$, and $R_{F_1} \Join R_{F_2} = R_F$. 
\end{itemize}
\end{theorem}
\begin{proof}
The first condition follows trivially from the commutativity of our operators. The second condition follows from the fact that we can ``push down'' aggregations. \begin{align*}
 \Sigma_{(A, \oplus)} \Sigma_{(B, \oplus')} R_{F_1} \Join R_{F_2} & = \left(\Sigma_{(B, \oplus')} R_{F_1} \right) \Join \left(\Sigma_{(A, \oplus)} R_{F_2}  \right) \\
 & = \Sigma_{(B, \oplus')} \Sigma_{(A, \oplus)} R_{F_1} \Join R_{F_2}
\end{align*}\end{proof}

\begin{algorithm}[t]
\caption{TestEquivalence($\mH = (\mV_\mH, \mE_\mH)$, $\alpha$, $\beta$)}
\label{algo:equivalance-test}
\textbf{Input:} Query hypergraph $\mH$, orderings $\alpha$, $\beta$.\\
\textbf{Output:} True if $\alpha \equiv_{\mH} \beta$, False otherwise.
\begin{algorithmic}
\If{$|\alpha| = |\beta| = 0$}
\State \Return True
\EndIf
\State Remove $V(-\alpha)$ from $\mH$, then divide $\mH$ into connected components $C_1,\ldots C_m$.
\If{$m > 1$}
\State \Return $\land_{i} \text{TestEquivalence}(\mH, \alpha_{C_i}, \beta_{C_i})$
\EndIf
\State Choose $j$ such that $\beta_j = \alpha_1$. Let $\beta_j = (b_j, \oplus'_j)$.
\If{$\exists i < j : \beta_i = (b_i, \oplus'_i), \oplus'_i \neq \oplus'_j$ and there is a path from $b_i$ to $b_j$ in $\{b_i,b_{i+1},\ldots,b_{|\alpha|}\}$}
\State \Return False
\EndIf
\State Let $\beta'$ be $\beta$ with $\beta_j$ removed.
\State Let $\alpha'$ be $\alpha$ with $\alpha_1$ removed.
\State \Return $\text{TestEquivalence}(\mH, \alpha', \beta')$
\end{algorithmic}
\vspace{-5pt}
\end{algorithm}

These two conditions give us a simple procedure for testing when an ordering $\beta$ is equivalent to the given $\alpha$. Algorithm~\ref{algo:equivalance-test} gives the procedure's pseudo-code. To avoid triviality, we assume $\alpha$ and $\beta$ have the same set of attributes and assign the same operator to the same attributes. First we return true if both $\alpha$ and $\beta$ are empty. Then we check if $\alpha$ can be shown to be equivalent to $\beta$ using the conditions from Theorem~\ref{commute}. This procedure is both sound and complete:

\begin{lemma}\label{lemma:equivalence-test-sound-complete}
Algorithm~\ref{algo:equivalance-test} returns True iff $\alpha \equiv_\mH \beta$.
\end{lemma}

We omit this lemma's proof because it is very similar to and implied by the proofs required in Section~\ref{sec:equivalent-orderings}.

To answer \ajar queries, we need one more component in addition to Algorithm~\ref{algo:equivalance-test}; namely AggroGHDJoin, a straightforward variant of GHDJoin that performs aggregations (Algorithm~\ref{Aggro} in Appendix~\ref{sec:background}). The first step of AggroGHDJoin is similar to that of GHDJoin, namely performing joins within each bag of the GHD to get intermediate relations. We need to do some extra work to ensure that each annotation is multiplied only once, since a relation may be joined in multiple bags. After that, instead of calling Yannakakis' algorithm on the intermediate relations, AggroGHDJoin calls AggroYannakakis (Algorithm~\ref{AggroY} in Appendix~\ref{sec:background}), a well-known variant of Yannakakis that performs aggregations. AggroYannakakis initially performs semijoins like Yannakakis (lines $1$-$8$ in Algorithm~\ref{NoAggroY}). But in the bottom-up join phase (line $11$), AggroYannakakis aggregates out all attributes that have $F$ as their $TOP$ node, before joining $R_F$ with $R_P$. 

Armed with Algorithm~\ref{algo:equivalance-test} and AggroGHDJoin, we have a simple way to answer an $\ajar$ query $Q_{\mH, \alpha}$. We search through all orders, running Procedure $1$ to check for equivalence with $\alpha$. For each order $\beta$ such that $\beta \equiv_\mH \alpha$, we search all through GHDs and check if they are compatible with $\beta$. A GHD $\mT$ is defined to be {\em compatible} with an ordering $\beta$ if, for all attribute pairs $A,B$, $TOP_\mT(A)$ being an ancestor of $TOP_\mT(B)$ implies that either $A$ is an output variable or $A$ occurs before $B$ in $\beta$ (note this precludes $B$ from being an output variable). We can run AggroGHDJoin on any compatible GHD to answer the \ajar query. The runtime of AggroGHDJoin on a compatible GHD $(\mT, \chi)$ is given by $\Ot(\IN^{fhw(\mT,\mH)} + \OUT)$. We choose the compatible GHD that has the smallest $fhw$, and use it to answer the query. The theorem below states our runtime:
\begin{theorem}\label{thm:ajar-runtime}
Given a \ajar query $Q_{\mH,\alpha}$, let $w^*$ denote the smallest fhw for a GHD compatible with an ordering $\equiv_\mH \alpha$; the runtime of our approach is $\Ot(\IN^{w^*} + \OUT)$.
\end{theorem}

\paragraph*{Comparison to Prior Work}
Work by Olteanu and Zavodny~\cite{OZVLDB13,OZTODS15} focuses on a special case of \ajar queries, having a single aggregation operator. For these queries, they have a similar algorithm that iterates over GHDs to find the best compatible one. Their algorithm achieves the same runtime as ours, but cannot handle queries with more than one type of aggregation operator. The FAQ paper uses an algorithm called InsideOut to answer general \ajar queries. The running time of InsideOut equals $\Ot(\IN^{faqw})$ where faqw (FAQ-width) is a new notion of width defined by the FAQ authors~\cite[Section 9.1]{FAQ}. Our algorithm has runtime that is no worse than InsideOut ($w^* \leq faqw, \OUT \leq \IN^{faqw}$), and can be much better when output attributes are present. 
\begin{theorem}\label{thm:ajar-vs-faq-runtime}
For any \ajar query, $w^* \leq faqw$ and $\OUT \le O(\IN^{faqw})$.
\end{theorem}
This theorem is proved in Appendix~\ref{subsec:faq-comparison-proof}. Notice that the InsideOut runtime is not {\em output-sensitive}, i.e. it does not have a $+\text{ }\OUT$ term. As a result the runtime can be very high when the output is small relative to the number of output attributes; this is demonstrated by Example~\ref{ex:faqoutput} in the appendix. FAQ does have a high-level discussion of approaches to make InsideOut output-sensitive~\cite[Section 10.2]{FAQ}; indeed, simply using GHDJoin instead of their bespoke algorithm can achieve output-sensitive bounds, which we discuss in Appendix~\ref{sec:app-related}.


\paragraph*{Discussion} 
We presented a remarkably simple procedure for solving \ajar
queries. The procedure involves a brute force search over different
orderings and GHDs, but this is usually unavoidable as finding the
best ordering and GHD is NP-Hard. Deciding if an ordering is
equivalent to the given ordering is enabled by
Algorithm~\ref{algo:equivalance-test}, which takes time polynomial in
the number of attributes. Determining if a GHD is compatible with an
ordering is straightforward as well. Once the best GHD is found, we
use well known, standard algorithms like AggroGHDJoin to answer the
query efficiently. The resulting runtime exponents are smaller than
those of previous work. The simplicity of the algorithm makes it easy
to implement; we have already implemented a special case of a single
additive operator $\oplus$ in our engine~\cite{EH}.

The equivalence/compatibility tests raise the technically interesting question of finding succinct characterizations of: 
\begin{itemize}
\item All orderings equivalent to any given $\alpha$.
\item All GHDs that are compatible with at least one of the equivalent orderings.
\end{itemize}
We answer the first question in Section~\ref{sec:equivalent-orderings} by providing a simple characterization of all equivalent orderings, and the second question in Section~\ref{sec:decomposing} by defining `valid' GHDs and characterizing their structure in relation to unrestricted GHDs.

\section{Characterizing Equivalent Orderings}\label{sec:equivalent-orderings}
We described a procedure for determining when two orderings are equivalent. The equivalence relation $\equiv_\mH$ defines equivalence classes among the orderings, but these classes may be exponential in size; we find a more succinct characterization that lets us enumerate all equivalent orderings. Chen and Dalmau~\cite{QCQChenDalmau} obtained a similar order-equivalence characterization for a special case of the \ajar problem, namely for aggregations ``and'' and ``or''. The characterization was based on a procedure that generated all equivalent orderings. We improve on this result by providing a simple and succinct characterization of the equivalence class of an aggregation ordering with any number operators.

To that end, we develop an enumeration of the constraints that are sufficient and necessary for an ordering to be in the equivalence class of $\alpha$. The constraints are of the form ``$A$ must always occur before $B$'': 
\begin{definition}[$\pre$]
Given an $\ajar$ query $Q_{\mH, \alpha}$, define a constraint $\pre \subseteq \mV \times \mV$ such that $(A,B) \in \pre$ if and only if $A$ precedes $B$ in all orderings that are equivalent to $\alpha$.
\end{definition}
We say $\pre(A,B)$ is true if and only if $(A,B) \in \pre$.

Trivially, the number of pairs in $\pre$ is less than $n^2$. We note that we can use $\pre$ to define a (strict) partial ordering on the attributes; the constraints are clearly antireflexive, antisymmetric, and transitive. We use $<_{\mH, \alpha}$ to denote this partial order. Given an $\ajar$ query $Q_{\mH, \alpha}$, $<_{\mH, \alpha}$ is a partial order of attribute-operator pairs such that for any $(A, \oplus), (B, \oplus') \in \alpha$, $(A, \oplus) <_{\mH, \alpha} (B, \oplus')$ if $\pre(A,B)$ (see Definition~\ref{def:partial-order} for the exact definition). The partial order $<_{\mH, \alpha}$ is easier to use for proofs; we use the partial order to show the soundness and completeness of these constraints.

\begin{theorem}[Soundness and Completeness of $<_{\mH,\alpha}$]
\label{thm:dag-soundness-completeness}
Suppose we are given a hypergraph $\mH = (\mV, \mE)$ and aggregation orderings $\alpha, \beta$. Then $\alpha \equiv_\mH \beta$ if and only if $\beta$ is a linear extension of $<_{\mH,\alpha}$.
\end{theorem}

We first describe a procedure to compute the precedence relation $\pre$. After that, we reason about its completeness. 

\paragraph*{Computing $\pre$}
To assist in building $\pre$, we define a constraint of the form `$A$ and $B$ cannot commute'':

\begin{definition}[$\dnc$]
Given an $\ajar$ query $Q_{\mH, \alpha}$, define a constraint $\dnc \subseteq \mV \times \mV$ such that $(A,B) \in \dnc$ if and only if $A$ and $B$ are in the same order in any $\beta$ such that $\beta \equiv_{{\cal H}} \alpha$.
\end{definition}

Once again, we say $\dnc(A,B)$ is true if and only if $(A,B) \in \dnc$. We prefer to work with $\dnc$ because we have already discussed when aggregations can commute in Theorem~\ref{commute}; the conditions of that theorem specify when $\dnc$ is $FALSE$. However, we can immediately derive a simple relationship between $\pre$ and $\dnc$:

\begin{lemma}\label{lemma:PREC&DNC}
Given an \ajar query $Q_{\mH, \alpha}$, for any $A, B \in \mV$, $\pre(A,B)$ iff $\dnc(A,B)$ and $A$ precedes $B$ in $\alpha$.
\end{lemma}

We now develop conditions when $\dnc$ is true. Recall that Theorem~\ref{commute} states that two aggregations can commute if $(1)$ they have the same operator or $(2)$ if they can be separated in the join query; the simplest structure that violates both of these conditions is an edge that contains two attributes with differing aggregating operators.

\begin{lemma}
\label{lemma:DNCbasecase}
Given an \ajar query $Q_{\mH, \alpha}$, suppose $(A, \oplus)$, $(B, \oplus')$ $\in \alpha$. If $\oplus \neq \oplus'$ and there exists an edge $E \in \mE$ such that $A, B \in E$, then $\dnc(A,B)$.
\end{lemma}

Lemma~\ref{lemma:DNCbasecase} serves as a base case, but we want to extend the violation of Theorem~\ref{commute}'s conditions beyond single edges to paths. To do so, consider the following examples of how our commuting conditions interact with paths of length two.

\begin{example}
  Consider the query
  \[ \sum_A \max_B \max_C R(A,B) \Join S(B,C) \text{ hence } \alpha = (A, B, C).\]
  No two attributes can be separated, which implies $\dnc(A,B)$ and $\dnc(A,C)$. Lemma~\ref{lemma:DNCbasecase} gives us the former constraint, but not the latter one. This example indicates that it may be possible to extend a constraint $\dnc(A,B)$ along an edge $\{B,C\}$. On the other hand, consider the query
  \[ \max_B \sum_A \max_C R(A,B) \Join S(B,C) \text{ so } \alpha = (B, A, C).\]
 Note that $A$ and $C$ can be separated, which implies that only $\dnc(A,B)$ holds. Note that, as before, Lemma~\ref{lemma:DNCbasecase} gives us this constraint. This example suggests that we cannot extend every $\dnc(A,B)$ constraint along an additional edge. 
\end{example}
The key difference between the two examples is the {\em relative order} of $A$ and $B$ in $\alpha$, which suggests that we can only extend $\dnc(A,B)$ along an edge if $A$ precedes $B$ in $\alpha$, i.e. if $\pre(A,B)$.

\begin{lemma}\label{lemma:DNCextension}
Given an \ajar query $Q_{\mH, \alpha}$, suppose $(A, \oplus)$, $(B, \oplus')$ $\in \alpha$. If $\oplus \neq \oplus'$ and $\exists C \in \mV, E \in \mE: \pre(A,C)$ and $B,C \in E$, then $\dnc(A,B)$.
\end{lemma}

$\pre$ is transitive, which implies:
\begin{lemma}\label{lemma:DNCtransitive}
Given an \ajar query $Q_{\mH, \alpha}$, suppose $(A, \oplus), (B, \oplus') \in \alpha$. If $\exists C: \pre(A,C) \text{ and } \pre(C,B)$, then $\dnc(A,B)$.
\end{lemma}

The above transitivity condition interacts with the condition from Lemma~\ref{lemma:DNCextension} in interesting ways. 
\begin{example}
  Consider the query with $\alpha = (A, B, C, D)$,
  \[ \sum_A \max_B \max_C \sum_D R(A,B) \Join S(B,D) \Join T(C,D).\] No attributes can be separated, which implies $\dnc(A,B)$, $\dnc(A,C)$, $\dnc(B,D)$, and $\dnc(C,D)$. Transitivity gives $\dnc(A,D)$ as well. Now let us derive these constraints using Lemmas~\ref{lemma:DNCbasecase},~\ref{lemma:DNCextension}, and~\ref{lemma:DNCtransitive}. Lemma~\ref{lemma:DNCbasecase} gives us $\dnc(A,C)$, $\dnc(B,D)$, and $\dnc(C,D)$. Note that at this point, Lemma~\ref{lemma:DNCextension} gives us no more constraints. Only after the transitivity of Lemma~\ref{lemma:DNCtransitive} adds the constraint $\dnc(A,D)$ can Lemma~\ref{lemma:DNCextension} add the constraint $\dnc(A,B)$, completing the set of constraints.
\end{example}

It turns out that these three relatively simple lemmas are the sufficient and necessary constraints on the equivalence classes of orderings; no other conditions are necessary to complete the proofs the soundness and completeness of $<_{\mH, \alpha}$.

We note that our current specifications of $\pre$ and $\dnc$ are mutually recursive. The $\pre$ and $\dnc$ sets build up in rounds; Lemma~\ref{lemma:DNCbasecase} provides their initial values, and Lemmas~\ref{lemma:PREC&DNC},~\ref{lemma:DNCextension}, and~\ref{lemma:DNCtransitive} iteratively build up the sets further. We keep applying these lemmas until the sets reach a fixed point. This takes at most $2|\alpha|^2$ iterations, as we must add at least one additional attribute pair per iteration, and there can be only $|\alpha|^2$ pairs of attributes in each set. Thus the overall runtime of computing these constraints is polynomial in the number of attributes. We detail this process in Appendix~\ref{sec:app-equiv}.

For convenience of notation, we make one modification to the definition of the partial order $<_{\mH, \alpha}$. When $A$ is an output attribute and $B$ is not, we define $A <_{\mH, \alpha} B$ to be true. So we can formally state the definition as:

\begin{definition}[$<_{\mH,\alpha}$]\label{def:partial-order}
Given a $\ajar$ query $Q_{\mH,\alpha}$, we define $A <_{\mH,\alpha} B$ to be true if either (i) $A$ is an output attribute and $B$ is not, or (ii) $\pre(A,B)$ is true.
\end{definition}

\paragraph*{Soundness and Completeness of $<_{\mH,\alpha}$}
To give an intuition on how we prove the soundness and completeness of $<_{\mH, \alpha}$, we now state two key lemmas (with proofs in Appendix~\ref{sec:app-equiv}) illustrating properties of $<_{\mH, \alpha}$. 

\begin{lemma}
\label{Pathing}
Suppose we are given a hypergraph $\mH = (\mV, \mE)$ and an aggregation ordering $\alpha$. Suppose $(A, \oplus), (B, \oplus') \in \alpha$ for differing operators $\oplus \neq \oplus'$. Then, for any path $P$ in $\mH$ between $A$ and $B$, there must exist some attribute in the path $C \in P$ such that $C <_{\mH, \alpha} A$ or $C <_{\mH, \alpha} B$. 
\end{lemma}

Lemma~\ref{Pathing} intuitively states that incomparable attributes with different operators must be separated in $\mH$ by their common predecessors in $<_{\mH, \alpha}$.

\begin{lemma}
\label{Pathing2}
Given a hypergraph $\mH = (\mV, \mE)$ and an aggregation ordering $\alpha$, suppose we have two attributes $A, B \in V(\alpha)$ such that $A <_{\mH, \alpha} B$. Then there must exist a path $P$ from $A$ to $B$ such that for every $C \in P, C \neq A$ we have $A <_{\mH, \alpha} C$.
\end{lemma}

Given these two lemmas, the proof of Theorem~\ref{thm:dag-soundness-completeness} is straightforward. Lemma~\ref{Pathing} implies that, given an attribute ordering $\beta$ that is a linear extension of $<_{\mH, \alpha}$, each inversion of attribute-operator pairs must either have equal operators or have attributes that can be separated, allowing us to repeatedly use Theorem~\ref{commute} to transform $\beta$ into $\alpha$. Lemma~\ref{Pathing2} implies that, given an attribute ordering $\beta$ that is not a linear extension of $<_{\mH, \alpha}$, we can construct a counterexample. 

\paragraph*{Discussion}
We obtained a sound and complete characterization of all orderings equivalent to any given ordering. This result extends the work of Chen and Dalmau~\cite{QCQChenDalmau}, who had characterized equivalent orderings for queries with logical ``and'' and ``or'' operators. Our characterization is simple, consisting of a partial order whose linear extensions are precisely the equivalent orderings. FAQ~\cite{FAQ}'s method for identifying equivalent orderings is sound but not complete. That is, there exist equivalent orderings that the FAQ method does not identify as being equivalent (Appendix Example~\ref{example:faq-incompleteness}). In contrast, our characterization is guaranteed to cover all valid orderings. This completeness property lets us create a decomposition that is guaranteed to preserve all {\em node-monotone} widths (see Definition~\ref{def:node-monotone}). This in turn lets us get tighter guarantees on our runtime exponent, using the notion of submodular width (Section~\ref{subsec:dbp-width}).

\section{Decomposing Valid GHDs}
\label{sec:decomposing}

We express our \ajar algorithm directly in terms of GHDs, rather than in terms of aggregation orderings. As such, our goal is the characterization of GHDs that are compatible with at least one equivalent ordering, i.e. the GHDs that can be used to answer an \ajar query. We call a GHD {\em valid} if it is compatible with at least one equivalent ordering. We first give a simple characterization of valid GHDs. Then we demonstrate a way to reduce the problem of finding a minimum-width valid GHD to multiple subproblems on unconstrained GHDs (Section~\ref{subsec:decomposing}). This decomposition of the problem gets us three things:

\begin{itemize}
\item We can speed up our brute force search for an optimal valid GHD. We can also find approximately optimal valid GHDs in polynomial time using Marx's GHD approximation algorithm~\cite{Marx:2010:AFH:1721837.1721845} (Section~\ref{subsec:optimal-valid}).
\item We can apply existing MapReduce join algorithms that utilize GHDs~\cite{GYM}, obtaining efficient parallel algorithms for solving \ajar queries (Section~\ref{subsec:GYM}). 
\item We can apply improved join algorithms~\cite{Marx:2010:THP:1806689.1806790,2015arXiv150801239J} to further reduce our runtime exponent (Section~\ref{subsec:dbp-width}).
\end{itemize}

\subsection{Valid and Decomposable GHDs}\label{subsec:decomposing}
We can easily characterize valid GHDs by combining the definition of compatible GHDs with Theorem~\ref{thm:dag-soundness-completeness}. 

\begin{theorem}
\label{thm:validcorrect}
For a \ajar query $Q_{\mH,\alpha}$, a GHD $(\mT, \chi)$ is \emph{valid} if and only if for every pair of attributes $A,B$ such that $TOP_\mT(A)$ is an ancestor of $TOP_\mT(B)$, $B \not <_{\mH, \alpha} A$.
\end{theorem}

Theorem~\ref{thm:validcorrect} gives us a criterion specifying which GHDs can act as query plans. We now consider the problem of finding a minimum width valid GHD for any $\ajar$ query. We call a GHD {\em optimal} if it has the minimum width possible for valid GHDs. We show how to reduce the problem of finding an optimal valid GHD into smaller problems of finding ordinary optimal GHDs. This unlocks a trove of powerful GHD results and makes them applicable to our problem.

Given an \ajar query $Q_{\mH, \alpha}$, suppose we have a subset of the nodes $V \subseteq \mV$. Define $\mE_V$ to be $\{E \in \mE | E \cap V \neq \emptyset\}$, i.e. the set of edges that intersect with $V$. As before, $\alpha_V$ denotes the aggregation ordering restricted to the nodes in $V$. Additionally define $V^O$ to be $\{ v \in V | \forall w \in V, w \not <_{\mH, \alpha} v\}$, i.e. the nodes in $V$ that have no predecessors in $V$ according to the partial ordering $<_{\mH, \alpha}$. Finally, note that $\alpha_{V\backslash V^O}$ is then $\alpha_V$ with all the nodes in $V^O$ removed (note that this makes the nodes in $V^O$ output attributes).

\begin{definition}\label{def:decomp}
Given an \ajar query $Q_{\mH, \alpha}$, we say a GHD $(\mT, \chi)$ is \emph{decomposable} if:
\begin{itemize}
\item There exists a rooted subtree $\mT_0$ of $\mT$ such that $\chi(\mT_0) = \mV(-\alpha)$ (i.e. output attributes).
\item For each connected component $C$ of $\mH \backslash V_{-\alpha}$, there is exactly one subtree $\mT_{C} \in \mT \backslash \mT_0$ such that $\mT_C$ is a decomposable GHD of $Q_{(\cup_{E \in \mE_C} E, \mE_C), \alpha_{C \backslash C^O}}$.
\end{itemize}
\end{definition}

We start by connecting this idea of decomposable GHDs to valid GHDs. We only give proof sketches here; see appendix~\ref{sec:app-decomp} for the full proofs.

\begin{theorem}\label{thm:decompisvalid}
Every decomposable GHD is valid.
\end{theorem}
\begin{proof}(Sketch)
Suppose the $\ajar$ query is $Q_{\mH, \alpha}$. We need to show for any $A,B$ such that $TOP_\mT(A)$ is an ancestor of $TOP_\mT(B)$, $A \not <_{\mH, \alpha} B$. We use induction on $|\alpha|$. If $|\alpha|=0$, all GHDs are valid and decomposable. For $|\alpha| > 0$, $\mT_0$ ensures that the output attributes are above non-output attributes. If $A$ and $B$ are non-output attributes and $TOP_\mT(A)$ is an ancestor of $TOP_\mT(B)$, then both are in some $\mT_C$. By the inductive hypothesis, $\mT_C$ is valid with respect to $Q_{(\cup_{E \in \mE_C} E, \mE_C), \alpha_{C \backslash C^O}}$. By inspecting the partial order created by this subgraph, we conclude that $A \not <_{\mH, \alpha} B$ as desired.
\end{proof}

Every valid GHD may not be decomposable. However, every valid GHD can be transformed into a corresponding decomposable GHD using some simple transformations. Each bag of the resulting decomposable GHD is a subset of one of the bags of the original GHD. Thus the fhw of the decompsable GHD is at most the fhw of the original valid GHD. In fact, we can make a more general claim, using a notion of node-monotone functions, defined next.

\begin{definition}\label{def:node-monotone}
Given a hypergraph $\mH = (\mV_\mH, \mE_\mH)$, we define a function to
be {\em node-monotone} if it is a function $\gamma: 2^{\mV_\mH}
\rightarrow \mathbb{R}$ such that $\forall \text{ } A \subseteq B
\subseteq \mV_\mH : \gamma(A) \leq \gamma(B)$.  Given any
node-monotone function $\gamma$, we define the $\gamma$-width of a GHD $(\mT, \chi)$ over
$\mH$ as $\max_{v \in \mV_\mT} \gamma(\chi(v))$.
\end{definition}

Many standard notions of widths can be expressed as $\gamma$-widths for a suitably chosen $\gamma$. Specifically:

\begin{proposition}\label{prop:node-monotone-widths}
Suppose we are given a hypergraph $\mH = (\mV_\mH, \mE_\mH)$ and database instance $I$ on $\mH$. Then for each of following notions of width: (i) Treewidth (ii) Generalized Hypertree Width (iii) Fractional Hypertree Width (iv) Submodular Width, there exists a node-monotone function $\gamma$ such that $\gamma$-width equals the given notion of width. 
\end{proposition}

As a simple example, tree-width can be expressed as $\gamma$-width for $\gamma(A) = |A|-1$. We can now relate valid and decomposable GHDs with respect to their $\gamma$-widths.

\begin{theorem}\label{thm:decompwidth}
For every valid GHD $(\mT, \chi)$, there exists a decomposable GHD $(\mT', \chi')$ such that for all node-monotone functions $\gamma$, the $\gamma$-width of $(\mT', \chi')$ is no larger than the $\gamma$-width of $(\mT, \chi)$.
\end{theorem}
\begin{proof}(Sketch)
Suppose the $\ajar$ query is $Q_{\mH, \alpha}$. We transform the given GHD $(\mT, \chi)$ into $(\mT', \chi')$ such that for each $v' \in \mV_\mT'$, there exists a $v \in \mV_\mT$ such that $\chi'(v') \subseteq \chi(v)$. The result then follows from the node-monotonicity of $\gamma$ and the definition of $\gamma$-width. Any transformation of a GHD that ensures that all new bags are subsets of old bags, is called width-preserving.

We then transform the GHD $(\mT, \chi)$ to satisfy the following properties (using width-preserving transformations):
\begin{itemize}
\item Every $t \in \mT$ is $TOP_\mT(A)$ for exactly one attribute $A$.
\item For any node $t \in \mT$ and the subtree $\mT_t$ rooted at $t$, the attributes $\{v \in \mV | TOP_\mT(v) \in \mT_t\}$ form a connected subgraph of $\mH$.
\end{itemize}
We can show, by induction, any valid GHD that satisfies these two properties is decomposable. Intuitively, the first transformation ensures the subtree $\mT_0$ exists as desired. The second transformation ensures that each of the $\mT_C$'s exists and satisfies the requisite properties.
\end{proof}

This theorem lets us restrict our search to the smaller space of decomposable GHDs (instead of all valid GHDs) when looking for the optimal valid GHD. Moreover, the space of decomposable GHDs is simpler; it can be factored into smaller spaces of unconstrained GHDs, as we show next. We present the definition of \emph{characteristic hypergraphs}, which are intuitively the set of hypergraphs that specify the factors, i.e. the unconstrained GHDs. 

Our goal is two-fold: $(1)$ to be able to split a decomposable GHD into component GHDs of the characteristic hypergraphs and $(2)$ to be able to take arbitrary GHDs of the characteristic hypergraphs and connect them to create a decomposable GHD of the original \ajar problem. The definition of decomposable GHDs decomposes a GHD into a series of sub-trees $\mT_0, \dots, \mT_k$. The definition specifies that the subtrees $\mT_1, \dots, \mT_k$ must be decomposable GHDs of (smaller) \ajar problems. Additionally, it is simple to show $\mT_0$ is a GHD of the hypergraph $(V(-\alpha), \{E \in \mE | E \subseteq V(-\alpha)\})$. If we apply this decomposition recursively to the subtrees $\mT_1, \dots, \mT_k$, we can divide any decomposable GHD into a series of (unrestricted) GHDs of particular hypergraphs. This provides the basis of our definition of the characteristic hypergraphs; we define a hypergraph $\mH_0$ that specifies the hypergraph corresponding to $\mT_0$ and then recurse on the smaller \ajar queries specified in Appendix Definition~\ref{def:decomp2}.

However, if we are given arbitrary GHDs of the hypergraphs as defined thus far, we may not be able to stitch them together while preserving the running intersection property of GHDs. To ensure this stitching is possible, we need the characteristic hypergraphs to contain additional edges that we can use to guarantee the running intersection property. Intuitively the edges we add will be the intersections of the adjacent subtrees in our decomposition; for example, for any connected component $C$ of $\mH \backslash V(-\alpha)$, $\mT_0$ and $\mT_C$ are adjacent, and we will add the edge $\chi(\mT_0) \cap \chi(\mT_C)$ to the corresponding hypergraphs. We can use these `intersection edges' to connect particular nodes in the adjacent subtrees.

\begin{definition}
Given an $\ajar$ problem $Q_{\mH, \alpha}$, suppose $C_1, \dots, C_k$ are the connected components of $\mH \setminus \mV_{-\alpha}$. Define a function $H(\mH, \alpha)$ that maps $\ajar$ queries to a set of hypergraphs as follows: 
\begin{itemize}
\item $C_i^+ = \bigcup_{E \in \mE_{C_i}} E$ for all $1 \le i \le k$
\item $\mH_0 = (\mV_{-\alpha}, \{F \in \mE | F \subseteq \mV_{-\alpha}\} \cup \{\mV_{-\alpha} \cap C_i^+ | 1 \le i \le k\})$
\item $\mH_i^+ = (C_i^+, \mE_{C_i} \cup \{\mV_{-\alpha} \cap C_i^+\})$
\item $H(\mH, \alpha) = \{\mH_0\} \cup \bigcup_{1 \le i \le k} H(\mH_i^+, \alpha_{C_i \backslash C_i^O})$
\end{itemize}
The hypergraphs in the set $H(\mH, \alpha)$ are defined to be the \emph{characteristic hypergraphs}.
\end{definition}

Note that the definition of characteristic hypergraphs depends only on $(\mH, \alpha)$, and not on a specific GHD or the instance. Now we state a key result that lets us reduce the problem of searching for an optimal {\em valid} GHD over $\mH$ to that of searching for (not necessarily valid) optimal GHDs over characteristic hypergraphs. Each decomposable GHD corresponds to a GHD over each characteristic hypergraph; conversely, a combination of GHDs for characteristic hypergraphs gives us a decomposable GHD for $\mH$. Formally:

\begin{theorem}\label{thm:decomp}
For an $\ajar$ query $Q_{\mH, \alpha}$, suppose $\mH_0, \dots, \mH_k$ are the characteristic hypergraphs $H(\mH, \alpha)$. Then GHDs $G_0, G_1, \dots, G_k$ of $\mH_0, \dots, \mH_k$ can be connected to form a decomposable GHD $G$ for $Q_{\mH, \alpha}$. Conversely, any decomposable GHD $G$ of $Q_{\mH, \alpha}$ can be partitioned into GHDs $G_0, G_1, \dots, G_k$ of the characteristic hypergraphs $\mH_0, \dots, \mH_k$. Moreover, in both of these cases, $\gamma\text{-width}(G) = \max_{i} \gamma\text{-width}(G_i)$.
\end{theorem}

The proof is provided in the appendix, but it is a straightforward application of definitions.

\begin{corollary}\label{cor:optimal-construct}
Given an optimal GHD for each characteristic hypergraph of an $\ajar$ query $Q_{\mH, \alpha}$, we can construct an optimal valid GHD. The width of the optimal valid GHD equals the maximum optimal-GHD-width over its characteristic hypergraphs.
\end{corollary}

This reduces the problem of finding the optimal valid GHD to smaller problems of finding optimal GHDs. We first present the decomposition in the FAQ~\cite{FAQ} paper. Then we present several applications of our decomposition, and compare them to their FAQ analogues. 

\paragraph*{FAQ's Decomposition} The FAQ paper uses a
  decomposition of the problem that is not width-preserving. They
  remove the set of output attributes $V(-\alpha)$ and decompose the
  rest of the hypergraph into smaller hypergraphs. They construct a
  regular Variable-Ordering/GHD for each hypergraph. Then they
  add all output attributes $V(-\alpha)$ into each bag of each of the
  GHDs, and then stitch the GHDs together. This {\em output addition}
  to the bags of the GHDs leads to a potentially $2\times$ increase in
  width compared to our method which stitches the GHDs together
  without changing their width. As a result, FAQ's decomposition
  incurs higher runtime costs in each application of
  the decomposition, as we see in the next three subsections.

\begin{example}\label{example:faq-2x}
  Consider a query with output attribute $A$
  \[ \sum_{B,+}\sum_{C,+} (R(A,B) \Join S(B,C)). \] The optimal valid GHD for this query has bags $\{A,B\}$ and $\{B,C\}$, and thus has fhw $1$. The faqw is also $1$. If we apply our decomposition, we get a GHD with bags $\{A\}$, $\{A,B\}$, $\{B,C\}$ which still has fhw $1$. FAQ's decomposition on the reduced hypergraph (with output attribute $A$ removed) has one bag $\{B,C\}$. Adding $A$ to it gives a single bag $\{A,B,C\}$ resulting in a fhw of $2$. More generally, consider query $Q_n$ with $\alpha =$ $((B_1,+), (B_2,+),\ldots (B_n,+))$ and relations $T(A_1,B_1)$ and also $R_{i,j}(A_i,A_j)$, $S_{i,j}(B_i,B_j)$ for $i,j \in \{1,2,\ldots,n\}$. Our decomposition gives a GHD with bags $\{A_1,A_2,\ldots,A_n\}$, $\{A_1,B_1\}$, $\{B_1,B_2,\ldots,B_n\}$, which has fhw $n/2$. FAQ's decomposition has a single bag and fhw equal to $n$.
\end{example}  

\subsection{Finding optimal valid GHDs}\label{subsec:optimal-valid}
Armed with Corollary~\ref{cor:optimal-construct}, we simplify the brute force search algorithm for finding optimal valid GHDs.

\begin{theorem}
Let $Q_{\mH, \alpha}$ be an $\ajar$ query. The optimal width valid GHD for this query can be found in time $\Ot(|\mH|2^{\Ot(\max_{\mH' \in H(\mH, \alpha)}(|\mH'|))})$.
\end{theorem}
This runtime for finding the optimal valid GHD can be exponentially better than the naive runtime:
\begin{example}\label{example:star-ajar}
Consider the star query $\mH = (\{A, B_1,\ldots B_n\}$, $\{ \{A, B_i\} \mid 1 \leq i \leq n \} )$, $\alpha = (B_1, +), (B_2, +), \dots, (B_n, +)$. $A$ is the only output attribute. Removing $A$ breaks the hypergraph into $n$ components, so there are $n + 1$ characteristic hypergraphs, each of size $\leq 2$. Finding the optimal valid GHD takes time $\Ot(n)$, whereas the standard algorithm takes time exponential in $n$.
\end{example}

We can also approximate the GHD~\cite{Marx:2010:AFH:1721837.1721845}:

\begin{theorem}[Marx's GHD approximation]
Let $Q$ be a join query with hypergraph $\mH$ and fractional hypertree width $w$. Then we can find a GHD for $Q$ in time polynomial in $|\mH|$, that has width $w' \leq w^3$.
\end{theorem}

We can replicate Marx's result for valid GHDs. 
\begin{theorem}
Let $Q_{\mH, \alpha}$ be an $\ajar$ query, such that its minimum width valid GHD has width $w$. Then we can find a valid GHD in time polynomial in $|\mH|$ that has width $w' \leq w^3$.
\end{theorem}

FAQ~\cite{FAQ}'s decomposition lets them apply Marx's approximation as well. However, their decomposition is not width-preserving i.e. the width of their final GHD is higher than the width of the GHDs they construct for the hypergraphs in the decomposition. Thus their decomposition gives a weaker width guarantee of $faqw^3 + faqw$~\cite[Theorem~9.49]{FAQ}. The extra $+ faqw$ factor is due to output addition. Our guarantee, $w^3$, is strictly smaller ($w$ is the width of the optimal valid GHD) as $w \leq faqw$ by Theorem~\ref{thm:ajar-vs-faq-runtime}. 

\subsection{Tighter Runtime Exponents}\label{subsec:dbp-width}
Marx~\cite{Marx:2010:THP:1806689.1806790} introduced the notion of submodular width ($sw$) that is tighter than $fhw$, and showed that a join query can be answered in time $\IN^{O(sw)}$. The $O$ in the exponent is because Marx's algorithm requires expensive preprocessing that takes $\IN^{2 \times sw}$ time. After the pre-processing, the join can be performed in time $\IN^{sw}$. Despite the $O$ in the exponent, this algorithm can be very valuable because there are families of hypergraphs that have unbounded $fhw$ but bounded $sw$. We can apply Marx's algorithm to the characteristic hypergraphs, potentially improving our runtime. Marx also showed that joins on a family of hypergraphs are fixed parameter tractable if any only if the submodular width of the hypergraph family is bounded~\cite{Marx:2010:THP:1806689.1806790}. Moreover, adaptive width~\cite{AdaptiveWidth} (applicable only when relations are expressed as truth tables) is unbounded for a hypergraph family if and only if submodular width is unbounded. Corollary~\ref{cor:optimal-construct} gets us an analogous tractability result for $\ajar$ queries. 

\begin{theorem}\label{thm:submodular-width}
We can answer an $\ajar$ query $Q_{\mH, \alpha}$ in time $O(\IN^{O(\max_{\mH' \in H(\mH, \alpha)}(sw(\mH')))} + \OUT)$. 
\end{theorem}

Recent work~\cite{2015arXiv150801239J} uses degree information to more tightly bound the output size of a query. The bound in the reference, called the {\em DBP bound}, has a tighter exponent than the AGM bound, while requiring only linear preprocessing to obtain. The authors also provide algorithms whose runtime matches the DBP bound. We can define DBP-width $dbpw(\mT, \mH)$ such that $\IN^{dbpw(\mT, \mH)}$ is the maximum value of the DBP bound over all bags of GHD $T$. We then use the improved algorithm in place of GJ in AggroGHDJoin. This lets us get tighter results ``for free'', reducing our runtime to $\IN^{dbpw}$ instead of $\IN^{fhw}$. Formally:

\begin{theorem}\label{thm:dbp-width}
Given an $\ajar$ query $Q_{\mH, \alpha}$ and a valid GHD for $\mH$, we can answer the query in time $O(\IN^{dbpw(\mT, \mH)} + \OUT)$. Equivalently, we can answer the query in time $O(\IN^{\max_{\mH' \in H(\mH, \alpha)} dbpw(\mT, \mH)} + \OUT)$.
\end{theorem}

As discussed before, FAQ has a non-width-preserving decomposition. We can combine FAQ's decomposition with the DBP bound as we did above. Suppose we perform FAQ's decomposition, and $\IN^{faqw+}$ denotes the highest value of the DBP bound on each of their characteristic hypergraphs, and on the set of output attributes. Thus the DBP-width of each of their characteristic hypergraphs, and the outputs, is $faqw+$. However, when they perform output addition, the DBP-width of the resulting GHDs can go up to $2faqw+$. This happens when the DBP bound on both the outputs and one of the characteristic hypergraphs equals $\IN^{faqw+}$. So if we apply the DBP result to FAQ's decomposition, we get a runtime of $\Ot(\IN^{2faqw+} + \OUT)$. Thus their decomposition causes them to incur an extra factor of $2$ in the exponent. They similarly incur a factor of $2$ increase in exponent for the submodular width algorithm. 

\subsection{MapReduce and Parallel Processing}\label{subsec:GYM}
The GYM algorithm~\cite{GYM} uses GHDs to efficiently process joins in a MapReduce setting. GYM makes use of the GHD structure to parallelize different parts of the join. Given a GHD of depth $d$, and width $w$, with $n$ attributes, GYM can perform a join in a MapReduce setting in $O(d + \log(n))$ rounds at a communication cost of $M^{-1}(\IN^w + \OUT)^2$ where $M$ is the memory per processor on the MapReduce cluster. Combining this with the degree-based MapReduce algorithm~\cite{2015arXiv150801239J} gives us the following result:

\begin{theorem}\label{thm:ajar-mapreduce-n-rounds}
Given an optimal valid GHD $(\mT^*, \chi)$ of depth $d$, and DBP-width $dbpw$, we can answer an $\ajar$ query with Communication Cost equal to $O(M^{-1}(\IN^{dbpw(\mT, \mH)} + \OUT)^2)$ in $d + \log(n)$ MapReduce rounds, where $n$ is the number of attributes and $M$ is the available memory per processor.
\end{theorem}

A GHD can have depth up to $O(n)$, in which case the algorithm can take a very large number of MapReduce rounds ($O(n)$). To address this, the GYM paper uses the `Log-GTA' algorithm to reduce the depth of any given GHD to $\log(n)$ while at most tripling its width. This lets it process joins in $\log(n)$ MapReduce rounds at a cost of $M^{-1}(\IN^{3w} + \OUT)^2$.

Log-GTA involves some shuffling of the attributes in the GHD bags, so naively applying it to a valid GHD could make the GHD invalid (see example~\ref{example:Log-GTA-invalidation} in the Appendix). But our decomposition lets us apply Log-GTA to the GHD of each characteristic hypergraph, and then stitch the short GHDs together. Our decomposition is recursive in nature; let $d'$ be the maximum recursive depth of the decomposition for a given $Q_{\mH,\alpha}$. Then the depth of the shortened GHD of each characteristic hypergraph is $O(\log(n))$, and so the depth of the valid GHD obtained by stitching them together is $O(d'\log(n))$. This gives us the result:

\begin{theorem}\label{thm:ajar-mapreduce-log-rounds}
If $dbpw$ is the DBP width of a \ajar query, we can answer that query with Communication Cost equal to $O(M^{-1}(\IN^{3 \times dbpw(\mT, \mH)} + \OUT)^2)$ in $d'\log(n)$ MapReduce rounds, where $n$ is the number of attributes and $M$ is the available memory per processor.
\end{theorem}

$d'$ can vary from $O(1)$ to $O(n)$ depending on the query. The star query from example~\ref{example:star-ajar} has $d' = 2$, which lets us process it in $\log(n)$ MapReduce rounds. Any query that only has a single type of aggregation will have $d' = 2$ as well. On the other hand, a query with one relation having $n$ attributes, $1$ output attribute, and alternating $\sum$ and $\max$ aggregations, will have $d' = n$, and will be hard to parallelize. 

Olteanu and Zavodny~\cite{OZVLDB13,OZTODS15} use valid GHDs to answer \ajar queries for the special case of a single type of aggregation. But they have no notion of a decomposition and attempting to shorten a valid GHD directly, without using a decomposition, may make it invalid. FAQ's decomposition may be used to shorten GHDs similarly to ours, but leads to an increased width of $4faqw$ compared to our $3w$ (where $w \leq faqw$ is the width of our optimal valid GHD). This is again because of output addition, if the output attributes have a width of $faqw$, and the shortened GHDs of the characteristic hypergraphs have a width of $3faqw$, then the total width will be $4faqw$.

\section{Product Aggregations}\label{sec:univ-aggregation}
The primary application of queries with multiple aggregations is to establish bounds for the Quantified Conjunctive Query ($QCQ$) problem~\cite{FAQ}, and its counting variant, $\#QCQ$. We now introduce a new type of aggregation, called product aggregation, that lets us efficiently handle $QCQ$ queries. We define the \ajar problem for product aggregations, and then extend our algorithm from Section~\ref{subsec:simple-solution} to handle this new type of \ajar query. We then define a decomposition analogous to that in Section~\ref{sec:decomposing}. A more detailed version of this section with additional motivation, examples, and proofs can be found in Appendix~\ref{sec:univ-aggregation-proofs}.

\subsection{AJAR queries with product aggregates} 

A {\em product aggregation} aggregates using the $\otimes$ operator. Throughout the paper, we assumed that an absent tuple effectively has an annotation of $0$. Taking this into account, we formally define the product aggregation. Let $B = F \backslash A$:

\begin{definition}
$\displaystyle \sum_{(A, \otimes)} R_{AB} = \{(t_B, \lambda): \forall t_A \in \mD^A, t_B \circ t_A \in R_{AB} \text{ and } \lambda = \prod_{(t, \lambda_t) \in R_{AB}: \pi_B t = t_B} \lambda_t \}$
\end{definition}

We can adjust the definition of aggregation orderings and $\ajar$ queries to allow this new type of aggregation. $QCQ$ queries can now be expressed as \ajar queries on the $(\{0,1\}, \max, \cdot)$ semiring. We assume for this section that $\otimes$ is idempotent, i.e. $a \otimes a = a$ for all $a$. We describe how to work with non-idempotent products in Appendix~\ref{subsec:non-idempotent-product}. 


\subsection{Algorithms for product aggregates}
For aggregation orderings that have product aggregations, the rules for determining when two orderings are equivalent are somewhat different; product aggregations can be performed {\em before} a join. We illustrate this with an example:
\begin{example}
In the semiring $(\{0,1\}, \max, \cdot)$, suppose we have two relations $R(A,B) = \{((0,0), x), ((0,1), y)\}$ and $S(B,C) = \{((0, 1), p), ((1,1),q)\}$. Consider the \ajar query $\sum_{(B, \cdot)} R(A,B) \Join S(B,C)$. If we compute the join, we will get two tuples with the annotations $x \cdot p$ and $y \cdot q$, and then aggregating over $B$ will produce a relation with the element $((0,1), x \cdot p \cdot y \cdot q)$. However, note that $x \cdot p \cdot y \cdot q = (x \cdot y) \cdot (p \cdot q)$, implying that $\sum_{(B, \cdot)} R(A,B) \Join S(B,C) = (\sum_{(B, \cdot)} R(A,B)) \Join (\sum_{(B, \cdot)} S(B,C))$.
\end{example}

Now we describe our algorithm for solving \ajar queries when product aggregations are present. Our algorithm follows the same lines as the algorithm from Section~\ref{subsec:simple-solution}. Recall that the algorithm consisted of searching for {\em equivalent orderings}, then searching for GHD {\em compatible} with an equivalent ordering, and running AggroGHDJoin on the GHD with the smallest fhw. For product aggregations, we need to modify our algorithm for testing equivalent orderings, and our definition of compatibility; we do these in turn.

\paragraph*{Testing orderings for equivalence} We describe how we modify Algorithm~\ref{algo:equivalance-test} when product aggregations are present. Let $\pa(\alpha)$ denote the set of product attributes in ordering $\alpha$. We make two changes to Algorithm~\ref{algo:equivalance-test}. (1) Instead of removing $V(-\alpha)$ and dividing $\mH$ into components, we remove $V(-\alpha) \cup \pa(\alpha)$ and then divide $\mH$ into components $C_1, C_2,\ldots,C_m$. Then for each $C_i$ we define $C'_i = C_i \cup \bigcup_{e \in \mE_{C_i}} (\pa(\alpha) \cap e)$.
\footnote{Recall for any $V \subseteq \mV$, $\mE_V$ is defined to be $\{E \in \mE| E \cap V \neq \emptyset\}$.}
That is $C'_i$ has the attributes of $C_i$ as well as the product attributes that are in the same hyperedge as some attribute in $C_i$. Then we recursively call the equivalence test on $(\alpha_{C'_i}, \beta_{C'_i})$ instead of on $(\alpha_{C_i},\beta_{C_i})$. (2) When we are checking for a $i < j$ such that $\odot'_i \neq \odot'_j$ and there is a path in $\{b_i,b_{i+1},\ldots,b_{|\alpha|}\}$, we instead check for a path in 
$$(\{b_i,b_{i+1},\ldots,b_{|\alpha|}\} \setminus \pa(\alpha)) \cup \{b_i,b_j\}$$ 
That is, we look for a $b_i$ that has a different operator that $b_j$, and has a path to $b_j$ consisting only of $b_i$, $b_j$, and semiring attributes in $\{b_i,b_{i+1},\ldots,b_{|\alpha|}\}$. Appendix~\ref{sec:univ-aggregation-proofs} gives the pseudo-code for the modified algorithm (Algorithm~\ref{algo:equivalance-test-prod}) and proves that it is sound and complete.

\begin{lemma}\label{lemma:AlgoP-sound-complete}
The above Algorithm returns True if and only if $\alpha \equiv_\mH \beta$.
\end{lemma}

\paragraph*{Compatible GHDs} Product aggregates not only change the set of equivalent orderings, but also the set of GHDs compatible with a given ordering. In fact, product aggregates allow us to break the rules of GHDs without causing incorrect behavior. We express this using a simple variant of GHDs, called {\em aggregating generalized hypertree decompositions} (AGHDs). Informally, AGHDs are GHDs that can violate the running intersection property for attributes that have a product aggregation. AGHDs are formally defined in Appendix~\ref{sec:univ-aggregation-proofs}. We determine compatibility for AGHDs as follows: An AGHD is compatible with an ordering $\beta$ if for every attribute pair $a$, $b$ such that one of the $TOP(a)$ nodes is an ancestor of a $TOP(b)$ node, $a$ precedes $b$ in $\beta$. 

We can now modify our algorithm from Section~\ref{subsec:simple-solution} to detect equivalent orderings using Algorithm~\ref{algo:equivalance-test-prod}, then search for compatible AGHDs, and run AggroGHDJoin over the compatible AGHD with the smallest fhw. Our runtime is given by the next theorem.

\begin{theorem}\label{thm:ajar-product-runtime}
Given a \ajar query $Q_{\mH,\alpha}$ possibly involving idempotent product aggregates, let $w^*$ be the smallest fhw for an AGHD compatible with an ordering equivalent to $\alpha$. Then the runtime for our algorithm is $\Ot(\IN^{w^*} + \OUT)$.
\end{theorem}

\paragraph*{Decomposing AGHDs} We can apply the ideas from Section~\ref{sec:decomposing} to \ajar queries with product aggregates as well. We can define a notion of decomposable AGHDs for queries with product aggregates, and show the following results:

\begin{theorem}\label{thm:decompaghdisvalid}
All decomposable AGHDs are compatible with an ordering $\beta$ such that $\beta \equiv_\mH \alpha$.
\end{theorem}

\begin{theorem}\label{thm:decompaghdwidth}
For every valid AGHD $(\mT, \chi)$, there exists a decomposable $(\mT', \chi')$ such that for all node-monotone functions $\gamma$, the $\gamma$-width of $(\mT', \chi')$ is no larger than the $\gamma$-width of $(\mT, \chi)$.
\end{theorem}

We can define characteristic hypergraphs similarly to how we did in Section~\ref{sec:decomposing} (see Appendix~\ref{sec:univ-aggregation-proofs} for a formal definition). We have the following result:

\begin{theorem}\label{thm:decompaghd}
For an $\ajar$ query $Q_{\mH, \alpha}$ involving product aggregates, suppose $\mH_0, \dots, \mH_k$ are the characteristic hypergraphs $H(\mH, \alpha)$. Then GHDs $G_0, G_1, \dots, G_k$ of $\mH_0, \dots, \mH_k$ can be connected to form a decomposable AGHD $G$ for $Q_{\mH, \alpha}$. Conversely, any decomposable AGHD $G$ of $Q_{\mH, \alpha}$ can be partitioned into GHDs $G_0, G_1, \dots, G_k$ of the characteristic hypergraphs $\mH_0, \dots, \mH_k$. Moreover, in both of these cases, $\gamma\text{-width}(G) = \max_{i} \gamma\text{-width}(G_i)$.
\end{theorem}

These theorems let us apply all the optimizations from Section~\ref{subsec:optimal-valid},~\ref{subsec:dbp-width}, and~\ref{subsec:GYM} to \ajar queries with product aggregates.

\paragraph*{Comparison to FAQ} The runtime of InsideOut on a query involving idempotent product aggregations is given by $\Ot(\IN^{faqw})$, where the faqw depends on the ordering, and the presence of product aggregations. Our algorithm for handling product aggregations recovers the runtime of FAQ. Formally,

\begin{theorem}
For any \ajar query involving idempotent product aggregations, $\IN^{w^*} + \OUT \leq 2 \cdot \IN^{faqw}$.
\end{theorem}

The proof is in Appendix~\ref{subsec:faq-comparison-proof}. By applying ideas from the FAQ paper to our setting, we can also recover the FAQ runtime on $\#QCQ$ (Appendix~\ref{subsec:recovering-hash-qcq}). Algorithm~\ref{algo:equivalance-test-prod} for detecting equivalence of orderings is both sound and complete; in contrast, FAQ's equivalence testing algorithm is sound but not complete. Moreover, we have a width-preserving decomposition for queries with product aggregates. This allows us to get tighter runtime exponents in terms of submodular and DBP-widths (Theorems~\ref{thm:submodular-width},~\ref{thm:dbp-width}) and efficient MapReduce Algorithms (Theorems~\ref{thm:ajar-mapreduce-n-rounds},~\ref{thm:ajar-mapreduce-log-rounds}). 

\section{Conclusion}
\label{sec:conclusion}
We investigate solutions to and the structure of $\ajar$ queries: aggregate-join queries with multiple aggregators over annotated relations. We start by providing a very simple algorithm based on a variant of the standard GHDJoin algorithm that generates query plans by relying on a simple test of equivalence between aggregation orderings. This naive approach is sufficient to recover and surpass the runtime of state-of-the-art solutions. We proceed to investigate more interesting technical questions regarding the structure of $\ajar$ queries. We first develop a partial ordering that fully characterizes equivalent orderings. We then develop a characterization of the corresponding valid GHDs, describing how they can be decomposed into ordinary, unrestricted GHDs. This reduction connects us to a trove of parallel work on GHDs. We finish by extending our work to handle product aggregations.

\paragraph*{Acknowledgements}
We thank Atri Rudra for invaluable insights and feedback developing our approach. CR gratefully acknowledges the support of the Defense Advanced Research Projects Agency (DARPA) XDATA Program under No. FA8750-12-2-0335 and DEFT Program under No. FA8750-13-2-0039, DARPAs MEMEX program under No. FA8750-14-2-0240, the National Science Foundation (NSF) under CAREER Award No. IIS-1353606, Award No. CCF-1356918 and EarthCube Award under No. ACI-1343760, the Office of Naval Research (ONR) under awards No. N000141210041 and No. N000141310129, the Sloan Research Fellowship, the Moore Foundation Data Driven Investigator award, and gifts from American Family Insurance, Google, Lightspeed Ventures, and Toshiba.

\bibliography{ajar-bib_ARXIV}

\appendix
\section{Background}\label{sec:background}
The AggroGHDJoin algorithm is a simple variant of some well-known join algorithms. We describe these algorithms next.

\paragraph*{GenericJoin}
We first describe the AGM bound on the join output size developed by Atserias, Grohe, and Marx~\cite{AGM}.  Given query hypergraph $\mH_Q = (\mV, \mE)$ and relations $\{R_F | F \in \mE\}$, consider the following linear program:
\begin{align*}
\textrm{Minimize } &\sum_{F \in \mE} x_F \log_{\IN}(|R_F|) \\
\forall v \in \mV : &\sum_{F : v \in F} x_F \ge 1 \\
\forall F \in \mE: & x_F \ge 0
\end{align*}

Any feasible solution $\overrightarrow{x}$ is a \emph{fractional edge cover}. Suppose $\rho^*$ is the optimal objective. Then the \emph{AGM bound} on the worst-case output size of join $\Join_{F \in \mE} R_F$ is given by $\IN^{\rho^*} = \prod_{F \in \mE} |R_F|^{x_F^*}$. We will use $\IN^{AGM(Q)}$ to denote the AGM bound on a query $Q$. The GenericJoin (GJ) algorithm~\cite{NRR} computes a join in time $\Ot(\IN^{AGM(Q)})$ for any join query. GJ will be used as a subroutine in a later algorithm, where $GJ(\mH, \{ R_F | F \in \mE_\mH\})$ denotes a call to GenericJoin with one input relation $R_F$ per hyperedge $F$ in hypergraph $\mH$.

\paragraph*{Yannakakis}
Yannakakis' algorithm~\cite{Yannakakis81} operates on $\alpha$-acyclic queries. There are several different equivalent definitions of $\alpha$-acyclicity; we provide the definition that builds a tree out of the relations as it most naturally relates to generalized hypertree decompositions. 
\begin{definition}
Given a hypergraph $\mH = (\mV_\mH, \mE_\mH)$, a \emph{join tree} over $\mH$ is a tree $\mT = (\mV_\mT, \mE_\mT)$ with $\mV_\mT = \mE_\mH$ such that for every attribute $A \in \mV_\mH$, the set $\{v \in \mV_\mT | A \in v\}$ forms a connected subtree in $\mT$.
\end{definition}

A hypergraph $\mH$ is $\alpha$-acyclic if there exists a join tree over $\mH$ ~\cite{ALICE, Yannakakis81}. We can use the classic GYO algorithm to produce a join tree~\cite[ch.6]{ALICE}. The Yannakakis algorithm
takes a join tree as input. It's pseudo-code is given in Section~\ref{subsec:background-1}.

\begin{theorem}
Algorithm \ref{NoAggroY} runs in $O(\IN + \OUT)$ where $\IN$ and $\OUT$ are the sizes of the input and output, respectively.
\end{theorem}

To leverage the speed of Yannakakis for cyclic queries, we look to GHDs~\cite{FHTW,gottlob:ght}. The intuition behind a GHD is to group the attributes into bags (as specified by the function $\chi$) such that we can build a \emph{join tree} over these bags. This allows us to run $GJ$ within each bag and then Yannakakis on the join tree.
The resulting algorithm is GHDJoinwhose pseudo-code is given in Algorithm~\ref{NoAggroGHD}. The runtime of GHDJoin is given by $\Ot(\IN^{fhw(\mT,\mH)} + \OUT)$

\begin{algorithm}[t]
\caption{GHDJoin($\mH = (\mV_\mH, \mE_\mH)$, $(\mT(\mV_\mT, \mE_\mT), \chi)$, $\{R_F | F \in \mE_\mH\}$)}
\label{NoAggroGHD}
\textbf{Input:} Query hypergraph $\mH$, GHD $(\mT, \chi)$, Relations $R_F$ for each $F \in \mE_\mH$
\begin{algorithmic}[1]
\State $S_R \gets \emptyset$
\ForAll{$t \in \mV_\mT$}
\State $\mH_t \gets (\chi(t), \{ \pi_{\chi(t)} F | F \in \mE_\mH \})$
\State $S_R \gets S_R \cup GJ(\mH_t, \{ \pi_{\chi(t)} R_F | F \in \mE_\mH\})$
\EndFor
\State \Return $Yannakakis(\mT, S_R)$
\end{algorithmic}
\end{algorithm}

\begin{theorem}
Algorithm \ref{NoAggroGHD} runs in $\Ot(\IN^{fhw(\mT, \mH)} + \OUT)$.
\end{theorem}

We can make some straightforward modifications to the above join algorithms to perform aggregations. The traditional Yannakakis and GHDJoin algorithms perform the join in a bottom up fashion, after a semijoin phase to ensure that there are no dangling tuples. The modified algorithms above handle aggregations using the same intuition as in traditional query plans: ``push down'' aggregations as far as possible. Since each attribute must occur in a connected subtree of the GHD, we can push its aggregation down to the root of this connected subtree, which is the $TOP$ node of the attribute. There is a standard modification to Yannakakis for project-join queries that projects away attributes at their $TOP$ node~\cite{Yannakakis81}. Instead of projecting, we perform aggregation.

We provide the pseudo-code of AggroYannakakis, which is a simple variant of the well-known Yannakakis~\cite{Yannakakis81} algorithm, in Algorithm~\ref{AggroY}. Algorithm~\ref{Aggro} gives the pseudo-code of AggroGHDJoin, which is a variant of GHDJoin that calls AggroYannakakis instead of Yannakakis. AggroGHDJoin also does some extra work to ensure we pass each annotation to $GJ$ only once. The $\pi^1$ operator in AggroGHDJoin denotes a projection that projects tuples while replacing the annotation by $1$, to ensure that the same annotation isn't counted more than once.

\begin{algorithm}[H]
\caption{AggroYannakakis($\mT = (\mV, \mE$), $\alpha$, $\{R_F | F \in \mV\}$)}
\label{AggroY}
\textbf{Input:} Join tree $\mT = (\mV, \mE)$, Aggregation order $\alpha$, Relations $R_F$ for each $F \in \mV$
\begin{algorithmic}
\ForAll{$F \in \mV$ in some bottom-up order}
\Comment Semi-join reduction up
\State $P \gets$ parent of $F$
\State $R_P \gets R_P \ltimes R_F$
\EndFor
\ForAll{$F \in \mV$ in some top-down order}
\Comment Semi-join reduction down
\State $P \gets$ parent of $F$
\State $R_F \gets R_F \ltimes R_P$
\EndFor
\While{$F \in \mV$ in some bottom-up order}
\Comment Aggregation
\State $\beta \gets \alpha \cap \{a \in \mV | TOP_\mT(a) = F\}$
\State $R' \gets \Sigma_{\beta} R_F$
\If{$F$ is not the root}
\State $P \gets$ parent of $F$
\State $R_P \gets R_P \Join R'$
\Comment Compute the join
\EndIf
\EndWhile
\State \Return $R_R$ for the root $R$
\end{algorithmic}
\end{algorithm}

\begin{algorithm}[H]
\caption{AggroGHDJoin($\mH = (\mV_\mH, \mE_\mH)$, $(\mT(\mV_\mT, \mE_\mT), \chi)$, $\{R_F | F \in \mE_\mH\}$)}
\label{Aggro}
\textbf{Input:} Query hypergraph $\mH$, GHD $(\mT, \chi)$, Relations $R_F$ for each $F \in \mE_\mH$
\begin{algorithmic}
\State $S_R \gets \emptyset$
\ForAll{$t \in \mV_\mT$}
\State $\mH_t \gets (\chi(t), \{ \pi_{\chi(t)} F | F \in \mE_\mH \})$
\State $I \gets \{ R_F | F \subseteq \chi(t), \exists a \in F : TOP_\mT(a) = t\} \cup \{ \pi_{\chi(t)}^1 R_F | F \not\subseteq \chi(t) \text{ or } \forall a \in F : TOP_\mT(a) \neq t\}$
\State $S_R \gets S_R \cup GJ(\mH_t, I)$
\EndFor
\State \Return AggroYannakakis$(\mT, S_R)$
\end{algorithmic}
\end{algorithm}

In the classic analysis of Yannakakis, the runtime of the semi-join portion is bounded by $O(\IN)$ and the bottom-up join is bounded by $O(\OUT)$. In AggroYannakakis, the analysis of the semi-join portion is unchanged, but the aggregation reduces the size of the output, thereby making the $\OUT$ bound harder to achieve. In particular, during the bottom-up join, we may compute an intermediate relation whose attributes are not a subset of the output attributes, meaning that its size may not bounded by $\OUT$. These potentially large intermediate relations are the underlying cause for the traditional $\IN \cdot \OUT$ runtime.

However, using intuition discovered in reference~\cite{OZVLDB13}, if we require the output attributes to appear above non-output attributes, we can preserve the $\IN + \OUT$ runtime.

\begin{theorem}\label{thm:aggro-complexity}
Suppose we have a GHD such that for any output attribute $A$ and non-output attribute $B$, $TOP_\mT(B)$ is not an ancestor of $TOP_\mT(A)$. AggroGHDJoin runs in $\Ot(\IN^{fhw(\mT, \mH)} + \OUT)$ given this GHD.
\end{theorem}
\begin{proof}
$GJ$ on each bag still runs in $\Ot(\IN^{fhw(\mT, \mH)})$. We need to prove the Yannakakis portion runs in $O(\IN+\OUT)$ after running $GJ$.

The semijoin portion runs in $O(IN)$ as in the original Yannakakis algorithm. In the join phase, we have two types of joins. In the first type, $F \setminus \beta \subseteq P$. This implies the join output is a subset of $R_P$ (with different annotations). So the total runtime of this type of join is $O(\IN)$. For the second type, $F \setminus \beta \subsetneq P$. This means some attribute in $(F \setminus \beta) \setminus P$ must be an output attribute, and all attributes in $P$ must be output attributes as well (as their $TOP$ value is an ancestor of $F$). So the result of our join must be a subset of the output table; the total runtime of this type of join is $O(\OUT)$. Thus the total runtime of the algorithm is $O(\IN+\OUT)$.
\end{proof}

Note that while AggroGHDJoin runs in $\Ot(\IN^{fhw(\mT, \mH)} + \OUT)$ time on the GHDs above, it may not necessarily produce the right output unless the GHD satisfies additional conditions, to ensure that aggregations can be done in the proper order. In particular, recall our definition a GHD $(\mT, \chi)$ is compatible with $\alpha$ if for all attribute pairs $A,B$, $TOP_\mT(A)$ being an ancestor of $TOP_\mT(B)$ implies that either $A$ is an output variable or $A$ occurs before $B$ in $\alpha$.

\begin{theorem}\label{thm:aggro-compatible}
If a GHD $(\mT, \chi)$ is compatible with $\alpha$, then AggroGHDJoin given $(\mT, \chi)$ correctly computes $Q_{\mH, \alpha}$.
\end{theorem}
\begin{proof}
We first show that AggroYannakakis works as expected. We note that the semi-join reduction does not change the output; it only quickens the process. We only consider the bottom-up join. For each node $t$ in the join tree, let $R(t)$ be the relation associated with that node before this loop (i.e. after the semi-join portion). Let $R'(t)$ be the final relation associated with node $t$ when we are processing node $t$ (i.e. after the bottom up join with $t$'s descendants is done, and after the aggregation in $t$). Let $\mT_t$ be the subtree that includes $t$ and all of its descendants. Let $s(t)$ be the attributes aggregated at node $t$, i.e. $\alpha \cap \{a \in \mV | TOP_\mT(a) = t\}$, and let $s(\mT_t) = \cup_{t \in \mT_t} s(t)$. For each non-leaf node $t$, let $c(t)$ be the set of $t$'s children.

For each node $t$, we claim $R'(t) = \sum_{\alpha_{s(\mT_t)}} \Join_{t' \in \mT_t} R(t')$. Proof by induction on the tree. For each leaf $l$, $R'(l) = \sum_{\alpha_{s(l)}} R(l)$ by definition.

For a non-leaf node $t$,
\begin{align*}
R'(t) &= \sum_{\alpha_{s(t)}} R(t) \Join (\Join_{t_c \in c(t)} R'(t_c)) \\
&= \sum_{\alpha_{s(t)}} R(t) \Join \left(\Join_{t_c \in c(t)} \sum_{\alpha_{s(\mT_{t_c})}} \Join_{t' \in \mT_{t_c}} R(t')\right)\\
&= \sum_{\alpha_{s(\mT_t)}} \Join_{t' \in \mT_t} R(t')
\end{align*}
The second step is due to the inductive hypothesis. The final step is simply ``pulling out" the aggregations from the sub-orderings one at a time; we can arbitrarily interleave the aggregation orders $\alpha_{s(\mT_{t_c})}$. We can simply interleave them to match $\alpha_{\cup_{t_c \in c(t)} s(\mT_t)}$. Since the original GHD is compatible with $\alpha$, we know the aggregations $\alpha_{s(t)}$ precede $\alpha_{s(\mT_{t_c})}$ in $\alpha$, implying that $\sum_{\alpha_{s(\mT_t)}} \sum_{\alpha_{\cup_{t_c \in c(t)} s(\mT_t)}} = \sum_{\alpha_{s(\mT_t)}}$.  Our output is $R'(t_r)$ where $t_r$ is the root node, which is $\sum_{\alpha} \Join_{t \in \mT} R(t)$ as desired.

Since AggroYannakakis works as expected, we simply need to ensure that the bags are computed appropriately. Note the GHD ensures for every relation $R_F$, there is a node $t$ such that $F \subseteq \chi(t)$. This means that no tuple is lost; computing AggroYannakakis on the bags will compute the correct tuples. To ensure it computes the correct annotations, we need to ensure every annotation appears in the bags at most once; our algorithm places the annotation of a relation $R_F$ in the top-most node that contains all of the attributes $R_F$.
\end{proof}

{\bf Product Aggregations:} When product aggregations are present in an \ajar query $Q_{\mH, \alpha}$, we have a notion of product partition hypergraphs, AGHDs over product partition hypergraphs, and a corresponding notion of AGHDs compatible with an ordering. We now prove theorem~\ref{thm:ajar-product-runtime} that extends theorems~\ref{thm:aggro-complexity} and~\ref{thm:aggro-compatible} to the case where product aggregations are present.

A product partition partition $P = (\mV_P, \mE_P)$ essentially creates multiple renamed copies of each product attribute $a$ ($a_1, a_2,\ldots,a_{|P_a|}$), and assigns one of the renamed copies to each relation containing $a$. An AGHD is essentially a GHD over $P$. Given $P$, and $a \in \mV_\mH$, let $P(a)$ equal $\{a\}$ if $a$ is not a product attribute, and $\{a_1,\ldots,a_{|P_a|}\}$ otherwise. Given $a' \in \mV_P$, let $P^{-1}(a')$ equal $a$ such that $a' \in P(a)$. Given an edge $F \in \mE_P$, let $P^{-1}(F)$ denote the edge $\{P^{-1}(a') \mid a' \in F\}$. We define a modified ordering $\alpha^P$ over $\mV_P$ that takes $\alpha$ and replaces each occurrence of $(a,\otimes)$ with $(a_1,\otimes)$,$(a_2,\otimes)$,$\ldots$,$(a_{|P_a|},\otimes)$ for each product attribute $a$. For any $F \in \mE_P$, we define the relation $R_F$ to be same as the $R_{P^{-1}(F)}$ (but with the attribute name changed. This gives us the modified \ajar query $Q^P_{P, \alpha^P} = \sum_{\alpha^P} \Join_{F \in \mE_P} R_F$. Then we have,

\begin{lemma}\label{lemma:product-renaming-one}
Suppose $R'(A,C_1)$ is a copy of $R(A,C)$ with $C$ renamed to $C_1$, and $S'(B,C_2)$ is a copy of $S(B,C)$ with $C$ renamed to $C_2$. Then
$$\sum_{(C_1,\otimes)}\sum_{(C_2,\otimes)} R'(A,C_1) \Join S'(B,C_2) = \sum_{(C,\otimes)} R(A,C) \Join S(B,C)$$
\end{lemma}
\begin{proof}
Suppose the annotations for the $C$ values in $R$ are $n_1,n_2,\ldots,n_k$ and in $S$ are $m_1,m_2,\ldots,m_k$ (assume all annotations are present i.e. absent tuples have a zero-annotation). Then the RHS is $\otimes_{i=1}^{k}n_im_i$. The LHS will have $\otimes_{i=1}^{k}n_i \otimes_{j=1}^{k}m_j$. The RHS is equal to the LHS because of idempotence of $\otimes$. Note that if $\otimes$ wasn't idempotent, the LHS would have the $m_j$ terms multiplied $k$ times while the RHS has them once.
\end{proof}

\begin{lemma}\label{lemma:product-renaming-ajar}
For each database instance $I$, $Q_{\mH,\alpha}(I) = Q^P_{P,\alpha^P}(I)$.
\end{lemma}
This lemma can be proved by repeated application of Lemma~\ref{lemma:product-renaming-one}. 

Now we can easily prove Theorem~\ref{thm:ajar-product-runtime}. Suppose we have a AGHD $D = (\mT, \chi, P)$ which is compatible with an ordering $\alpha$. Then the GHD $(\mT, \chi)$ over hypergraph $P$, is compatible with $\alpha^P$. Running AggroGHDJoin over this GHD, with ordering $\alpha^P$ correctly computes $Q^P_{P,\alpha^P}(I)$, due to theorem~\ref{thm:aggro-compatible}. And by Lemma~\ref{lemma:product-renaming-ajar}, this also equals $Q_{\mH,\alpha}(I)$, which is the output we want. Also, since the AGHD is compatible with $\alpha$, the GHD must satisfy the condition of Theorem~\ref{thm:aggro-complexity}, and hence AggroGHDJoin runs on it in time $\Ot(\IN^{fhw} + \OUT)$.

\section{Comparison with Related Work}
\label{sec:app-related}
\subsection{Section \ref{sec:prelim}}
\label{subsec:faq-comparison-proof}

In Section \ref{sec:prelim}, we define a simple approach to solving $\ajar$ queries, and we claim in Theorem~\ref{thm:ajar-vs-faq-runtime} that our runtime guarantee of $\Ot(IN^{w^*} + \OUT) \le \Ot(IN^{faqw})$. We note that the $faqw$ exponent is actually the optimum value of $faqw(\sigma)$ over the equivalent orderings $\sigma$ they consider (we discuss the space of orderings they consider in the next subsection). Our approach will recognize $\sigma$ as being equivalent, and will search for the best compatible GHD for $\sigma$. We will show that there exists a compatible AGHD $(\mT, \chi, P)$ for every equivalent ordering $\sigma$ such that $fhw(\mT, \mH) = faqw(\sigma)$ (as Example~\ref{ex:faqoutput} shows, the compatible AGHD $\mT$ we construct may not be the optimal compatible GHD).

We start by briefly summarizing FAQ's algorithm, with the pseudo-code (written in the notation of this paper) given in Algorithm~\ref{algo:InsideOut}. Let $\sigma$ be the ordering used for aggregation. Let $n$ denote the total number of attributes $|\mV_\mH|$ and $f$ denote the number of output attributes (thus $|\sigma| = n - f$). For notational convenience, we will be using $\sigma[i]$ to denote both the attribute and the operator that make up the $i^{th}$ operator-attribute pair in the ordering. 

\begin{algorithm}[H]
\caption{InsideOut($\mH = (\mV_\mH, \mE_\mH)$, $\sigma$, $\{R_F | F \in \mE_\mH\}$)}
\label{algo:InsideOut}
\textbf{Input:} Hypergraph $\mH = (\mV_\mH, \mE_\mH)$, Aggregation ordering $\sigma$, Relations $R_F$ for each $F \in \mE_\mH$
\begin{algorithmic}
\State $E_n \gets \{R_F \mid F \in \mE_\mH\}$
\For{$(k = n; k > f; k--)$}
\State $\delta(k) \gets \{R_F \in E_k \mid \sigma[k-f] \in F\}$
\If{$\sigma[k-f]$ is not a product aggregation}
\State $U_k \gets \Join_{R \in \delta(k)} R$
\State $E_{k-1} \gets (E_k \setminus \delta(k)) \cup \{\sum_{\sigma[k-f]} U_k\}$
\Else 
\State $E_{k-1} \gets (E_k \setminus \delta(k)) \cup \{\sum_{\sigma[k-f]} R \mid R \in \delta(k) \}$ 
\EndIf
\EndFor
\State \Return $\Join_{R \in E_f} R$
\end{algorithmic}
\end{algorithm}

FAQ relies on a worst-case optimal algorithm to compute each of the joins, implying that in the $\Ot(IN^{faqw})$ runtime guarantee, faqw is defined as the maximum AGM bound placed on each of the computed joins. Define $p^*_H : 2^{\mV_\mH} \to \mathcal{R}$ to be a function that maps a subset of the attributes to the AGM bound on the subset (i.e. the optimal value of the canonical linear program). Then $faqw = \max(\max_k p^*_H(U_k), p^*_H(V(-\sigma)))$~\cite{FAQ}.

We will build up the compatible AGHD $(\mT, \chi, P)$ in rounds corresponding to each of the $k$ values of $InsideOut$. We first describe how to construct $(\mT, \chi)$, and later describe how to obtain $P$. At the start of round corresponding to a particular $k$, we will have a forest of AGHDs, each of which will have a root mapped (by $\chi$) to the attribute sets of $E_k$, and at the end of each round, the forest's roots will be mapped to the relations of $E_{k-1}$.

For an attribute set $F$, let $t(F)$ represent the node such that $\chi(t(F)) = F$. We start by creating the $|\mE_\mH|$ nodes $\{t(F) | F \in E_n\}$, which are simply nodes mapped to the input relations. Then for each $k$ from $n$ to $f+1$, let $T$ represent the set of nodes $\{t(F) | F \in \delta(n)\}$; these are the nodes that will be processed (i.e. the nodes for whom we will create parents). If $\sigma[k-f]$ is not a product aggregation, we create a node $t(U_k)$ and set $parent(t) = t(U_k)$ for all $t \in T$. We then create a node $t(U_k \backslash \{\sigma[k-f]\})$ and set it to be $parent(t(U_k))$. Note that this process has transformed the set of the forest's roots by removing $T$ and adding $t(U_k \backslash \{\sigma[k-f]\})$, mirroring the transformation between $E_k$ and $E_{k-1}$. If $\sigma[k-f]$ is a product aggregation, then for each $F \in \delta(n)$, we create a node $t(F \backslash \{\sigma[k-f]\})$ and set it to be $parent(t(F))$; in this case as well the set of the forest's roots match $E_{k-1}$.

At the end of this process, we will have a forest of AGHDs whose roots map to the relations in $E_f$. To conclude our construction, we simply construct the node $t(V(-\sigma))$ and set it to be $parent(t(F))$ for all $F \in E_f$. If there are no product aggregations, then $(\mT, \chi)$ forms a GHD.

$(\mT, \chi)$ satisfies the running intersection property for all non-product attributes, but a product attribute $a$ can be present in multiple disconnected parts of $\mT$. We now describe a product partition $P$ such that $(\mT,\chi,P)$ forms an AGHD for the ordering $\sigma$. Let $P_a$ denote the number of distinct connected components of $\mT$ in which $a$ is present. Then we create $P_a$ copies of $a$ ($a_1$, $a_2$,$\ldots$,$a_{|P_a|}$), and assign a copy to each component in some order. For each $F \in \mE_\mH$ that contains $a$, if the component that $t(F)$ belongs to is assigned $a_i$, then $P$ assigns $a_i$ to $F$. Then $(\mT, \chi, P)$ is an AGHD for $\sigma$.

The AGHD $(\mT, \chi, P)$ as described is trivially compatible with $\sigma$ since we construct $parent(TOP_\mT(\sigma[k-f]))$ explicitly in round $k$; this ensures that $TOP_\mT(\sigma[i])$ cannot be an ancestor of $TOP_\mT(\sigma[j])$ if $i > j$. 

\begin{lemma}
Define $p^*_\mH$ to be a function that maps a set of attributes to the AGM bound on the set (the optimal value of the canonical linear program). The AGHD $(\mT, \chi, P)$ constructed as described satisfies $$fhw(\mT, \mH) = \max(p^*_\mH(V(-\alpha)), \max_{k}{p^*_\mH(U_k)}) = faqw(\sigma).$$
\end{lemma}
\begin{proof}
The nodes in our tree that do not map to $V(-\alpha)$ or the $U_k$ either map to an input relation or to a relation created by aggregating an attribute from a single child node. In the former case, $p^*_\mH$ would evaluate to $1$, so we can ignore them in our maximum. In the latter case, the attributes are a strict subset of its child's attributes, implying we can ignore them too. As such, the fractional hypertree width is simply the maximum fractional cover over $V(-\alpha)$ and the $U_k$. This shows the first part of the equality.
 
The second part of the equality is the definition of $faqw$~\cite{FAQ}.
\end{proof}

Theorem~\ref{thm:ajar-vs-faq-runtime}, as well as its analogue for product aggregations, follow as a simple corollary. We now show an example where the runtime of InsideOut is much worse than the runtime of our Algorithm, primarily due to the fact that it is not output-sensitive.

\begin{example}\label{ex:faqoutput}
Let $n$ be an even number, and consider an \ajar query $Q_{\mH, \alpha}$ where $\mH = (\{A_i \mid 1 \leq i \leq n\}, \{\{A_i, A_{i+1}\} \mid 1 \leq i \leq n-1 \} \cup \{\{A_n,A_1\}\})$, and $\alpha$ is empty (i.e. the query is just a join). Also let each attribute take values $1, 2, 3,\ldots 2\times\lfloor \sqrt{N} \rfloor$. Suppose each relation $\{A_i,A_{i+1}\}$ for $1 \leq i \leq n-1$ connects values of the same parity, while relation $\{A_n,A_1\}$ connectes values of opposite parities. Thus each relation has size $N$, and $\IN = O(N)$ ($n$ is a constant), and the join output is empty. There is a GHD with bags $\{A_1,A_2,A_3\}, \{A_1,A_3,A_4\}, \ldots \{A_1, A_{n-1},A_n\}$ that is compatible with the empty ordering. The fhw of this GHD is $2$, so we have $w^* = 2$. Thus the runtime of our algorithm will be $\Ot(\IN^2)$. InsideOut will compute an intermediate output consisting of the join of $n-1$ of the relations, which has size $N^{(n-1)/2}$ , so InsideOut's runtime will be at least $\Ot(\IN^{(n-1)/2})$.
\end{example}

FAQ does discuss, at a very high-level and without proofs, changes to InsideOut that will allow their runtime to be output-sensitive~\cite[Section 10.2]{FAQ}. Their most general and useful change involves building a GHD for the output variables and running a message passing algorithm between the bags, which exactly describes GHDJoin. Implementing this change would make InsideOut completely equivalent to AggroGHDJoin. We note that the FAQ paper frames these changes as decisions in how to \emph{represent} the output, whereas we present the optimization in an algorithmic context, independent of any other storage optimizations.

\subsection{Section \ref{sec:equivalent-orderings}}

In Section \ref{sec:equivalent-orderings}, we define a partial order $<_{\mH, \alpha}$ that exactly characterizes the constraints an aggregation ordering must satisfy to be equivalent to a given ordering $\alpha$. Our partial ordering is complete, which is a result that FAQ cannot match. Much like our approach, FAQ actually defines their own partial ordering, which we denote $<_{FAQ}$, and their work only considers orderings that are linear extensions of $<_{FAQ}$. However, we will show an example where $<_{FAQ}$ has unnecessary constraints:

\begin{example}\label{example:faq-incompleteness}
Consider the $\ajar$ query given by \main{\\}$\sum_A \max_B \sum_C R(A,B) S(A,C)$. By our characterization, $A <_{\mH, \alpha} B$ is the only constraint, giving rise to $3$ different valid orderings. The FAQ characterization, however, has two constraints: $A <_{FAQ} B$ and $C <_{FAQ} B$, which only allows for $2$ different valid orderings. Note that FAQ constraints \emph{preclude the original ordering} ABC.
\end{example}

\subsection{Section \ref{sec:decomposing}}
In Section \ref{sec:decomposing}, we define a decomposition that relates the width of a valid GHD to the widths of a series of ordinary GHDs. Variable orderings (as used by FAQ) are not as readily suited as GHDs are for decompositions. FAQ does derive their own version of a decomposition, but the difficulties that arise when using variable orderings are exemplified in the way FAQ switches between GHDs and variable orderings in their proofs~\cite{FAQ}. In addition, the FAQ decomposition is demonstrably weaker than ours; their decomposition incurs some overhead costs when combining the sub-orderings to build the overall ordering, precluding a result like Corollary~\ref{cor:optimal-construct} that provides the groundwork for the variety of extensions we provide. To exemplify the gap in the two decompositions, we inspect a specific $\ajar$ query:

\begin{example}\label{example:faq-decomposition}
Consider the query \main{\\} $\sum_B \sum_C \sum_D R(A,B) S(B,C) T(C,D), U(D,A)$. Suppose $|A| = \sqrt{N}$, $|B| = 2 = |D|$, $|C| = N$, and all of the pairwise relations are constructed as complete cross products of the attributes' values. Our decomposition will result in the chain GHD $A-ABD-BCD$, while the FAQ decomposition will result in the GHD $A-ABC-ACD$. The runtimes of both FAQ and GHDJoin using the former GHD is $\Ot(N)$, whereas the runtimes using the latter GHD are $\Ot(N^{3/2})$. As such, the FAQ decomposition will perform asymptotically worse than our decomposition.

More generally, consider a query $Q_n$ with relations $R_i(A_i,B)$, $S_i(B,C_i)$, $T_i(C_i,D)$, $U_i(D,A_i)$ for $1 \leq i \leq n$. Like before, all $|A_i|$'s are $\sqrt{N}$, all $|C_i|$'s are $N$, and $|B| = |D| = 2$. And the aggregation ordering only has the $+$ operator, on $B$, $D$ and all $C_i$'s. Our decomposition gives the chain $A_1\ldots A_n$ $-$ $A_1\ldots A_nBD$ $-$ $C_1\ldots C_n BD$. This results in a runtime of $\Ot(N^n)$. FAQ's decomposition gives $A_1\ldots A_nBC_1\ldots C_n$ $-$ $A_1\ldots A_n D C_1\ldots C_n$. This decomposition, and its corresponding ordering, give a runtime of $\Ot(N^{3n/2})$. Thus the difference between runtime exponents caused by FAQ's decomposition and our decomposition can be arbitrarily high.
\end{example}

\subsection{Section~\ref{subsec:GYM}}
\begin{example}\label{example:Log-GTA-invalidation}
Suppose we a \ajar query $Q_{\mH, \alpha}$ with 
$$\mH = (\{A,B,C,D,E,F\}, \{\{A,B\}, \{B,C\}, \{B,D,E\}, \{D,F\}\})$$
and $\alpha = ((D, \sum), (E, \sum), (F, \sum))$. We start with the width-$1$ valid GHD $(\mT, \chi)$ with $\mV-\mT = \{v_1,v_2,v_3,v_4\}$  and $$\mE_\mT = \{v_1,v_2\}, \{v_2,v_3\}, \{v_3,v_4\}$$ such that $v_1$ is the root, and $\chi(v_1) = \{A,B\}$, $\chi(v_2) = \{B,C\}$, $\chi(v_3) = \{B,D,E\}$, $\chi(v_4) = \{D,F\}$. 

Applying Log-GTA gives us a shorter GHD $(\mT', \chi')$ with $\mV_\mT' = \{u, v_1, v_2, v_3, v_4\}$, $$\mE_\mT' = \{\{v_1,u\}, \{u,v_2\}, \{u,v_3\}, \{u,v_4\} \}$$ with $v_1$ as the root. $\chi'(u) = \{B,D\}$ and $\chi'(v_i) = \chi(v_i)$ for all $i$. Now $TOP(D) = u$ which is an ancestor of $TOP(C) = v_2$, despite $C$ being an output attribute and $D$ not being an output attribute. This means GHD $(\mT', \chi')$ is invalid, showing that applying Log-GTA to a valid GHD may make it invalid. 
\end{example}

As the above example shows, we cannot directly apply Log-GTA to a valid GHD to get a shorter valid GHD. 

\section{Characterizing Equivalent Orderings: Proofs}\label{sec:app-equiv}

We now formally present our partial order $<_{\mH,\alpha}$ that characterizes the interaction of the two forms of commuting. As we said in Section~\ref{sec:equivalent-orderings}, we have two relations $\pre$ and $\dnc$, that are mutually recursive. We initialize the constraints to a base case and iteratively update them till we reach a fixed point. We now formalize this. We use binary operator $<^i_{\mH,\alpha}$ to denote the constraint $\pre$ after $i$ iterations, and operator $\sim^i_{\mH,\alpha}$ to denote $\dnc$ after $i$ iterations, with one difference; both operators behave slightly differently for output attributes. To readily incorporate output attributes into the constraints, we define an augmented aggregation ordering below:

\begin{definition}
For any aggregation ordering $\alpha$, let $F$ be the set of output variables. Then define $\alpha^O = \alpha^O_1, \alpha^O_2, \dots, \alpha^O_n$ to be a sequence such that $\alpha^O_i = (F_i, \texttt{NULL})$ for $1 \le i \le |F|$ and $\alpha^O_i = \alpha_{i+|F|}$ for $|F|+1 \le i \le n$.
\end{definition}

Note that $n$ is defined to be the number of attributes in the query. Now we can formally define $<^i_{\mH,\alpha}$ and $\sim^i_{\mH,\alpha}$. Both of these binary operators operator over attribute-operator pairs, but since each attribute occurs at most once in an ordering, we can equivalently think of them as operating over attributes. We use these two interchangeably e.g. $A <_{\mH, \alpha} B$ denotes the same thing as $(A, \oplus) <_{\mH, \alpha} (B, \oplus')$.

\begin{definition}
For a given query $Q_{\mH, \alpha}$ with $\mH = (\mV, \mE)$, we define relations $\sim_{\mH, \alpha}^i$ and partial orders $<_{\mH, \alpha}^i$ over attribute-operator pairs in $\alpha^O$. For any $A,B \in \mV$, suppose $(A, \oplus), (B, \oplus') \in \alpha^O$. Then, for $i = 0$, $(A, \oplus) \sim_{\mH, \alpha}^0 (B, \oplus')$ if and only if one of the following is true:
\begin{itemize}
\item $\oplus \neq \oplus'$ and $\exists E \in \mE: A,B \in E.$ \hfill $(0.1)$
\item $\oplus \neq \oplus'$ and either $\oplus = \texttt{NULL}$ or $\oplus' = \texttt{NULL}$ \hfill $(0.2)$
\end{itemize}

For $i > 0$, $(A, \oplus) \sim^i_{\mH, \alpha} (B, \oplus')$ if and only if $(A, \oplus) \not\sim^j_{\mH, \alpha} (B, \oplus')$ for all $j < i$ and one of the following is true:
\begin{itemize}
\item $\oplus \neq \oplus'$ and $\exists E \in \mE, (C, \oplus'') \in \alpha^O: B,C \in E, (A, \oplus) <_{\mH, \alpha}^{i-1} (C, \oplus'')$ \hfill $(i.1)$
\item $\exists (C, \oplus'') \in \alpha^O$ and $j,k < i: (A, \oplus) <_{\mH, \alpha}^j (C, \oplus'') <_{\mH, \alpha}^k (B, \oplus')$ \hfill $(i.2)$
\end{itemize}
For any $i \ge 0$, $(A, \oplus) <_{\mH, \alpha}^i (B, \oplus')$ if and only if $(A, \oplus) \sim_{\mH, \alpha}^i (B, \oplus')$ and $(A, \oplus)$ precedes $(B, \oplus')$ in $\alpha^O$.

Finally, $(A, \oplus) \sim_{\mH, \alpha} (B, \oplus')$ if and only if $(A, \oplus) \sim_{\mH, \alpha}^i (B, \oplus')$ for some $i \ge 0$. Similarly, $(A, \oplus) <_{\mH, \alpha} (B, \oplus')$ if and only if $(A, \oplus) <_{\mH, \alpha}^i (B, \oplus')$ for some $i \ge 0$.
\end{definition}

The core of our definition is the four labeled conditions for $\sim$. The condition $0.1$ represents the simplest structure that violates both conditions of Theorem~\ref{commute}; it represents our base case. Condition $0.2$ simply ensures the output attributes precede non-output attributes. Our condition $i.1$ extends the structure from $0.1$ beyond single relations. If $A < C$ and $C$ appears in a relation with $B$, we can guarantee that $A$ and $B$ cannot be separated in the way the second condition of Theorem~\ref{commute} requires, and if $\oplus \neq \oplus'$, the first condition is violated as well. Condition $i.2$ simply ensures that transitivity interacts properly with condition $i.1$. 

We now prove the two lemmas stated in Section~\ref{sec:equivalent-orderings}, followed by proving soundness and completeness of $<_{\mH,\alpha}$.
\begin{lemma}[Copy of Lemma~\ref{Pathing}]
Suppose we are given a hypergraph $\mH = (\mV, \mE)$ and an aggregation ordering $\alpha$. Fix two arbitrary attributes $A,B \in \mV$ such that $(A, \oplus), (B, \oplus') \in \alpha^O$ for differing operators $\oplus \neq \oplus'$. Then, for any path $P$ in $\mH$ between $A$ and $B$, there must exist some attribute in the path $C \in P$ such that $C <_{\mH, \alpha} A$ or $C <_{\mH, \alpha} B$. 
\end{lemma}
\begin{proof}
We use induction on the length of path $P$. \\
{\em Base Case:} Let $|P| = 2$. This implies that there exists some edge $E \in \mE$ such that $A,B \in E$. Thus $A \sim_{\mH, \alpha}^0 B$. Then, by definition, either $A <_{\mH, \alpha} B$ or $B <_{\mH, \alpha} A$ depending on which attribute appears first in $\alpha$.

{\em Induction:} Suppose $|P| = N > 2$ and assume the lemma is true for paths of length $< N$. We call this assumption the {\em outer inductive hypothesis}, for reasons that will become apparent later. Path $P$ can be rewritten as $P = AP'B$ where $P'$ is a path of length at least $1$. Let $C$ be the node in $P'$ that appears earliest in $\alpha^O$; this implies that there exists no attribute in our path $D \in P'$ such that $D <_{\mH, \alpha} C$. Define an operator $\oplus''$ such that $(C, \oplus'') \in \alpha^O$. Since $\oplus \neq \oplus'$, either $\oplus \neq \oplus''$ or $\oplus' \neq \oplus''$. Without loss of generality, assume that $\oplus \neq \oplus''$.

Consider the subpath of $P$ from $A$ to $C$. It is shorter than $N$ and connects two attributes with different operators. We apply our inductive hypothesis to get that there exists some $D \in P$ such that either $D <_{\mH, \alpha} A$ or $D <_{\mH, \alpha} C$. In the first case, we have found an attribute that satisfies our conditions and we are done. In the second case, we know that $D \notin P'$ by our definition of $C$. Thus $D$ must be $A$; we have that $A <_{\mH, \alpha} C$.

Consider the subpath of $P$ from $C$ to $B$; let $X_i$ denote the $i^{th}$ node in this path for $0 \le i \le k$, where $X_0 = C$ and $X_k = B$. We claim that for all $i < k$, $A <_{\mH, \alpha} X_i$. We argue this inductively; for our base case, we are given that $A <_{\mH, \alpha} C = X_0$. Now let $i \ge 1$, and assume $A <_{\mH, \alpha} X_{j}$ for $j < i$. Call this the {\em inner inductive hypothesis}.

Note that we have $A<_{\mH, \alpha} C$ and that $C$ must precede $X_i$ by definition. Thus $A$ precedes $X_i$ in $\alpha^O$. All that remains is showing that $A \sim_{\mH, \alpha} X_i$. Define $\oplus^i$ such that $(X_i, \oplus^i) \in \alpha^O$. Since we assumed earlier that $\oplus \neq \oplus''$, we know that either $\oplus^i \neq \oplus$ or $\oplus^i \neq \oplus''$. 
\begin{itemize}
\item $\oplus^i \neq \oplus$ \\
By our (inner) inductive hypothesis, we know that $A <_{\mH, \alpha} X_{i-1}$. We also know that there must exist some edge $E \in \mE$ such that $X_{i-1}, X_i \in E$. Thus by condition $i.1$, $A \sim_{\mH, \alpha} X_i$.
\item $\oplus^i \neq \oplus''$ \\
By our (outer) inductive hypothesis, we know that for some $0 \le j \le i$, $X_j <_{\mH, \alpha} C$ or $X_j <_{\mH, \alpha} X_i$. By our definition of $C$, the first case is impossible. And by our (inner) inductive hypothesis, we have that $A <_{\mH, \alpha} X_j$. We thus have that $A <_{\mH, \alpha} X_j <_{\mH, \alpha} X_i$, which implies that $A \sim_{\mH, \alpha} X_i$ by condition $i.2$.
\end{itemize}
This gives us that $A <_{\mH, \alpha} X_{k-1}$. Since there exists an edge $E \in \mE$ such that $X_{k-1}, B \in E$, condition $i.1$ tells us that $A \sim_{\mH, \alpha} B$. As before, this implies that either $A <_{\mH, \alpha} B$ or $B <_{\mH, \alpha} A$. 
\end{proof}

\begin{lemma}[Copy of Lemma~\ref{Pathing2}]
Given a hypergraph $\mH = (\mV, \mE)$ and an aggregation ordering $\alpha$, suppose we have two attributes $A, B \in V(\alpha)$ such that $A <_{\mH, \alpha} B$. Then there must exist a path $P$ from $A$ to $B$ such that for every $C \in P, C \neq A$ we have $A <_{\mH, \alpha} C$.
\end{lemma}
\begin{proof}
Define $\oplus$ and $\oplus'$ such that $(A, \oplus), (B, \oplus') \in \alpha$. In addition define $i$ such that $A <_{\mH, \alpha}^i B$, which implies that $A$ precedes $B$ in $\alpha$ and that $A \sim_{\mH, \alpha}^i B$. Our proof is by induction on $i$. For our basecase, if $A \sim_{\mH, \alpha}^0 B$, we know that $\exists E \in \mE: A, B \in E$. Thus the path $P = AB$ satisfies our conditions.

For $i > 0$, we have the following cases:
\begin{itemize}
\item $\oplus \neq \oplus'$ and $\exists E \in \mE, (C, \oplus'') \in \alpha^O: B,C \in E, A <_{\mH, \alpha}^{i-1} C$\\
By our inductive hypothesis, there must exist a path $P'$ from $A$ to $C$ such that for all $D \in P', D \neq A$ we have $A <_{\mH, \alpha} D$. Then the path $P = P'B$ satisfies our conditions.
\item $\oplus \neq \oplus'$ and $\exists E \in \mE, (C, \oplus'') \in \alpha^O: A,C \in E, B<_{\mH, \alpha}^{i-1} C$ \\
By our inductive hypothesis, there must exist a path $P'$ from $B$ to $C$ such that for all $D \in P', D \neq B$ we have $B <_{\mH, \alpha} D$, which also implies that $A <_{\mH, \alpha} D$ by $i.2$. Let $\overline{P'}$ be the reverse of $P'$. Then the path $P = A\overline{P'}$ satisfies our condition.
\item $\exists C \in \mV$ and $j,k < i: A <_{\mH, \alpha}^j C <_{\mH, \alpha}^k B$ \\
By our inductive hypothesis, there must exist two paths $P'$ and $P''$. $P'$ is a path from $A$ to $C$ such that for all $D \in P', D \neq A$ we have $A <_{\mH, \alpha} D$. Similarly, $P''$ is a path from $C$ to $B$ such that for all $D \in P'', D \neq C$ we have $C <_{\mH, \alpha} D$, which implies $A <_{\mH, \alpha} D$. Thus the path $P = P'P''$ satisfies our conditions.
\end{itemize}
\end{proof}

\begin{theorem}[Copy of Theorem~\ref{thm:dag-soundness-completeness}]
Suppose we are given a hypgergraph $\mH = (\mV, \mE)$ and aggregation orderings $\alpha, \beta$. Then $\alpha \equiv_\mH \beta$ if and only if $\beta$ is a linear extension of $<_{\mH,\alpha}$.
\end{theorem}
\begin{proof}
{\bf Soundness:} 

We use induction on the number of inversion in $\beta$ with respect to the ordering $\alpha$. 
{\em Base Case:} $0$ inversions. Then $\beta$ is identical to $\alpha$ and $\alpha \equiv_\mH \beta$. 

{\em Induction:} Suppose $\beta$ has $N > 0$ inversions, and assume the lemma is true for orderings with $<N$ inversions. There must be some $\beta_i$ and $\beta_{i+1}$ that are inverted with respect to $\alpha$. Consider the ordering $\beta'$ derived by swapping $\beta_i$ and $\beta_{i+1}$. It has $N-1$ inversions with respect to $\alpha$ and is clearly a linear extension of $<_{\mH, \alpha}$. Thus, by the inductive hypothesis, $\alpha \equiv_\mH \beta'$.

We now show that $\beta \equiv_\mH \beta'$. Suppose $\beta_i = (A, \oplus)$ and $\beta_{i+1} = (B, \oplus')$. We have two cases to consider.
\begin{itemize}
\item $\oplus = \oplus'$ \\
By Theorem \ref{commute}, we can swap $\beta_i$ and $\beta_{i+1}$ without affecting the output. This implies that $\beta \equiv_\mH \beta'$.
\item $\oplus \neq \oplus'$ \\
By Lemma \ref{Pathing} and since we know $A$ and $B$ are incomparable under $<_{\mH, \alpha}$, any path between $A$ and $B$ must go through some attribute $C$ such that $C<_{\mH, \alpha} A$ or $C<_{\mH, \alpha} B$. Since $\beta$ is a valid linear extension of $<_{\mH,\alpha}$, these attributes $C$ appear earlier than index $i$ in $\beta$. This implies that $A$ and $B$ are in separate connected components in $\pi_{V(\beta_i, \beta_{i+1}, \dots, \beta_{|\beta|})} \mH$, which implies that we can swap $\beta_i$ and $\beta_{i+1}$ without affecting the output by Theorem \ref{commute}. This implies that $\beta \equiv_\mH \beta'$.
\end{itemize}

{\bf Completeness:}

We prove the contrapositive: we assume that we are given aggregation orderings $\alpha, \beta$ such that $\beta$ is not a linear extension of $<_{\mH, \alpha}$, and we will show that $\alpha \not \equiv_{\mH} \beta$. We will do so by constructing an instance $\hat{I}$ such that $Q_{\mH, \alpha}(\hat{I}) \neq Q_{\mH, \beta}(\hat{I})$.

We assume without loss of generality that $\mV = V(\alpha)$, i.e. that there are no output attributes. We will provide an example where $\beta$ and $\alpha$ must differ in the single annotation that comprises the output. If there are output attributes, we can augment our example by putting $1$s in all the output attributes; our output will be composed of a single tuple composed of all $1$s with the same annotation as in our example below.

Consider the set of all valid linear extensions of $<_{\mH, \alpha}$. Suppose the maximum length prefix identical to the prefix of $\beta$ is of size $k$. Among all linear extensions with maximum length identical prefixes, suppose the minimum possible index for $\beta_{k+1}$ is $k'$. Consider a linear extension $\alpha'$ such that $\alpha'_i = \beta_i$ for $i \le k$ and $\alpha'_{k'} = \beta_{k+1}$. By the soundness part of our proof, $\alpha' \equiv_\mH \alpha$; to show that $\alpha \not\equiv_\mH \beta$ we can simply show that $\alpha' \not\equiv_\mH \beta$. 

Suppose $\alpha'_{k'} = (A, \oplus) = \beta_{k+1}$ and $\alpha'_{k'-1} = (B, \oplus')$. We know that $B <_{\mH, \alpha} A$ since $k'$ is the minimum possible index for $\beta_{k+1}$ in any linear extension of $<_{\mH, \alpha}$. Also, since $B$ and $A$ are adjacent in $\alpha'$, we know that there cannot exist any $C$ such that $B <_{\mH, \alpha} C <_{\mH, \alpha} A$. These two facts combine to imply $\oplus \neq \oplus'$. Then, by Lemma \ref{Pathing2}, there exists a path $P$ from $A$ to $B$ such that every attribute in our path $C \in P$ other than $A$ and $B$ must appear after index $k'$ in $\alpha'$.

Since $\oplus \neq \oplus'$, there must exist $x,y \in \textbf{D}$ such that $x \oplus y \neq x \oplus' y$. Define a relation $\widehat{R_\mV}$ with two tuples. The first tuple will contain a $1$ for each attribute and an annotation $x$. The second tuple will contain a $2$ for each attribute in $P$ (including $A$ and $B$) and a $1$ for every other attribute. The second tuple will be annotated with $y$. Note that among the attributes in $P$, $(A, \oplus)$ is the outermost aggregation in $\beta$ and $(B, \oplus')$ is the outermost aggregation in $\alpha$. This implies that
\[\Sigma_{\alpha'_1} \Sigma_{\alpha'_2} \cdots \Sigma_{\alpha'_{n'}} \widehat{R_\mV} = x \oplus' y
\hspace{15mm}
\Sigma_{\beta_1} \Sigma_{\beta_2} \cdots \Sigma_{\beta_{|\beta|}}
\widehat{R_\mV} = x \oplus y \]

Let $C$ be the attribute in $P$ right before $B$; by definition there must exist an edge $E \in \mE$ such that $B,C \in E$. Consider the following instance over the schema $\mH$:
\[ \hat{I} = \{ \pi_{E}\widehat{R_\mV} \} \cup \{ \pi_F^1 \widehat{R_\mV} | F \in \mE, F \neq E\} .\]

By definition, $\Join_{R_F \in \hat{I}} R_F = \widehat{R_\mV}$. Since we know that $x \oplus y \neq x \oplus' y$, we have that
\[\Sigma_{\alpha'_1} \Sigma_{\alpha'_2} \cdots \Sigma_{\alpha'_{n'}} \Join_{R_F \in \hat{I}} R_F
\neq
\Sigma_{\beta_1} \Sigma_{\beta_2} \cdots \Sigma_{\beta_{|\beta|}} \Join_{R_F \in \hat{I}} R_F \]

We thus have that $Q_{\mH, \alpha'}(\hat{I}) \neq Q_{\mH, \beta}(\hat{I})$, which implies that $\alpha' \not\equiv_\mH \beta$.
\end{proof}

\section{Decomposing Valid GHDs: Proofs}\label{sec:app-decomp}
We start by stating and proving a useful lemma about the aggregation orderings seen in the sub-trees of a decomposable GHD.

\begin{lemma}
\label{lemma:sub-orders}
Given a hypergraph $\mH = (\mV, \mE)$ and an aggregation ordering $\alpha$, suppose $C$ is a connected component of $\mH \backslash V(-\alpha)$. Define $\mH_C = (\bigcup_{E \in \mE_C} E, \mE_C)$\footnote{Recall $\mE_C = \{E \in \mE | E \cap C \neq \emptyset\}$.}. For any $A \in C, B \in \mV$, if $A <_{\mH, \alpha} B$ then $A <_{\mH_C, \alpha_{C \backslash C^O}} B$. Similarly, if $A <_{\mH_C, \alpha_C \backslash C^O} B$, either $A \in C^O$ or $A <_{\mH, \alpha} B$. 
\end{lemma}
\begin{proof}
First we show that for any $A \in C$, if $A <_{\mH, \alpha} B$ then $B \in C \backslash C^O$. If $B \in C^O$, then by definition, $A \not <_{\mH, \alpha} B$. If $B \notin C$, then every path between $A$ and $C$ must go through attributes in $V(-\alpha)$. Thus, by the contrapositive of Lemma~\ref{Pathing2}, $A \not <_{\mH, \alpha} B$. This implies that $A <_{\mH, \alpha} B$ only for $B \in C \backslash C^O$.

Note that for any $A \in C^O$, since $A$ is an output attribute in $\alpha_{C \backslash C^O}$, $A <_{\mH_C, \alpha_{C \backslash C^O}} B$ for all $B \in C \backslash C^O$. This proves our lemma for $A \in C^O$.

For $A \in C \backslash C^O$, we prove the lemma by showing that for any $i \ge 0$, $\{B | A <_{\mH, \alpha}^i B\} = \{B | A <^i_{\mH_C, \alpha_{C \backslash C^O}} B\}$. Note that our earlier result shows that both of these sets are subsets of $C \backslash C^O$, so we know that for any $i$, any $B$ in either set appears in both aggregation orderings.

Proof by induction on $i$. We first consider the base case: $i = 0$. We note that since $A$ is not an output attribute condition $(0.2)$ is irrelevant. Since $\mH_C$ contains all edges involving attributes in $C$ and $\alpha_{C \backslash^i C^O}$ preserves the ordering and the operators of elements of $\alpha$, condition $(0.1)$ applies to the same set of attributes in both $\ajar$ queries $Q_{\mH, \alpha}$ and $Q_{\mH_C, \alpha_{C \backslash C^O}}$. Thus $\{B | A <_{\mH, \alpha}^0 B\} = \{B | A <_{\mH_C, \alpha_{C \backslash C^O}}^0 B\}$.

For $i > 0$, the inductive hypothesis supposes $\{B | A <_{\mH, \alpha}^j B\} = \{B | A <_{\mH_C, \alpha_{C \backslash C^O}}^j B\}$ for all $j < i$. Again since $\mH_C$ contains all edges involving attributes in $C$ and $\alpha_{C \backslash^i C^O}$ preserves the ordering and the operators of elements of $\alpha$, the inductive hypothesis trivially implies conditions $(i.1)$ and $(i.2)$ apply to the same set of attributes.
\end{proof}

\begin{theorem}[Copy of Theorem~\ref{thm:decompisvalid}]
Every decomposable GHD is valid.
\end{theorem}
\begin{proof}
Suppose the $\ajar$ query is $Q_{\mH, \alpha}$. We need to show for any $A,B$ such that $TOP_\mT(A)$ is an ancestor of $TOP_\mT(B)$, $A \not <_{\mH, \alpha} B$. Proof by induction on $|\alpha|$. If $|\alpha|=0$, all GHDs are valid and decomposable. For $|\alpha > 0|$, we note $\mT_0$ ensure the output attributes are above non-output attributes.

If $A$ and $B$ are non-output attributes and $TOP_\mT(A)$ is an ancestor of $TOP_\mT(B)$, then both are in some $\mT_C$. By the inductive hypothesis, $\mT_C$ is valid with respect to \main{\\}$Q_{(\cup_{E \in \mE_C} E, \mE_C), \alpha_{C \backslash C^O}}$. By Lemma~\ref{lemma:sub-orders}, this implies $A \not <_{\mH, \alpha} B$. 
\end{proof}

\begin{theorem}[Copy of Theorem~\ref{thm:decompwidth}]
For every valid GHD $(\mT, \chi)$, there exists a decomposable GHD $(\mT', \chi')$ such that for all node-monotone functions $\gamma$, the $\gamma$-width of $(\mT', \chi')$ is no larger than the $\gamma$-width of $(\mT, \chi)$.
\end{theorem}

In the proof sketch provided in Section~\ref{sec:decomposing}, we claim to have width-preserving transformations of a GHD that can enforce two additional properties, which we now present and name:
\begin{itemize}
\item \emph{TOP-unique}: every node $t \in \mT$ is $TOP_\mT(A)$ for exactly one attribute $A$\\
\item \emph{subtree-connected}: for any node $t \in \mT$ and the subtree $\mT_t$ rooted at $t$, the attributes $\{v \in \mV | TOP_\mT(v) \in \mT_t\}$ form a connected subgraph of $\mH$
\end{itemize}
We first have two lemmas proving the transformations required to enforce these properties are width-preserving. 

\begin{lemma}\label{lemma:decomp-prop1}
Given a valid GHD $(\mT, \chi)$ with $\gamma$-width $w$, we can transform it to be TOP-unique while ensuring $\gamma$-width $\le w$.
\end{lemma}
\begin{proof}
Define a function $TOP^{-1}_\mT: \mT \to 2^\mV$ from nodes to sets of attributes such that $TOP^{-1}_\mT(t) = \{A | TOP_\mT(A) = t\}$.

First we eliminate nodes $t \in \mT$ such that $|TOP^{-1}_\mT(t)| = 0$. We note, by definition, $\chi(t) \subseteq \chi(parent(t))$. This implies that we can simply remove $t$, connecting all of its children to $parent(t)$ without violating any properties of the valid GHD. And the width is trivially preserved.

Now suppose for some node $t \in \mT$, $|TOP^{-1}_\mT(t)| = k > 1$. Let $A_1$ be the attribute in $TOP^{-1}_\mT(t)$ that is earliest in the aggregation ordering. Let $X = \chi(t) \cap \chi(parent(t))$. Then create a new node $t'$ such that $\chi(t') = \{A_1\} \cup X$ and add it to $\mT$ between $t$ and $parent(t)$. All of properties of the valid GHD must still hold, and since the new node contains a subset of the attributes in $t$, the width must be preserved. Note that after adding this node, $|TOP^{-1}_\mT(t)| = k-1$; we can repeat this process until the set is of size $1$.
\end{proof}

\begin{lemma}\label{lemma:decomp-prop2}
Given a valid GHD $(\mT, \chi)$ with $\gamma$-width $w$ that is TOP-unique, we can transform it to be subtree-connected while preserving TOP-unique and $\gamma$-width $\le w$.
\end{lemma}
\begin{proof}
For any node $t \in \mT$, define $\mV_t = \{v \in \mV | TOP_\mT(v) \in \mT_t\}$.

We proceed with a proof by (bottom-up) induction on the tree $\mT$. As our base case, we consider the leaves of $\mT$. Since $l$ does not have any children, $\mV_l$ must contains exactly one attribute, which is trivially a connected subgraph of $\mH$.

Now we consider the subtree $\mT_t$ rooted at some internal node $t$. Let $A$ be the attribute such that $TOP_\mT(A) = t$. Let $c_1, c_2, \dots, c_k$ be the children of $t$. By the inductive hypothesis, the subtrees rooted at these children satisfy all of the desired properties. We note that, by definiton, $\chi(t) \backslash A \subseteq \chi(parent(t))$. For any child $c_i$ such that $\mV_{c_i}$ and $A$ are not connected in $\mH$, we can remove $A$ from $\chi(\mT_{c_i})$ and set $parent(c_i)$ to be $parent(t)$. By doing so for all such children $c_i$, we ensure that $\mV_t$ is a connected subgraph of $\mH$. Since $A$ is not connected to $\mV_{c_i}$, this transformation does not violate the properties of GHDs. Since we are not creating any new ancestral relationships between nodes, the transformation does not violate the properties of valid GHDs. Finally, the $\gamma$-width $\le w$ and TOP-unique properties are preserved trivially.
\end{proof}

We have thus established that we can transform any valid GHD to additionally satisfy TOP-unique and subtree-connected while preserving width. We now show that any valid GHD satisfying the two additional properties is decomposable. Combined with the two lemmas above, this will complete the proof of Theorem~\ref{thm:decompwidth}. Before we dive into the proof, we prove two helpful lemmas.

\begin{lemma}
\label{lemma:inbetween}
Given an AJAR+ query $Q_{\mH, \alpha}$, if $A <_{\mH, \alpha} B$ for $A,B$ with identical operators, there must exist some $C$ with a different operator such that $A <_{\mH, \alpha} C <_{\mH, \alpha} B$. 
\end{lemma}
\begin{proof}
$A <_{\mH, \alpha}^i B$ for some fixed $i$. If $A,B$ have identical operators, the only way $A <_{\mH, \alpha}^i B$ is via rule $(i.2)$, which requires some $C$ and $j,k < i$ such that $A <_{\mH, \alpha}^j C <_{\mH, \alpha}^k B$. If this $C$ has the same operator as $A$ and $B$, we can repeatedly apply this rule until we find some attribute between $A$ and $B$ with a different operator (since both of the rules for $i=0$ only apply to attributes with differing operators). 
\end{proof}

\begin{lemma}
\label{lemma:ancestor}
Given an AJAR+ query $Q_{\mH, \alpha}$ and valid GHD $(\mT< \chi)$. Suppose $A<_{\mH, \alpha} B$, $A$ is not an output attribute, and $TOP_\mT(A)$ is a top node only for $A$. Then, $TOP_\mT(A)$ must be an ancestor of $TOP_\mT(B)$ in any valid GHD.
\end{lemma}
\begin{proof}
Lemma~\ref{Pathing2} implies that there exists a path from $A$ to $B$ such that for every $C$ in the path such that $C \neq A$, $A <_{\mH, \alpha} C$. Let $C_0, C_1, \dots, C_k$ represent the path, where $C_0 = A$ and $C_k = B$. We claim that $TOP_\mT(A)$ is an ancestor of $TOP_\mT(C_i)$ for all $1 \le i \le k$. Proof by induction on $i$. Our base case is for $i=1$. By the definition of a path, $A$ and $C_1$ must appear together in some hyperedge, implying that they appear together in some bag of $\mT$. Both $TOP_\mT(C_1)$ and $TOP_\mT(A)$ must either be equal to or an ancestor of this bag. Since $TOP_\mT(C_1)$ cannot be equal to or an ancestor of $TOP_\mT(A)$, $TOP_\mT(A)$ is an ancestor of $TOP_\mT(C_1)$.

For $i>1$, we note since $C_i$ and $C_{i-1}$ appear in an edge together, by the same logic as above, $TOP_\mT(C_i)$ and $TOP_\mT(C_{i-1})$ must both be equal to or an ancestor of some node $t \in \mT$. By the inductive hypothesis, $TOP_\mT(A)$ is an ancestor of $TOP_\mT(C_{i-1}$, implying that $TOP_\mT(A)$ is an ancestor of $t$. Since $TOP_\mT(C_i)$ cannot equal or be an ancestor of $TOP_\mT(A)$, $TOP_\mT(A)$ must be an ancestor of $TOP_\mT(C_i)$.
\end{proof}

\begin{lemma}
Any valid GHD $(\mT, \chi)$ that is TOP-unique and subtree-connected must be decomposable. 
\end{lemma}
\begin{proof}
We actually prove a slightly stronger statement. Define the property TOP-semiunique as follows: every non-root node $t \in \mT$ is the $TOP_\mT$ node for exactly one attribute and the root node is either the $TOP_\mT$ for exactly one attribute or more than one output attribute (and zero non-output attributes). Note that the TOP-unique property directly implies the TOP-semiunique property. We will show that if $(\mT, \chi)$ is a valid, TOP-semiunique, and subtree-connected GHD for the $\ajar$ query $Q_{\mH, \alpha}$, it must be decomposable.

Proof by induction on $|\alpha|$. If $|\alpha = 0|$, then every GHD is decomposable. 

Suppose $|\alpha > 0|$. Consider the set of nodes that are $TOP_\mT$ nodes for output attributes, i.e. $\{t \in \mT | \exists A \in V(-\alpha): TOP_\mT(A) = t\}$. Since no non-output attribute can have a top node above an output attribute's top node, the TOP-semiunique property guarantees that this set of nodes forms a rooted subtree $\mT_0$ of $\mT$ such that $\chi(\mT_0) = V(-\alpha)$. 

Consider the subtrees in $\mT \backslash \mT_0$. Call them $\mT_1, \mT_2, \dots, \mT_k$. For any $\mT_i$, let $\mV_i$ be the attributes that have $TOP_\mT$ nodes in $\mT_i$, i.e. $\mV_i = \{A \in \mV | TOP_\mT(A) \in \mT_i\}$. None of these $\mV_i$ can contain any output attributes, and connected-subtree guarantees that each of the $\mV_i$ are connected.  Thus, the $\mV_i$ must be the connected components of $\mH \backslash V(-\alpha)$. So for each connected component $C$ of $\mH \backslash V(-\alpha)$, the corresponding subtree $\mT_C$ is the subtree $\mT_i$ such that $\mV_i = C$. Since for any $A \in C$, $TOP_\mT(A) \in \mT_C$, the attributes in $C$ only appear in $\mT_C$. Note that for every edge $E \in \mE$, there exists a node $t \in \mT$ such that $E \subseteq \chi(t)$. This implies that for every edge $E \in \mE_C$, there exists a node $t \in \mT_C$ such that $E \subseteq \chi(t)$. As such, we can conclude that each $\mT_C$ is a GHD for the hypergraph $(\bigcup_{E \in \mE_C} E, \mE_C)$.

Define $\mV_C = \bigcup_{E \in \mE_C} E$. To complete this proof, we now need to show that each $\mT_C$ is a decomposable GHD for the $\ajar$ query $Q_{(\mV_C, \mE_C), \alpha_{C \backslash C^O}}$. By the inductive hypothesis, if $\mT_C$ is valid, TOP-semiunique and subtree-connected, it must be decomposable. Note that since $\mT$ is TOP-semiunique and subtree-connected, $\mT_C$ must also be TOP-semiunique and subtree-connected. We have also established that $\mT_C$ is a GHD for $(\mV_C, \mE_C)$. Thus to finish this proof, we only need to show that for any $A, B \in \mV_C$ such that $TOP_{\mT_C}(B)$ is an ancestor of $TOP_{\mT_C}(A)$, $A \not <_{(\mV_C, \mE_C), \alpha_{C \backslash C^O}} B$.

For ease of notation, in the rest of this proof we will use $<_C$ to represent $<_{(\mV_C, \mE_C), \alpha_{C \backslash C^O}}$. We show the contrapositive: if $A <_C B$, $TOP_{\mT_C}(B)$ is not an ancestor of $TOP_{\mT_C}(A)$. We consider a few cases. If $A \in \mV_C \backslash C$, $A$ must be in $V(-\alpha)$, implying $TOP_{\mT_C}(A)$ is the root of $\mT_C$. For any $A \in C$, note that $TOP_{\mT_C}(A) = TOP_\mT(A)$. By Lemma \ref{lemma:sub-orders}, for any $A \in C$, if $A <_C B$ then either $A <_{\mH, \alpha} B$ or $A \in C^O$. In the former case, the fact that $\mT$ is valid ensures $TOP_\mT(B)$ is not an ancestor of $TOP_\mT(A)$. For the latter case, assume for contradiction that there exist $A, B$ such that $TOP_\mT(B)$ is an ancestor of $TOP_\mT(A)$, $A <_C B$, and $A \in C^O$. 

We first claim that, without loss of generality, we can suppose that $A$ and $B$ have different operators. To do so, we show that if $A$ and $B$ have the same operator, there must exist a $B'$ with a different operator such that $TOP_\mT(B')$ is an ancestor of $TOP_\mT(A)$ and $A <_C B'$. By the definition of $C^O$, there must exist some $A' \in C^O$ such that $A' <_{\mH, \alpha} B$. Since $A', A \in C$, there must exist a path exclusively in $C$ that connects the two. And since $A', A \in C^O$, no attribute along the path precedes either $A$ or $A'$ in $<_{\mH, \alpha}$. The contrapositive of Lemma~\ref{Pathing} implies that $A'$ and $A$ must have the same operator, which implies that $A'$ and $B$ have the same operator. Lemmas~\ref{lemma:inbetween} and~\ref{lemma:ancestor} imply that there exists some $B'$ with a different operator such that $A' <_{\mH, \alpha} B' <_{\mH, \alpha} B$ and $TOP_\mT(B')$ is an ancestor of $TOP_\mT(B)$. The former result implies $B' \in C \backslash C^O$ and $A <_C B'$. The latter result implies $TOP_\mT(B')$ is an ancestor of $TOP_\mT(A)$.

We now suppose $A$ and $B$ have different operators without loss of generality. Since $A \in C^O$, any $O \in \mV$ such that $O <_{\mH, \alpha} A$ must be an output attribute, thereby implying $O <_{\mH, \alpha} B$ as well. This fact, combined with Lemma~\ref{Pathing}, implies every path between $A$ and $B$ must contain some $D$ such that $D <_{\mH, \alpha} B$. Since $\mT$ is valid and TOP-semiunique, the $TOP_\mT(D)$ for each of these $D$ cannot be in the subtree rooted at $TOP_\mT(B)$. This implies that $A$ and $B$ are disconnected in the subtree rooted at $TOP_\mT(B)$, contradicting the subtree-connected property.
\end{proof}

\begin{theorem}[Copy of Theorem~\ref{thm:decomp}]
For an $\ajar$ query $Q_{\mH, \alpha}$, suppose $\mH_0, \dots, \mH_k$ are the characteristic hypergraphs $H(\mH, \alpha)$. Then GHDs $G_0, G_1, \dots, G_k$ of $\mH_0, \dots, \mH_k$ can be connected to form a decomposable GHD $G$ for $Q_{\mH, \alpha}$. Conversely, any decomposable GHD $G$ of $Q_{\mH, \alpha}$ can be partitioned into GHDs $G_0, G_1, \dots, G_k$ of the characteristic hypergraphs $\mH_0, \dots, \mH_k$. Moreover, in both of these cases, $\gamma\text{-width}(G) = \max_{i} \gamma\text{-width}(G_i)$.
\end{theorem}
\begin{proof}
Proof by induction on $|\alpha|$. Our base case is $|\alpha| = 0$. In this case, the only characteristic hypergraph is the input hypergraph, that is $H(\mH, \alpha) = \mH$. The theorem is then trivially true.

Suppose $|\alpha > 0|$. Any decomposable GHD $G$ for $Q_{\mH, \alpha}$ must be decomposable into subtrees $\mT_0, \dots, \mT_l$ such that $\chi(\mT_0) = V(-\alpha)$ and $\mT_i$ is a decomposable GHD for $Q_{(\cup_{E \in \mE_{C_i}} E, \mE_{C_i}), \alpha_{C_i \backslash C_i^O}}$ where $C_i$ is the $i^{th}$ connected component of $H \backslash V(-\alpha)$. Define $\mV_{C_i}$ to be $\cup_{E \in \mE_{C_i}} E$. To preserve the running intersection property of a GHD, the root of $\mT_i$ and its parent (in $\mT_0$) must contain the attributes $\mV_{C_i} \cap V(-\alpha)$. This implies each of the $\mT_i$ are decomposable GHDs of $Q_{\mH_i^+, \alpha_{C_i \backslash C_i^O}}$, where $\mH_i^+$ is the hypergraph defined in the definition of the characteristic hypergraphs. By the inductive hypothesis, the $\mT_i$ (for $i \ge 1$) can be broken down into GHDs $G_1, \dots G_k$ of the characteristic hypergraphs $\mH_1, \dots, \mH_k$. In addition, $\mT_0$ must also have nodes that contain the edge $E \in \mE$ such that $E \subseteq V(-\alpha)$, implying it is the GHD $G_0$ of the characteristic hypergraph $\mH_0$.

In the other direction, by the inductive hypothesis, the GHDs $G_1, \dots G_k$ can be stitched together to form $\mT_1, \dots \mT_l$ such that each $\mT_i$ is the decomposable GHD for the \ajar query $Q_{\mH_i^+, \alpha_{C_i \backslash C_i^O}}$. Note that, by definition, for each $i$, $\mT_i$ and $G_0$ must both have a node containing the attributes $\mV_{C_i} \cap V(-\alpha)$; let $t_i$ and $g_i$ denote the appropriate node in $\mT_i$ and $G_0$, respectively. We can re-root $\mT_i$ at $t_i$ without violating any conditions since it amounts to re-rooting the top-most GHD of its decomposition; re-rooting $\mT_i$ at $t_i$ can only change the ancestor relationship between $TOP_\mT$ nodes of output attributes. Once we re-root the $\mT_i$ appropriately, we can simply set $parent(t_i)$ to be $g_i$ to generate a decomposable GHD for the \ajar query $Q_{\mH, \alpha}$.
\end{proof}

\section{Product Aggregations (Detailed version)}\label{sec:univ-aggregation-proofs}
The primary application of queries with multiple aggregations is to establish bounds for the Quantified Conjunctive Query ($QCQ$) problem~\cite{FAQ}. A $QCQ$ query consists of an arbitrary conjunctive query preceded by a series of (existential and/or universal) quantifiers, and a solution must report the satisfying assignments to the non-quantified variables. A $\#QCQ$ query is similar to a $QCQ$ query, but instead of reporting satisfying assignments, we report the number of satisfying assignments. 

We now introduce a new type of aggregation, called product aggregation, that lets us efficiently handle $QCQ$ queries. We define the \ajar problem for product aggregations, and then extend our algorithm from Section~\ref{subsec:simple-solution} to handle this new type of \ajar query.

\subsection{\ajar queries with product aggregates} 

In order to recover $QCQ$ as an \ajar query, we need {\em product aggregations} i.e. aggregations that use the $\otimes$ operator. Throughout the paper, we have assumed that an absent tuple effectively has an annotation of $0$. To maintain this for product aggregations, we need to define product aggregation so that it returns $0$ if any tuple is absent. In particular, we redefine $\sum_{(A, \otimes)}R_F$ to include a projected tuple $t_{F \backslash A}$ in the output only if $(t_{F\backslash A} \circ t_A)$ exists in $R_F$ for \emph{every} possible value $t_{A} \in \mD^A$. More formally, let $B = F \backslash A$:

\begin{definition}
$\displaystyle \sum_{(A, \otimes)} R_{AB} = \{(t_B, \lambda): \forall t_A \in \mD^A, t_B \circ t_A \in R_{AB} \text{ and } \lambda = \prod_{(t, \lambda_t) \in R_{AB}: \pi_B t = t_B} \lambda_t \}$
\end{definition}

Note that this adjusted definition implies an annotation of $0$ is once again fully equivalent to absence. We can adjust the definition of aggregation orderings (and $\ajar$ queries) to possibly include this new type of aggregation. We can construct valid GHDs for such aggregations as before, and run AggroGHDJoin to solve them. 

\begin{example}
Consider the semiring $(\{0,1\}, \max, \cdot)$. Note that in this domain $\max$ is equivalent to a disjunction (and the logical existential quantifier) and $\prod$ is equivalent to a conjunction (and the logical universal quantifier). Thus the space of \ajar queries that use these two aggregators recover all $QCQ$ queries.
\end{example}

An aggregation using $\otimes$ is called a {\em product aggregation}, and an attribute that is aggregated using a product aggregation is called a {\em product attribute}. Aggregations that are not product aggregations are called {\em semiring aggregations}, while attributes that are neither output attributes nor product attributes are called {\em semiring attributes}. 

{\bf Idempotence Assumption:} Using the product aggregation as defined raises one issue. Our semiring aggregates satisfy the distributive property, which is integral in our ability to push-down aggregations and for our results about commuting aggregations (Theorem~\ref{commute}). In general, product aggregations do not distribute: $(a \otimes b) \otimes (a \otimes c) = (a \otimes a) \otimes (b \otimes c) \neq a \otimes (b \otimes c)$. However if we require our product aggregations to be idempotent, that is that $a \otimes a = a$ for any element $a$, our product aggregations will distribute. And for $QCQ$, the domain is restricted to $\{0,1\}$, in which product aggregations are idempotent. So in this section, we will study idempotent product aggregations; we will generalize to non-idempotent aggregations in Appendix~\ref{subsec:non-idempotent-product}. 

\subsection{Solving \ajar queries with product aggregates}
For aggregation orderings that have product aggregations, the rules for determining when two orderings are equivalent are somewhat different. We now discuss how we can optimize this new type of aggregation further; product aggregations are fundamentally different from ordinary aggregation because we can do the aggregation \emph{before} the join, as seen in the following example:
\begin{example}
In the semiring $(\{0,1\}, \max, \cdot)$, suppose we have two relations $R(A,B) = \{((0,0), x), ((0,1), y)\}$ and $S(B,C) = \{((0, 1), p), ((1,1),q)\}$. Consider the \ajar query $\sum_{(B, \cdot)} R(A,B) \Join S(B,C)$. If compute the join, we will get two tuples with the annotations $x \cdot p$ and $y \cdot q$, and then aggregating over $B$ will produce a relation with the element $((0,1), x \cdot p \cdot y \cdot q)$. However, note that $x \cdot p \cdot y \cdot q = (x \cdot y) \cdot (p \cdot q)$, implying that $\sum_{(B, \cdot)} R(A,B) \Join S(B,C) = (\sum_{(B, \cdot)} R(A,B)) \Join (\sum_{(B, \cdot)} S(B,C))$.
\end{example}

Now we describe our algorithm for solving \ajar queries when product aggregations are present. Our algorithm follows the same lines as the algorithm from Section~\ref{subsec:simple-solution}. Recall that the algorithm consisted of searching for {\em equivalent orderings}, then searching for GHD {\em compatible} with an equivalent ordering, and running AggroGHDJoin on the GHD with the smallest fhw. For product aggregations, we need to modify our algorithm for testing equivalent orderings, and our definition of compatibility; we do these in turn.

\paragraph*{Testing orderings for equivalence} Algorithm~\ref{algo:equivalance-test-prod} gives the pseudo-code for our equivalence test for orderings containing product aggregates.

\begin{algorithm}[t]
\caption{TestEquivalence($\mH = (\mV_\mH, \mE_\mH)$, $\alpha$, $\beta$)}
\label{algo:equivalance-test-prod}
\textbf{Input:} Query hypergraph $\mH$, orderings $\alpha$, $\beta$.\\
\textbf{Output:} True if $\alpha \equiv_{\mH} \beta$, False otherwise.
\begin{algorithmic}
\If{$|\alpha| = |\beta| = 0$}
\State \Return True
\EndIf
\State Remove $V(-\alpha)$ from $\mH$, then divide $\mH$ into connected components $C_1,\ldots C_m$.
\If{$m > 1$}
\State \Return $\land_{i} \text{TestEquivalence}(\mH, \alpha_{C_i}, \beta_{C_i})$
\EndIf
\State Choose $j$ such that $\beta_j = \alpha_1$. Let $\beta_j = (b_j, \oplus'_j)$.
\If{$\exists i < j : \beta_i = (b_i, \oplus'_i), \oplus'_i \neq \oplus'_j$ and there is a path from $b_i$ to $b_j$ in $\{b_i,b_{i+1},\ldots,b_{|\alpha|}\}$}
\State \Return False
\EndIf
\State Let $\beta'$ be $\beta$ with $\beta_j$ removed.
\State Let $\alpha'$ be $\alpha$ with $\alpha_1$ removed.
\State \Return $\text{TestEquivalence}(\mH, \alpha', \beta')$
\end{algorithmic}
\end{algorithm}

We have the lemma analogous to Lemma~\ref{lemma:equivalence-test-sound-complete}.

\begin{lemma}[Copy of Lemma~\ref{lemma:AlgoP-sound-complete}]
Algorithm~\ref{algo:equivalance-test-prod} returns True if and only if $\alpha \equiv_\mH \beta$.
\end{lemma}
\begin{proof}
{\bf Soundness:}
Suppose Algorithm~\ref{algo:equivalance-test-prod} returns true; we will show $\alpha \equiv_\mH \beta$. We induct on the length $|\alpha|$. For our base case, when $|\alpha|=0$, we return true when $|\beta|=0$. In this case, the two (empty) orderings are trivially equivalent.

Suppose $|\alpha| > 0$. We have two cases: when $\mH \backslash (V(-\alpha) \cup \pa(\alpha))$ has one component and when it has multiple components. We first consider the multiple components case. Let the components be $C_1, \dots, C_m$. Then we define $C'_1, \dots, C'_m$ as in the algorithm i.e. For $1 < i < m$, let $\mE_i$ be $\{E \in \mE | E \cap C_i \neq \emptyset\}$, the elements of $\mE$ that intersect with $C_i$. Then $C'_i = C_i \cup \bigcup_{E \in \mE_i} E \cap \pa(\alpha)$. We define $\mE_0$ to be $\mE \backslash (\bigcup_{1 \le i m} \mE_i)$ (these are relations with only output attributes or product aggregations). Accordingly let $C'_0$ be the product aggregations that appear in $\mE_0$. We can then express the following identities:
\begin{align*}
\Join_{F \in \mE} R_F &= \Join_{0 \le i \le m} \Join_{F \in \mE_i} R_F\\
\sum_{\alpha} \Join_{F \in \mE} R_F &= \Join_{0 \le i \le m} \sum_{\alpha_{C'_i}} \Join_{F \in \mE_i} R_F
\end{align*}
The RHS may have a product aggregation $(a,\otimes)$ happening in multiple components, but it happens exactly once per relation containing $a$.
We note this identity holds for $\beta$ as well. This identity implies that $\alpha \equiv_\mH \beta$ if $\alpha_{C'_i} \equiv_\mH \beta_{C'_i}$ for all $i$. We note that for $i=0$, all of the aggregations contain the same operator, so any ordering is equivalent. For $i>0$, we note that we return true only if all of the recursive calls return true, implying $\alpha_{C'_i} \equiv_\mH \beta_{C'_i}$ by the inductive hypothesis.

When $\mH \backslash (V(-\alpha) \cup \pa(\alpha))$ has one component, we choose $j$ such that $\beta_j = \alpha_1$ and define $\beta'$ to be $\beta$ with $\beta_j$ removed. Note $\alpha'$ is defined to be $\alpha$ with $\alpha_1 = \beta_j$ removed. To show $\beta \equiv_\mH \alpha$, we need to show $\alpha' \equiv_\mH \beta'$ and $\beta \equiv_\mH \beta_j \beta'$. Since we return true only when our recursive call on $\alpha'$ and $\beta'$ returns true, the former equivalence holds by the inductive hypothesis.

To show $\beta \equiv_\mH \beta_j \beta'$, we ensure $\beta_j$ and $\beta_i$ can commute for all $i<j$. More specifically, we ensure that if $\beta_j$ can be moved to index $i+1$, it can be moved to index $i$. For any $\beta_i$ with the same operator, $\beta_i$ and $\beta_j$ trivially commute. If $\beta_i$ has a different operator, we know there is no path between their attributes $b_i$ and $b_j$ among the nodes
$$(\{b_i,b_{i+1},\ldots,b_{|\alpha|}\} \setminus \pa(\alpha)) \cup \{b_i,b_j\}.$$ 
Let $V$ be this set of attributes. Define $V_1 \subset V$ to be the set of nodes connected to $b_i$ in the hypergraph restricted to $V$ (we know $b_j \notin V_1$). Let $\mE_1$ be the set of edges that contain some attribute in $V_1$, i.e. $\{ E \in \mE | E \cap V_1 \neq \emptyset\}$. We note that the attributes of $V \backslash V_1$ do not appear in the edges of $\mE_1$. Let $\mE_2 = \mE \backslash \mE_1$; the attributes of $V \backslash V_1$ all appear in $\mE_2$. We can then express the following identities:
\begin{align*}
\Join_{F \in \mE} R_F = &(\Join_{F \in \mE_1} R_F) \Join (\Join_{F \in \mE_2} R_F)\\
\sum_{\beta_{V \cup \pa(\alpha)}} \Join_{F \in \mE} R_F = &\left(\sum_{\beta_{V_1 \cup \pa(\alpha)}} \Join_{F \in \mE_1} R_F \right) \Join \\
&\left(\sum_{\beta_{V_2 \cup \pa(\alpha)}} \Join_{F \in \mE_2} R_F\right)
\end{align*}
We note, by definition, that $\beta_i$ and $\beta_j$ must be pushed down into different aggregations in the previous expression. This implies that we can commute $\beta_i$ and $\beta_j$ when they are adjacent, completing the soundness proof.

{\bf Completeness:} We prove that if Algorithm~\ref{algo:equivalance-test-prod} returns false, then there must exist a database instance $I$ such that $Q_{\mH, \alpha}(I) \neq Q_{\mH, \beta}(I)$. 

If Algorithm~\ref{algo:equivalance-test-prod} returns false, there must be a component $C'$, $\alpha' = \alpha_{C'}$, $\beta' = \beta_{C'}$, such that $\beta_j = \alpha_1$, and there exists a $i < j$ such that $\beta_i = (b_i, \odot'_i)$, $\beta_j = (b_j, \odot'_j)$, $\odot'_i \neq \odot'_j$ and there is a path from $b_i$ to $b_j$ that consists of only $b_i$, $b_j$, and semiring attributes in $\{b_i, b_{i+1},\ldots,b_{|\alpha'|}\}$. We now define our instance $I$ that gives different outputs on these orderings. 

If neither $\odot'_i$ nor $\odot'_j$ are product operators, then choose $x$, $y$ such that $x \odot'_i y \neq x \odot'_j y$. If one of them is a product operator while the other is not, choose $x = y = 1$. Now we define the attribute domains. Let $B$ be the set of attributes in the path from $b_i$ to $b_j$ consisting of $b_i$, $b_j$ and semiring attributes in $\{b_i, b_{i+1},\ldots,b_{|\alpha'|}\}$. For every $b \in B$, we set $\mD^b = \{0,1\}$. For every $b' \notin B$, we set its $\mD^{b'}$ to $\{0\}$. In every relation that has at least one attribute from $B$, it has two tuples. One tuple has value $0$ for all attributes in $B$, the other has value $1$ for all attributes in $B$. The values of the other attributes are of course always $0$. One of the relations containing a attribute from $B$ has annotation $x$ for the tuple with $0$s and annotation $y$ for the tuple with $1$s. All other annotations are $1$. 

Clearly, each aggregation for an attribute $b' \notin B$ is a no-op, since the domain size $|\mD^{b'}| = 1$. Moreover, all aggregations other than $\beta_i$, $\beta_j$ in $\beta$ and $\alpha$ are also no-ops, because they are non-product aggregations (from the way we chose $B$) and there is a unique value of the attribute for each tuple it maps to after aggregation. 

Thus if both $\beta_i$ and $\beta_j$ are non-product aggregations themselves, then we have
$Q_{\mH, \alpha}(I) = x \odot'_j y$, $Q_{\mH, \beta}(I) = x \odot'_i y$ which are unequal due to how we chose $x$ and $y$. If one of them, say $\beta_j$ is a product aggregation, then $Q_{\mH, \alpha}(I) = 1$ while $Q_{\mH, \beta}(I) = 0$ (and vice versa if $\beta_i$ is a product aggregation). This is because in $\beta$, when we do the product aggregation $\beta_j$, there is only one value of $b_j$ per corresponding output value, so the product annotation is $0$ (and finally the $\beta_i$ aggregation adds two $0$'s to get $0$). On the other hand, for $\alpha$, $\beta_j = \alpha_1$ happens when $b_j$ has two values $0$, $1$ corresponding to a single output tuple, so their annotations are multiplied to get $x \otimes y = 1$. This shows that Algorithm~\ref{algo:equivalance-test-prod} is complete.
\end{proof}

\paragraph*{Compatible GHDs} Product aggregations not only change the set of equivalent orderings, but also the set of GHDs compatible with a given ordering. In fact, product aggregations allow us to break the rules of GHDs without causing incorrect behavior. In particular, we can have a product attribute $P$ appear in completely disparate parts of the GHD. Thus before defining compatibility for GHDs, we define the notion of {\em product partitions}. 

\begin{definition}
Given a hypergraph $\mH = (\mV, \mE)$ and aggregation ordering $\alpha$, let $S = \{a \in \mV | (a, \otimes) \in \alpha\}$ be the set of attributes with product aggregations. A \emph{product partition} is a set $\{P_a | a \in S\}$ where $P_a$ is a partition of $\{F \in \mE | a \in F\}$ (the relations that contain $a$).
\end{definition}

We will duplicate each attribute $a$ for each partition of $P_a$ and have the partition specify which edges contain each instance of $a$.

\begin{definition}
Suppose we are given a hypergraph $\mH = (\mV, \mE)$, aggregation ordering $\alpha$, and product partition $P$. The \emph{product partition hypergraph} $\mH_P$ is the pair $(\mV_P, \mE_P)$ such that
\begin{itemize}
\item $S = \{a \in \mV | P_a \in P\}$
\item $\mV_P = \left(\bigcup_{a \in S} \{a_1, a_2, \dots, a_{|P_a|}\} \right) \cup \mV \backslash S$
\item $p: \mV \times \mE \to \mV_P \text{ where } p(a, F) = a \text{ if } a \notin S \text{ otherwise }$\\$a_i \text{ where } F \text{ is in } i^{th} \text{ partition of } P_a$
\item $\mE_P = \bigcup_{F \in \mE} \{ p(a, F) | a \in F \}$
\end{itemize}
\end{definition}
\begin{definition}
Given a hypergraph $\mH$ and aggregation ordering $\alpha$, an \emph{aggregating generalized hypertree decomposition} (AGHD) is a triple $(\mT, \chi, P)$ such that $(\mT, \chi)$ is a GHD of the product partition hypergraph $\mH_P$.
\end{definition}

For any attribute $a$ in the \ajar query, $TOP_\mT(a)$ for an AGHD $(\mT, \chi, P)$ can be defined as the set $\{TOP_\mT(a_1)$, $TOP_\mT(a_2)$,$\ldots$,$TOP_\mT(a_{|P_a|})\}$. Now we can define the notion of compatibility of an AGHD, with an ordering. 
\begin{definition}
A AGHD $(\mT, \chi, P)$ for an $\ajar$ query $Q_{\mH, \alpha}$ is compatible with an ordering $\beta \equiv_\mH \alpha$ if for each attribute pair $a$, $b$ for which there exists $v_1 \in TOP_\mT(a)$, $v_2 \in TOP_\mT(b)$ such that $v_1$ is an ancestor of $v_2$, $a$ must occur before $b$ in the ordering $\beta$.
\end{definition}

\paragraph*{Solving \ajar queries with product aggregates}
In our proofs and discussions for the remainder of this section, we will treat the set of $TOP_\mT$ as a single element for convenience, implicitly placing an existential quantifier before the statement. For example, when we say $TOP_\mT(A)$ is an ancestor of $TOP_\mT(B)$, we mean $\exists t_A \in TOP_\mT(A), t_B \in TOP_\mT(B)$ such that $t_A$ is an ancestor of $t_B$. We also often omit the partition $P$ when referring to an AGHD $G = (\mT, \chi, P)$; the partition $P$ can be uniquely defined by $(\mT,\chi)$, so we will always assume it is defined appropriately. 

We can now modify our algorithm from Section~\ref{subsec:simple-solution} to detect equivalent orderings using Algorithm~\ref{algo:equivalance-test-prod}, then search for compatible AGHDs, and run AggroGHDJoin over the compatible AGHD with the smallest fhw. Our runtime is given by the next theorem. Note that any AGHD of the original hypergraph is also a GHD of some product partition hypergraph. 

\begin{theorem}[Copy of Theorem~\ref{thm:ajar-product-runtime}]
Given a \ajar query $Q_{\mH,\alpha}$ possibly involving idempotent product aggregates, let $w^*$ be the smallest fhw for an AGHD compatible with an ordering equivalent to $\alpha$. Then the runtime for our algorithm is $\Ot(\IN^{w^*} + \OUT)$.
\end{theorem}

The theorem is proved in Appendix~\ref{sec:background}.

\paragraph*{Decomposing AGHDs} We can apply the ideas from Section~\ref{sec:decomposing} to \ajar queries with product aggregates as well. In this section we will assume without loss of generality that for any relation $R_F$, the last aggregation in $\alpha_F$ is not a product aggregation. Suppose this assumption is violated, i.e. there exists some relation $R_F$ such that the last aggregation in $\alpha_F$ is the product aggregation $(A_P, \otimes)$. We can then immediately perform this aggregation, transforming the relation to $R_{F \backslash \{A_P\}}$ and removing the product aggregation. This assumption ensures that every relation appears in one of the subtrees in the decomposition defined below. We now define some terms.

Given an \ajar query $Q_{\mH, \alpha}$, suppose we have a subset of the nodes $V \in \mV$. Define $\mE_V$ to be $\{E \in \mE | E \cap V \neq \emptyset\}$, i.e. the set of edges that intersect with $V$. Additionally, define $\alpha_{-[i]}$ to be $\alpha$ with the first $i$ elements removed. We will be looking at the connected components of $\mH \backslash (V_{-\alpha} \cup \pa(\alpha))$. For any connected component $C$, let $C^+ = C \cup \{ v \in \pa(\alpha) | \exists E \in \mE_C : v \in E \}$. Additionally, given an ordering $\alpha$, we define $\alpha^O$ based on a conditional: if $\alpha_1$ is a product aggregation, let $\alpha^O$ be just $\alpha_1$; if $\alpha_1$ is not a product aggregations, let $\alpha^O$ be the set of attributes that can be commuted to the beginning of the ordering. To be more precise for this second case, given an attribute $A$ that appears in $\alpha_j$ with operator $\odot$, $A \in \alpha^O$ if for all $\alpha_i = (B, \odot')$ such that $i<j$ either $\odot' = \odot$ or $A$ and $B$ are not connected among the nodes $(\alpha_{-[i-1]} \backslash \pa(\alpha_{-[i-1]}) \cup \{A, B\}$.

\begin{definition}\label{def:decomp2}
Given an \ajar query $Q_{\mH, \alpha}$, we say an AGHD $(\mT, \chi, P)$ is \emph{decomposable} if:
\begin{itemize}
\item There exists a rooted subtree $\mT_0$ of $\mT$ such that $\chi(\mT_0) = \mV(-\alpha)$ (i.e. output attributes).
\item For each connected component $C$ of $\mH \backslash (V_{-\alpha} \cup \pa(\alpha))$, there is exactly one subtree $\mT_{C} \in \mT \backslash \mT_0$ such that $\mT_C$ is a decomposable AGHD of $Q_{(\cup_{E \in \mE_C} E, \mE_C), \alpha_{C^+ \backslash \alpha_{C^+}^O}}$.
\end{itemize}
\end{definition}

Then we have theorems analogous to theorems~\ref{thm:decompisvalid},~\ref{thm:decompwidth}, and~\ref{thm:decomp}.

\begin{theorem}[Copy of Theorem~\ref{thm:decompaghdisvalid}]
All decomposable AGHDs are compatible with an ordering $\beta$ such that $\beta \equiv_\mH \alpha$.
\end{theorem}
\begin{proof}
Suppose we are given an $\ajar$ query $Q_{\mH, \alpha}$ and a decomposable AGHD $G$ for this query. We show a stronger statement: all decomposable AGHDs are compatible with an ordering $\beta$ such that $\beta \equiv_\mH \alpha$ and $\pa(\beta) = \pa(\alpha)$ (i.e. the order of the product attributes does not change). Proof by induction on $|\alpha|$. When $|\alpha| = 0$, all GHDs are decomposable and all GHDs are compatible with $\alpha$.

Suppose $|\alpha| > 0$.  By definition, there is a subtree $\mT_0$ of $G$ such that $\chi(\mT_0) = V(-\alpha)$. And for each connected component $C$ of $\mH \backslash (V(-\alpha) \cup \pa(\alpha))$, we have a subtree $\mT_C$ that is a decomposable GHD for the query $Q_{(\cup_{E \in \mE_C} E, \mE_C), \alpha_{C^+ \backslash \alpha_{C^+}^O}}$. We will use $\mV_C$ to denote $\cup_{E \in \mE_C} E$ and $\mH_C$ to denote $(\mV_C, \mE_C)$. Similarly, we will use $\alpha^C$ to represent $\alpha_{C^+ \backslash \alpha_{C^+}^O}$. By the inductive hypothesis, each of these subtrees $\mT_C$ is compatible with some ordering $\beta^C$ such that $\beta^C \equiv_{\mH_C} \alpha^C$ and $\pa(\beta^C) = \pa(\alpha^C)$. Note that $\beta^C \equiv_{\mH_C} \alpha^C$ trivially implies $\beta^C \equiv_\mH \alpha^C$.

For each $C$ we will construct a $\beta^C+$ such $\mT_C$ is compatible with $\beta^C+$, $\beta^C+ \equiv_\mH \alpha_{C^+}$, and $\pa(\beta^C+) = \pa(\alpha_{C^+})$. Since $\alpha^C = \alpha_{C^+ \backslash alpha_{C^+}^O}$, this requires adding the elements of $\alpha_{C^+}^O$ to $\beta^C$. Define $\beta^O$ to be some ordering of the elements compatible with $G$ (i.e. for any $A,B \in V(\alpha_{C^+}^O)$ if $TOP_{\mT_C}(A)$ is an ancestor of $TOP_{\mT_C}(B)$, $A$ precedes $B$ in $\beta^O$). We claim the ordering $\beta^C+ = \beta^O \circ \beta^C$ satisfies our three conditions. 

The first condition is that $\mT_C$ is compatible with this $\beta^C+$. This is trivially true because we constructed the ordering by adding output attributes to the start of $\beta^C$, with which $\mT_C$ is already compatible, in an order that is guaranteed to be compatible.

The second condition is that $\beta^C+ \equiv_\mH \alpha_{C^+}$. By the definition of $\alpha_{C^+}^O$, $\alpha_{C^+} \equiv_\mH \alpha_{C^+}^O \circ \alpha^C$. We know $\beta^C \equiv_{\mH} \alpha^C$ by the inductive hypothesis. And we claim $\beta^O \equiv_\mH \alpha_{C^+}^O$, which implies $\beta^C+ \equiv_\mH \alpha_{C^+}$ by definition. We show this claim by showing that the operators of $\alpha_{C^+}^O$ are uniform, implying that its elements can be reordered freely. In particular, consider the first element $(A_1, \odot_1)$ of $\alpha_{C^+}$. Since $C^+$ is a connected component, there must exist a path between $A_1$ and every other node among the nodes $C^+$. Thus, for any $(B, \odot') \in \alpha_{C^+}$ such that $\odot' \neq \odot_1$, $A_1$ will violate the path condition for commuting and ensure $B \notin V(\alpha_{C^+}^O)$.

The third condition is that $\pa(\beta^C+) = \pa(\alpha_{C^+})$. By the inductive hypothesis, $\pa(\beta^C) = \pa(\alpha^C)$. We simply need to show $\pa(\beta_{C^O}) = \pa(\alpha_{C^+}^O)$. There are two cases to consider, from the definition of $\alpha_{C^+}^O$. In the first case, both $\beta_{C^O})$ and $\alpha_{C^+}^O$ contains only one (product) aggregation. In the second case, the two orderings have no product aggregations. In either case, $\pa(\beta_{C^O}) = \pa(\alpha_{C^+}^O)$ trivially.

We now need to combine the $\beta^C+$ for each $C$ to construct the desired ordering $\beta$ as desired. We construct $\beta$ by repeating the two following steps algorithm until every $\beta^C$ is empty: $(1)$ remove the non-product output prefixes of $\beta^C+$ and append them to $\beta$ (interleaved arbitrarily) and $(2)$ remove the earliest remaining product aggregation of $\pa(\alpha)$ from the start of the appropriate $\beta^C+$ and append it to $\beta$. Note that this procedure ensures $\beta_{C^+} = \beta^C+$ for each $C$, which implies $\beta_{C^+} \equiv_\mH \alpha_{C^+}$ and (by the soundness of Algorithm~\ref{algo:equivalance-test-prod}) $\beta \equiv_\mH \alpha$. Also note that the procedure preserves the ordering of the product aggregates, so $\pa(\beta) = \pa(\alpha)$. Finally, the given AGHD $G$ must be compatible with $\beta$. The construction of $G$ ensure the top nodes of output attributes are all above the top nodes of non-output attributes, and the top nodes of non-output attributes are in the subtrees $\mT_C$, which means the fact that $\beta_{C^+} = \beta^C+$ ensures these top nodes are ordered in a compatible manner.
\end{proof}

\begin{theorem}[Copy of Theorem~\ref{thm:decompaghdwidth}]
For every valid AGHD $(\mT, \chi)$, there exists a decomposable $(\mT', \chi')$ such that for all node-monotone functions $\gamma$, the $\gamma$-width of $(\mT', \chi')$ is no larger than the $\gamma$-width of $(\mT, \chi)$.
\end{theorem}
\begin{proof}
We first modify the definition of subtree-connected from Appendix~\ref{sec:app-decomp}:
\begin{itemize}
\item \emph{subtree-connected}: for any node $t \in \mT$ and the subtree $\mT_t$ rooted at $t$, consider the set the attributes $V_t = \{v \in \mV | TOP_\mT(v) \in \mT_t\}$; we require for any two attributes $A,B \in V_t$, there exists a path from $A$ to $B$ in the set $(V_t \backslash \pa(\alpha)) \cup \{A,B\}$.
\end{itemize}

This same transformation described Lemma~\ref{lemma:decomp-prop2} can be used for this adjusted definition. Note that this transformation ensures that any node that is $TOP_\mT$ for a product aggregation has only one child. Also note that the described transformation might change the partition function $P$ of the AGHD, but it does not change the compatible order.

Suppose the given \ajar problem is $Q_{\mH, \alpha}$. Since $(\mT, \chi)$ is valid, there must exist an ordering $\beta$ such that $(\mT, \chi)$ is compatible with $\beta$ and $alpha \equiv_\mH \beta$. The width-preserving transformations of Appendix~\ref{sec:app-decomp} preserve the compatibility with an ordering. So we can apply them to get a TOP-unique and subtree-connected AGHD $(\mT', \chi')$ that is compatible with $\beta$ and has $\gamma$-width no larger than that of $(\mT, \chi)$. We claim that this AGHD is decomposable.

As in Appendix~\ref{sec:app-decomp}, we prove that any valid, TOP-semiunique, and subtree-connected GHD for an is decomposable. Proof by induction on $|\alpha|$. If $|\alpha = 0|$, then every GHD is decomposable.

Suppose $|\alpha| > 0$. Consider the set of nodes that are $TOP_\mT$ nodes for output attributes, i.e. $\{t \in \mT | \exists A \in V(-\alpha): TOP_\mT(A) = t\}$. Since $(\mT', \chi')$ is compatible with $\beta$, no non-output attributes can have a top node above an output attributes top node. Thus, the TOP-semiunique property guarantees that this set of nodes forms a rooted subtree $\mT_0$ of $\mT$ such that $\chi(\mT_0) = V(-\alpha)$.

Consider the subtrees in $\mT \backslash \mT_0$. Call them $\mT_1, \mT_2, \dots, \mT_k$. For any $\mT_i$, let $\mV_i$ be the attributes that have $TOP_\mT$ nodes in $\mT_i$, i.e. $\mV_i = \{A \in \mV | TOP_\mT(A) \in \mT_i\}$. None of these $\mV_i$ can contain any output attributes, and connected-subtree guarantees that each of the $\mV_i$ are connected.  Thus, the $\mV_i$ must be the $C_i^+$ as defined earlier. So for each connected component $C$ of $\mH \backslash (V(-\alpha) \cup \pa(\alpha))$, the corresponding subtree $\mT_C$ is the subtree $\mT_i$ such that $\mV_i = C^+$. Since for any $A \in C$, $TOP_\mT(A) \in \mT_C$, the attributes in $C$ only appear in $\mT_C$. Note that for every edge $E \in \mE$, there exists a node $t \in \mT$ such that $E \subseteq \chi(t)$. This implies that for every edge $E \in \mE_C$, there exists a node $t \in \mT_C$ such that $E \subseteq \chi(t)$. As such, we can conclude that each $\mT_C$ is a GHD for the hypergraph $(\bigcup_{E \in \mE_C} E, \mE_C)$.

Define $\mV_C = \bigcup_{E \in \mE_C} E$. To complete this proof, we now need to show that each $\mT_C$ is a decomposable GHD for the $\ajar$ query $Q_{(\mV_C, \mE_C), \alpha_{C^+ \backslash \alpha_{C^+}^O}}$. By the inductive hypothesis, if $\mT_C$ is valid, TOP-semiunique and subtree-connected, it must be decomposable. Note that since $\mT$ is TOP-semiunique and subtree-connected, $\mT_C$ must also be TOP-semiunique and subtree-connected. We have also established that $\mT_C$ is a GHD for $(\mV_C, \mE_C)$. Thus to finish this proof, we only need to show that there exista an ordering $\beta'$ such that $\beta' \equiv_{(\mV_C, \mE_C)} \alpha_{C^+ \backslash \alpha_{C^+}^O}$ and $\mT_C$ is compatible with $\beta'$.

We know $\mT$ is compatible with $\beta$ and $\beta \equiv \alpha$. We set $\beta' = \beta_{C^+ \backslash \beta_{C^+}^O}$; this implies that $\beta' \equiv_\mH \alpha_{C^+ \backslash \alpha_{C^+}^O}$ since left hand and right hand sides are simply sub-orderings of $beta$ and $\alpha$, respectively. Furthermore, this implies $\beta' \equiv_{(\mV_C, \mE_C)} \alpha_{C^+ \backslash \alpha_{C^+}^O}$, as $(\mV_C, \mE_C)$ is simply $\mH$ with some output attributes (of $\beta'$) removed.

We now need to show that $\mT_C$ is compatible with $\beta'$. In other words, we need to show for any two attributes $A,B \in \mV_C$, if $TOP_{\mT_C}(A)$ is an ancestor of $TOP_{\mT_C}(B)$, either $A$ is an output attribute or $A$ precedes $B$ in $\beta'$. We show the contrapositive: if $A$ is not an output attribute and $A$ does not precede $B$ in $\beta'$, then $TOP_{\mT_C}(A)$ is not an ancestor of $TOP_{\mT_C}(B)$. There are a couple of cases to consider. If $B \in \mV_C \backslash C^+$, $B$ must be in $V(-\alpha)$, implying $TOP_{\mT_C}(B)$ is the root of $\mT_C$.  We note that for attributes in $C^+$, $TOP_{\mT_C}$ and $TOP_\mT$ are equivalent, so we use them interchangeably. If $B \in C^+ \backslash \beta_{C^+}^O$, then we know $B$ must precede $A$ in $\beta'$, which implies $B$ precedes $A$ in $\beta$. The fact that $\mT$ is compatible with $B$ implies $TOP_\mT(A)$ is not an ancestor of $TOP_\mT(B)$. The final case to consider is $B \in \beta_{C^+}^O$.

Even in this case, we have two cases to consider, based on the two definitions of $\beta_{C^+}^O$. If $B$ has a product aggregation, then $B$ must be the first element of $\beta_{C^+}$. This implies $B$ precedes $A$ in $\beta$, guaranteeing that $TOP_\mT(A)$ is not an ancestor of $TOP_\mT(B)$. The other case is a bit more involved.

Assume for contradiction that there exist $A,B$ such that $TOP_\mT(A)$ is an ancestor of $TOP_\mT(B)$, $B \in \beta_{C^+}^O$, and $A \in \beta'$. We first claim that, without loss of generality, we can suppose that $A$ and $B$ have different operators. To do so, we show that if $A$ and $B$ have the same operator, there must exist a $A' \in \beta'$ with a different operator such that $TOP_\mT(A')$ is an ancestor $TOP_\mT(B)$. The fact that $A \notin \beta_{C^+}^O$ implies there is an attribute $A'$ with a different operator such that there exists a path between $A'$ and $A$ composed of attributes that appear after $A'$ in $\beta_{C^+}$. We claim $TOP_\mT(A')$ is an ancestor of $TOP_\mT(A)$, which implies $TOP_\mT(A')$ is an ancestor of $TOP_\mT(B)$. Suppose the path between $A'$ and $A$ is $X_0, X_1, X_2, \dots, X_k$ where $A' = X_0$ and $A = X_k$; we will show $TOP_\mT(A')$ is an ancestor of $TOP_\mT(A)$ by showing $TOP_\mT(A')$ is an ancestor of all $X_i$ for $i \ge 1$. Proof by induction on $i$. For $i=1$, $A'$ and $X_1$ share an edge, implying they appear in $\chi(t)$ together for some tree node $t$. By definition, $TOP_\mT(A')$ and $TOP_\mT(X_1)$ are both ancestors of $t$. Since $\mT$ is TOP-semiunique (so $TOP_\mT(A') \neq TOP_\mT(X_1)$) and $\mT$ is compatible with $\beta$ (so $TOP_\mT(X_1)$ cannot be an ancestor of $TOP_\mT(A')$, this means that $TOP_\mT(A')$ is an ancestor of $TOP_\mT(X_1)$. For $i > 1$, we know that $X_{i-1}$ and $X_i$ share an edges, implying they appear together in $\chi(t)$ for some tree node $t$. $TOP_\mT(X_i)$ and $TOP_\mT(X_{i-1})$ must both ancestors $t$. Note that the inductive hypothesis gives us that $TOP_\mT(A')$ is an ancestor of $TOP_\mT(X_{i-1})$, implying it is an ancestor of $t$. By the same logic as before, this implies that $TOP_\mT(A)$ is an ancestor of $TOP_\mT(X_i)$. We thus have that $TOP_\mT(A')$ is an ancestor of $TOP_\mT(A)$.

We now suppose, without loss of generality, that $A$ and $B$ have different operators. Since $TOP_\mT(A)$ is an ancestor of $TOP_\mT(B)$, we know $A$ comes before $B$ in the compatible ordering $\beta$. However, the fact that $B \in \beta_{C^+}^O$ implies that every path between $B$ and $A$ includes an attribute $X$ that is either an output attribute or comes before $B$ in $\beta$. Either way, none of these $X$ is in the subtree rooted at $TOP_\mT(A)$, implying that $A$ and $B$ are disconnected in the subtree rooted at $TOP_\mT(A)$. This contradicts the subtree-connected property.
\end{proof}

\begin{definition}
Given an $\ajar$ problem $Q_{\mH, \alpha}$, suppose $C_1, \dots, C_k$ are the connected components of $\mH \setminus (\mV_{-\alpha} \cup \pa(\alpha))$. Define a function $H(\mH, \alpha)$ that maps $\ajar$ queries to a set of hypergraphs as follows: 
\begin{itemize}
\item $C_i^{++} = \bigcup_{E \in \mE_C} E$ for all $1 \le i \le k$
\item $\mH_0 = (\mV_{-\alpha}, \{F \in \mE | F \subseteq \mV_{-\alpha}\} \cup \{\mV_{-\alpha} \cap C_i^{++} | 1 \le i \le k\})$
\item $\mH_i^+ = (C_i^{++}, \mE_C \cup \{\mV_{-\alpha} \cap C_i^+\})$
\item $H(\mH, \alpha) = \{\mH_0\} \cup \bigcup_{1 \le i \le k} H(\mH_i^+, \alpha_{C_i^+ \backslash \alpha_{C_i^+}^O})$
\end{itemize}
The hypergraphs in the set $H(\mH, \alpha)$ are defined to be the \emph{characteristic hypergraphs}.
\end{definition}
\begin{theorem}[Copy of Theorem~\ref{thm:decompaghd}]
For an $\ajar$ query $Q_{\mH, \alpha}$ involving product aggregates, suppose $\mH_0, \dots, \mH_k$ are the characteristic hypergraphs $H(\mH, \alpha)$. Then AGHDs $G_0, G_1, \dots, G_k$ of $\mH_0, \dots, \mH_k$ can be connected to form a decomposable AGHD $G$ for $Q_{\mH, \alpha}$. Conversely, any decomposable AGHD $G$ of $Q_{\mH, \alpha}$ can be partitioned into AGHDs $G_0, G_1, \dots, G_k$ of the characteristic hypergraphs $\mH_0, \dots, \mH_k$. Moreover, in both of these cases, $\gamma\text{-width}(G) = \max_{i} \gamma\text{-width}(G_i)$.
\end{theorem}
\begin{proof}
The proof is the exact same as the proof of Theorem~\ref{thm:decomp} provided in Appendix~\ref{sec:app-decomp}.
\end{proof}

This lets us apply all the optimizations from Section~\ref{subsec:optimal-valid},~\ref{subsec:dbp-width}, and~\ref{subsec:GYM} to \ajar queries with product aggregates.

\paragraph*{Comparison to FAQ} The runtime of InsideOut on a query involving idempotent product aggregations is given by $\Ot(\IN^{faqw})$, where the faqw depends on the ordering, and the presence of product aggregations. Our algorithm for handling product aggregations recovers the runtime of FAQ. Formally,

\begin{theorem}
For any \ajar query involving idempotent product aggregations, $\IN^{w^*} + \OUT \leq 2 \cdot \IN^{faqw}$.
\end{theorem}

The proof is in Appendix~\ref{subsec:faq-comparison-proof}. By applying ideas from the FAQ paper to our setting, we can also recover the FAQ runtime on $\#QCQ$ (Appendix~\ref{subsec:recovering-hash-qcq}). Our algorithm for detecting when two orderings involving product aggregates are equivalent (Algorithm~\ref{algo:equivalance-test-prod}) is both sound and complete; in contrast, FAQ's equivalence testing algorithm is sound but not complete. Moreover, we have a width-preserving decomposition for queries with product aggregates. This allows us to apply all the optimizations from Section~\ref{sec:decomposing}, giving us tighter runtimes in terms of submodular and DBP-widths (Theorems~\ref{thm:submodular-width},~\ref{thm:dbp-width}) and efficient MapReduce Algorithms (Theorems~\ref{thm:ajar-mapreduce-n-rounds},~\ref{thm:ajar-mapreduce-log-rounds}). As shown before, FAQ gives a worse runtime exponent in each of these cases.

\subsection{Recovering \#QCQ}\label{subsec:recovering-hash-qcq}
We discussed idempotent product aggregations and how they can help \ajar generalize $QCQ$ in Section~\ref{sec:univ-aggregation}. There is a variant called $\#QCQ$ in which solutions are expected to output the number of solutions to a given $QCQ$ (instead of the solutions themselves). At first this seems like a fairly straightforward extension to $QCQ$. If we use $\ajar$ to solve a given $QCQ$, the output is a relation that lists the satisfying assignments, where each tuple's annotation is $1$; to count the number of tuples, we simply need to prefix the $QCQ$ query with aggregations using the operator $+$. 

An issue arises because these new aggregations need to occur in in the domain $\textbf{Z}_+$ (the non-negative integers) instead of $\{0,1\}$. Though $(\textbf{Z}_+, \max, \cdot)$ is still a semi-ring, the product aggregations are no longer idempotent in the given domain; we discuss how to handle non-idempotent aggregation in Appendix~\ref{subsec:non-idempotent-product}, but the added complexity (and runtime) required to deal with non-idempotent aggregations seems unnecessary in our case. Even though multiplication is not idempotent over the larger domain, we can guarantee that it is idempotent whenever a product aggregation occurs; the annotations do not leave the $\{0,1\}$ domain until the $+$ aggregations, which must occur after the product aggregations.

To handle this extra structure, we introduce the concept of specifying restricted domains in \ajar queries. To recover $\#QCQ$, we translate the approach of FAQ~\cite[Section 9.5]{FAQ}, which is the minimal application of the restricted domain concept to \ajar queries.

\begin{definition}
Given a domain $\mathbb{K}$ and operator set $O$, we define a \emph{restriction} to subsets of the domain $\mathbb{K}_r \subset \mathbb{K}$ and operator set $O_r \subseteq O$ such that $\{0, 1\} \subseteq \mathbb{K}_r, \otimes \in O_r$ and for any $a,b \in \mathbb{K}_r$ and $\odot \in O_r$, $a \odot b \in \mathbb{K}_r$.
\end{definition}

\begin{example}
In the context of $\#QCQ$, $\mathbb{K} = \textbf{Z}_+$ and $O = \{+, \max, \otimes\}$. The restriction is $\mathbb{K}_r = \{0,1\}$ and $O_r = \{\max, \otimes\}$.
\end{example}

Note that if we ensure that the specified operators are closed in the restricted domain, the semiring properties will all hold in the restricted domain. We then define an aggregation ordering that incorporates these restrictions - we will define an index $l$ divides the unrestricted and restricted portions of the ordering. 

\begin{definition}
Given an attribute set $\mV$, domain $\mathbb{K}$, operator set $O$, and restriction $\mathbb{K}_r$ and $O_r$, an \emph{restriction-compatible} aggregation ordering is an aggregation ordering $\alpha$ and index $l$ such that $1 \le l \le |\alpha|$ and for each $k \ge l$, $\alpha_k = (A, \odot)$ for $A \in \mV$ and $\odot \in O_r$.
\end{definition}

Any single operator $\odot$ that appears both before and after the division index $l$ will be treated as different operators (this issue does not come up in the context of $\#QCQ$). We can then define an \ajar query to use a restriction-compatible ordering, and any instance of the query must have $\mathbb{K}_r$-relations. Under this definition, we can treat the product aggregations as idempotent, allowing us to use the work in Section~\ref{sec:univ-aggregation} to recover $\#QCQ$.

This set-up is essentially a translation of FAQ's results to our language/notation. Using the exact same construction described in the previous Appendix section, we can now recover FAQ's runtime on $\#QCQ$ as well. We note that we could extend this idea of restricting domains even further by relying on our GHDs. In particular, we can have every single element of the aggregation ordering specify its own domain, and a valid GHD would have to ensure that for any $A,B$ such that $TOP_\mT(A)$ is an ancestor of $TOP_\mT(B)$, the semiring domain corresponding to $A$ is a superset of the semiring domain corresponding to $B$.

\subsection{Non-Idempotent Product Aggregations}\label{subsec:non-idempotent-product}
Our AggroYannakakis algorithm actually implicitly assumes that any product aggregation that arises consists of an \emph{idempotent} operator.

\begin{definition}
Given a set $S$, an operator $\oplus$ is \emph{idempotent} if and only if for any element $a \in S$, $a \oplus a = a$. 
\end{definition}

This is a reasonable assumption, as the problems that we've discovered using product aggregation all tend to have idempotent products. The key difference between an idempotent and non-idempotent operator is the distributive property; $(a \otimes b) \otimes (a \otimes c) = a \otimes (b \otimes c)$ only if $\otimes$ is idempotent. Note that the non-idempotent case would require an $a^2$. So, to be complete, we can support non-idempotent operators by raising the annotations of every other relation to a power. In particular, if we have a non-idempotent aggregator for an attribute $A$, we should raise the annotations for the relations in every other node in our tree to the $|\mD^A|$ power when we aggregate the attribute $A$ away.

\section{Extension: Computing Transitive Closure}
A standard extension to the basic relational algebra is the transitive closure or Kleene star operator. In this section, we explore how our framework for solving $\ajar$ queries can be applied to computing transitive closures. First we define the operator using the language of $\ajar$. Given a relation $R$ with two attributes, consider the query
$$Q_{k} = \sum_{A_2} \cdots \sum_{A_k} \Join_{1 \le i \le k} R(A_i, A_{i+1})$$
where each of $R(A_i, A_{i+1})$ are identical copies of $R$ with the attributes named as specified. Note that our output $Q_k$ is going to be a two-attribute relation. Suppose there exists some $k^*$ such that $Q_k$ is identical for all $k \ge k^*$. We can then define the transitive closure of a relation $R$, denoted $R^*$, to be $Q_{k^*}$.

This classic operator has natural applications in the context of graphs. If our relation $R$ is a list of (directed) edges (without meaningful annotations), computing $R^*$ is equivalent to computing the connected components of our graph. If we add annotations over the semiring $(\mathbb{Z} \cup \{\infty\}, \min, +)$ where each edge is annotated with a weight, then computing $R^*$ is equivalent to computing all pairs shortest paths~\cite{FunSemirings}. Note that we can guarantee $R^*$ exists as long $(i)$ our graph contains no negative weight cycles and $(ii)$ our relation contains self-edges with weight $0$. We will discuss computing $R^*$ in the context of graphs, applying it to the all pairs shortest path problem. Let $E$ be the number of edges and $V$ the number of nodes in the graph; we will derive the complexity of computing all pairs shortest paths in terms of $E$ and $V$.

A naive algorithm for finding $R^*$ is to compute $Q_{1}, Q_{2}, Q_{4}, \dots$ until we find two consecutive results that are identical. This approach requires answering $O(\log k^*)$ $\ajar$ queries. In the context of all pairs shortest path, we know $k^* \le V$, which means that the number of queries to answer is $O(\log V)$. We start by analyzing the computation required to answer a query of the form $Q_{2^n}$.

We define the GHD to use for $Q_{2^n}$ recursively. Our base case, when $n=1$, is to have a single bag containing all three attributes $A_1, A_2, A_3$. For $n>1$, the root of our GHD will contain the attributes $A_1, A_{2^{n-1}+1}, A_{2^n+1}$. It will have two children: on the left, it will have the GHD corresponding to $Q_{R^{2^{n-1}}}$, and on the right it will have an identical GHD over the attributes $A_{2^{n-1}+1}, A_{2^{n-1}+2}, \dots, A_{2^n+1}$ instead of the attributes $A_1, A_2, \dots, A_{2^{n-1}+1}$. Note that each bag of our constructed GHD has $3$ attributes, but they may not appear in any relation together. Additionally, note that the depth of our GHD is simply $n$.

If we naively apply the AGM bound to derive the fractional hypertree width, we get a width of $E^3$. However, if, for each attribute $A_i$, we (virtually) create a relation $S(A_i)$  of size $V$, our fractional hypertree width becomes $V^3$. Alternatively, we can also use DBP-width to derive the $V^3$ bound without introducing these relations.

Applying the results of GYM~\cite{GYM} gives us that we can answer $Q_{2^n}$ in $O(n)$ MapReduce rounds with $O(V^3)$ communication cost. Given that we need to answer $O(\log V)$ of these queries and that $n \le O(\log V)$ for each of these queries, we have a $O(\log^2 V)$ round MapReduce algorithm with $\Ot(V^3)$ total communication cost for all pairs shortest paths, which is within poly-log factors of standard algorithms for this problem.

In addition, if we allow a $O(k^* \log k^*)$ round MapReduce algorithm, we can reduce the total communication cost to $\Ot(EV)$ by using a chain GHD. In particular, for a query $Q_k$, the GHD will be a chain of $k$ bags such that the $i^{th}$ bag in our chain consists of $A_i$, $A_{i+1}$ and $A_{k+1}$. This construction ensures that two of the three attributes in each bag appear in a relation together, reducing the width to $EV$.

We note that we derived this MapReduce bound with our generic algorithms, without any specialization for this particular problem. We can also derive a serial algorithm for the problem with the same bound, but it requires a small optimization. By construction, our (original, non-chain) GHD has the property that every subtree whose root is at a particular level is completely identical. This means that AggroGHDJoin does not need to visit each bag; it simply needs to visit one bag per level, and then assign the result to the other bags on the level. With this optimization, our algorithm computes all pairs shortest paths in $\Ot(V^3)$, again within poly-log factors of specialized graph algorithms.
\end{document}